\documentclass[reqno,a4paper]{amsart}
\usepackage{amsthm}
\usepackage{enumitem}
\usepackage{amsmath,eucal}
\usepackage{float}

\newcommand{\review}[1]{#1}
\newcommand{\reviewinequation}[1]{#1}
 \usepackage{pdflscape}
      \usepackage{amssymb}
\usepackage{mathrsfs}
\usepackage{tabularx}
\usepackage{relsize}
\providecommand{\keywords}[1]{\textbf{\textit{Index terms---}} #1}
\usepackage[utf8]{inputenc}
\usepackage{tikz-cd}
\usetikzlibrary{arrows}
\tikzset{
commutative diagrams/.cd,
arrow style=tikz,
diagrams={>=stealth}} 

\usepackage[english]{babel}

\usepackage{lmodern}
\makeatletter
\ifcase \@ptsize \relax
  \newcommand{\miniscule}{\@setfontsize\miniscule{4}{5}}%
\or
\or
  \newcommand{\miniscule}{\@setfontsize\miniscule{5}{6}}%
\fi
\makeatother

\makeatletter
\ifcase \@ptsize \relax
  \newcommand{\nano}{\@setfontsize\miniscule{3.5}{4.5}}%
\or
  \newcommand{\nano}{\@setfontsize\miniscule{4.5}{5.5}}%
\or
  \newcommand{\nano}{\@setfontsize\miniscule{4.5}{5.5}}%
\fi
\makeatother

\newcommand{\fxpt}{{\raisebox{1.8pt}{\nano$\blacklozenge$\hspace{1.7pt}}}} 
\newcommand{\balita}{\raisebox{1.8pt}{\text{ \nano$\bullet$\hspace{1.7pt} }}}

\usepackage{mathtools}

\newtheorem{theorem}{Theorem}[section]

\newtheorem{cor}[theorem]{Corollary}
\newtheorem{lemma}[theorem]{Lemma}
\newtheorem{proposition}[theorem]{Proposition}

\theoremstyle{definition}
\newtheorem{definition}[theorem]{Definition}
\newtheorem{example}[theorem]{Example}

\theoremstyle{remark}
\newtheorem{remark}[theorem]{Remark}

 \usepackage{textcomp}
\numberwithin{equation}{section}
\usepackage{graphicx}
\usepackage{xcolor}
 \definecolor{VerdeFH}{HTML}{009374}
\definecolor{BR}{HTML}{002B36}

\newcommand{\eeqref}[1]{eq. \eqref{#1}}

    \usepackage[margin=1.51in]{geometry}
    \usepackage[bookmarks=false]{hyperref}
    \hypersetup{
         colorlinks   = true,
         citecolor    = blue!65!black
    }
    \hypersetup{linkcolor=blue!62!black}

    \definecolor{azulf}{HTML}{0092D2}
\definecolor{azulc}{HTML}{00AEEF}
    \numberwithin{equation}{section}
    \newcommand{\MN}{M_N(\C)}
    \newcommand{\sun}{\mathfrak{su}(N)}
    \newcommand{\totimes}{\otimes_\tau}
           \newcommand{\Sym}{\mathrm{Sym}}

    \renewcommand{\and}{\hphantom{a}\mbox{and}\hphantom{a}}

    \DeclareMathOperator{\Tr}{Tr}
        \DeclareMathOperator{\Hess}{Hess}
    \DeclareMathOperator{\STr}{STr}
    \newcommand{\STrN}{\STr_{\hspace{-1pt}N}}
    \newcommand{\tr}{\mathrm{tr}}

    \newcommand{\mtr}[1]{\mathrm{#1}}

    \newcommand{\Dif}{\mathcal{D}}
    \newcommand{\A}{\mathcal{A}}
    \newcommand{\dif}[1]{\mathrm{d}#1}

    \newcommand{\re}{\mathbb{R}}

    \newcommand{\dervpar}[2]{\frac{\partial #1}{\partial #2}}
    \newcommand{\diag}{\mtr{diag}}
    \newcommand{\dervfunc}[2]{\frac{\delta #1}{\delta #2}}

    \renewcommand{\H}{\mathcal{H}}
    
    \newcommand{\C}{\mathbb{C}}

    \newcommand{\M}{\mathcal{M}}
    \newcommand{\ii}{\mathrm{i}}
    \newcommand{\ee}{\mathrm{e}}
    
    \newcommand{\inv}{^{-1}}
    \newcommand{\mtc}[1]{\mathcal{#1}}
    \newcommand{\Z}{\mathbb{Z}}
    \newcommand{\N}{\mathbb{N}}

    \newcommand{\hp}[1]{^{(#1)}}

    \newcommand{\im}{\mtr{im}}

    \newcommand{\where}{\mbox{where}\,\,}
    
    \newcommand{\e}{\varepsilon}
\newcommand{\Mn}{M_n(\C)}

    \usepackage{booktabs}
\usepackage{colortbl}
\usepackage{xcolor}
\usepackage{graphicx}

\usepackage[skip=.6ex]{subcaption}
\newenvironment{salign} 
  {\csname align*\endcsname}
  {\csname endalign*\endcsname} 
  
  \newcommand{\numerada}{\refstepcounter{equation}\tag{\theequation}}
  
\DeclareFontEncoding{LS1}{}{}
\DeclareFontSubstitution{LS1}{stix}{m}{n}

\newcommand*\kay{%
  \text{%
  \fontencoding{LS1}%
  \fontfamily{stixscr}%
  \fontseries{\textmathversion}%
  \fontshape{n}%
  \selectfont\symbol{"6B}}}

  \newcommand*\DDay{%
  \text{%
  \fontencoding{LS1}%
  \fontfamily{stixscr}%
  \fontseries{\textmathversion}%
  \fontshape{n}%
  \selectfont\symbol{'104}}}
   
  \DeclareMathOperator{\Day}{\DDay} 
  
  \newcommand*\jay{%
  \text{%
  \fontencoding{LS1}%
  \fontfamily{stixscr}%
  \fontseries{\textmathversion}%
  \fontshape{n}%
  \selectfont\symbol{"7C}}}
  
    \newcommand*\iay{%
  \text{%
  \fontencoding{LS1}%
  \fontfamily{stixscr}%
  \fontseries{\textmathversion}%
  \fontshape{n}%
  \selectfont\symbol{"7B}}}

\makeatletter
  \newcommand*\textmathversion{\csname textmv@\math@version\endcsname}
  \newcommand*\textmv@normal{m}
  \newcommand*\textmv@bold{b}
\makeatother

   \let\suma=\sum
   \let\langleb=\langle
   \let\rangleb=\rangle
 \usepackage{mathabx}
\newcommand{\TrL}{\Tr_{\hspace{-1pt}\Lambda}}
\newcommand{\TrN}{\Tr_{\hspace{-1pt}N}}
\newcommand{\Cfree}[1]{\C_{\langleb #1\rangleb}}
\newcommand{\Cn}{\Cfree{n}}
\newcommand{\CnN}{\C_{\langleb n\rangleb, N}}

\newcommand{\CnNtens}[1]{\CnN^{\,\otimes \hspace{1pt} #1}}
\newcommand{\Cntens}[1]{\Cn^{\,\otimes \hspace{1pt} #1}}
\newcommand{\CntensL}[1]{\C^{\,\otimes \hspace{1pt} #1}_{\langleb n \rangleb, \Lambda }}

\newcommand{\TrNX}[1]{\TrN\left(\frac{X^{#1}}{#1}\right)}
\newcommand{\includegraphicsd}[2]{\raisebox{-.415\height}{\includegraphics[width=#1\textwidth]{#2}}}
\newcommand{\includegraphicsw}[2]{\raisebox{-.415\height}{\includegraphics[width=#1\textheight]{#2}}}
\newcommand{\includegraphicswextra}[3]{\raisebox{-.#3\height}{\includegraphics[width=#1\textheight]{#2}}}
\newcommand{\ac}{\scalebox{0.94}{\ensuremath{\mathsf{a}}}} 
\newcommand{\bc}{\scalebox{0.94}{\ensuremath{\mathsf{b}}}} 
\newcommand{\cc}{\scalebox{0.94}{\ensuremath{\mathsf{c}}}} 
\newcommand{\dc}{\scalebox{0.94}{\ensuremath{\mathsf{d}}}} 
\newcommand{\acb}{\bar{\ac}} 
\newcommand{\bcb}{\bar{\bc}} 
\newcommand{\ccb}{\bar{\cc}} 
\newcommand{\dcb}{\bar{\dc}} 
\newcommand{\ea}{e_{\scalebox{0.7}{\ensuremath{\mathsf{a}}}}} 
\newcommand{\eb}{e_{\scalebox{0.7}{\ensuremath{\mathsf{b}}}}} 
\newcommand{\Za}{Z_{\scalebox{0.7}{\ensuremath{\mathsf{a}}}}} 
\newcommand{\Zb}{Z_{\scalebox{0.7}{\ensuremath{\mathsf{b}}}}} 
\newcommand{\Fa}{F_{\scalebox{0.7}{\ensuremath{\mathsf{aa}}}}} 
\newcommand{\Fb}{F_{\scalebox{0.7}{\ensuremath{\mathsf{bb}}}}} 
\newcommand{\etaa}{\eta_{\scalebox{0.7}{\ensuremath{\mathsf{a}}}}} 
\newcommand{\etab}{\eta_{\scalebox{0.7}{\ensuremath{\mathsf{b}}}}} 
\newcommand{\itemb}{\item[$\balita$]}
\DeclareMathSymbol{\zplus}{\mathbin}{operators}{"2B}
\DeclareMathSymbol{\zminus}{\mathbin}{symbols}{"00}

\usepackage{afterpage}
 \DeclareCaptionLabelFormat{mylabel}{#1 #2 \hspace{1ex}}
\begin{document}

\title[On random NCG multimatrix models: Functional RG \& Free Algebra]{On multimatrix models motivated
by \\ random Noncommutative Geometry I: \\ the Functional Renormalization Group \\as a flow in the free algebra} 


  \author
  {Carlos I. {P\'erez-S\'anchez}}
\address{Faculty of Physics, University of Warsaw}
\curraddr{ul. Pasteura 5
 02-093 Warsaw, Poland, European Union   
} 
 \email{cperez@fuw.edu.pl}
  \begin{abstract}   
 Random noncommutative geometry can be seen as a Euclidean path-integral quantization approach 
 to the theory defined by the Spectral Action
 in noncommutative geometry (NCG). With the aim of investigating phase transitions in random
NCG of arbitrary dimension, we study the nonperturbative 
 Functional Renormalization Group for multimatrix models whose action consists of noncommutative polynomials
 in Hermitian and anti-Hermitian matrices. Such structure 
 is dictated by the Spectral Action for the Dirac operator in Barrett's 
spectral triple formulation of fuzzy spaces.
The present mathematically rigorous treatment 
puts forward ``coordinate-free'' language that
might be useful also elsewhere, all the more so because 
our approach holds for general multimatrix models.  The toolkit is a noncommutative calculus
on the free algebra that allows to describe the generator of the renormalization group flow---a noncommutative Laplacian  introduced here---in terms of Voiculescu's cyclic gradient and Rota-Sagan-Stein noncommutative derivative. 
We explore the algebraic structure of the Functional Renormalization Group Equation and, as an application of this formalism, we find the $\beta$-functions, identify the fixed points in the 
large-$N$ limit and obtain the critical exponents of $2$-dimensional geometries in two different signatures.

\end{abstract}
\keywords{random noncommutative geometry, functional renormalization group, Wetterich equation, matrix models, fuzzy geometry, 
quantum spacetime, large-$N$ limit, free probability}

 \subjclass[2010]{Primary: 58B34, 81-XX; Secondary:  15B52, 46L54 }

  \maketitle

\fontsize{11.0}{14.0}\selectfont  
\setcounter{tocdepth}{1}
\tableofcontents
 \fontsize{11.49}{14.9}\selectfont  

\section{Introduction}\label{sec:intro}%
Random Noncommutative Geometry (NCG), initiated by 
 Barrett and Glaser \cite{BarrettGlaser}, is a
 path-integral approach to the quantization of noncommutative geometries.   
 This problem is mathematically interesting \cite[\S 18.4]{ConnesMarcolli}  and has already been  addressed by diverse methods in \cite{MvS}, \cite{EntropySpectral} and \cite{KhalkhaliQuantization}.
 Also in physics, a satisfactory answer would shed light 
 on the quantum structure of spacetime
 from a different angle. Namely, what seems to individuate a formulation of quantum of gravity 
 in terms of NCG-structures is that these 
 provide
 a natural language to treat both pure gravity and gravity coupled with matter  at a geometrically indistinguishable footing. 
This holds for (the classical theory of) established matter sectors 
like the Standard Model \cite{CCM,BarrettSM} and 
some theories beyond it \cite{surveySpectral}.   
\par 

Although the last point evokes rather the mathematical 
elegance of the NCG-applications, also from a pragmatic viewpoint 
it is important to stress that 
the search for a quantum theory of gravity that is capable of incorporating matter is of
physical relevance: ``matter matters'' reads 
 for instance in the asymptotically safe road to quantum gravity \cite{Dona:2013qba} (see also \cite{BookReuterSaueressig}).  
Indeed, a quick argument \cite{EichhornTalk} based on the Renormalization Group (RG)  
discloses the mutual importance of each sector to the other, concretely \begin{itemize}
\itemb gravity loops like $ \raisebox{-6pt}{\includegraphics[width=24pt]{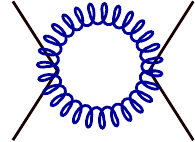}} $ appear and influence matter and 
\itemb in a similar way, matter modifies the gravity sector 
$ \raisebox{-5pt}{\includegraphics[width=24pt]{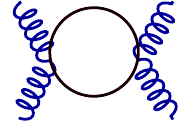}} $
\end{itemize}
in the RG-flow. This suggests that both ought to be simultaneously treated and motivates 
us to develop, as a first step, the \textit{Functional Renormalization Group}
in random NCG, where potentially both sectors might harmonically coexist.  \par 
\thispagestyle{empty}
The Functional Renormalization Group Equation (FRGE; see the comprehensive up-to-date review \cite{Tausend}) is a modern framework describing the Wilsonian RG-flow \cite{Wilson:1974mb}
that governs the change of a quantum theory with
scale. From the technical viewpoint, in order to determine the effective action, 
the FRGE---derived by Wetterich and Morris \cite{Wetterich,Morris}---offers an alternative 
to path-integration by replacing that task with
a differential equation.  \par 
 In this paper, the model of space(time) we focus on
is an abstraction of fuzzy spaces \cite{DolanOConnor,DolanHuetOConnor,SteinackerFuzzy},
whose elements were later assembled into a spectral triple (the spin geometry object in NCG) called \textit{fuzzy geometry} \cite{BarrettMatrix, BarrettGaunt}. 
For the future, in a broader NCG context, it would be desirable to relate the FRGE 
to the newly investigated truncations 
in the spectral NCG formalism \cite{GlaserStern,GlaserSternZwei,ConnesWvS} (see \cite{DAndrea:2013rix} for a preceding related idea), but for initial investigations fuzzy geometries are interesting enough and also in line with them, 
e.g., for the case of the sphere \cite[Sec. 3.3]{vanSuijlekom:2020bvh}. 
\par 

One particular advantage of a fuzzy geometry being a spectral triple 
is the contact with Connes' NCG formalism, in particular, the ability to encode the geometry in a (Dirac) operator $D$ that 
serves as path integration variable in the quantum theory. Since fuzzy geometries are finite-dimensional, one can provide a mathematically precise definition of the partition function $\mathcal Z= \int_{\text{\scriptsize}} \ee^{-\Tr f(D)} \dif {D}$ that corresponds to 
the Spectral Action  $\Tr f(D)$,
as far as $f$ is a polynomial, in contrast to the bump function $f$ 
used originally by 
Connes-Chamseddine \cite{Chamseddine:1996zu}.    
In fact, this way to quantize fuzzy geometries
 was shown \cite{BarrettMatrix,BarrettGlaser} to 
lead to a certain class of multimatrix models further 
characterized in \cite{SAfuzzy}. \par 
On the physics side, finite-dimensionality should not 
be seen as a shortage, as this dimension is related 
to energy or spatial resolution; in fact, rather it is in line with the 
existence of a minimal or Planck's length.
This is intuitively clear for the fuzzy sphere \cite{Grosse:1994ed} on which---being spanned
by finitely many spherical harmonics---it is impossible
to separate (i.e. to measure) points lying arbitrarily near.  \par 

This discrete-dual picture (Fig. \ref{fig:fuzzytr}) 
can be interpreted as a pre-geometric phase, analogous to 
having simplices as building blocks of spacetime in 
discrete approaches to quantum gravity as Group Field Theory \cite{Baratin:2013rja},
Matrix Models \cite{EynardCounting, dFGZ} or Tensor Models \cite{InvitationGurau}. For those theories, but also for \review{other approaches (e.g., Causal Dynamical Triangulations \cite{CDT})},
it is important \cite{Ambjorn:2012ij,Delepouve:2015nia,StatusTM,ABAB,AndreasJohannes} to explore phase transition to a manifold-like phase; in analogous way, the study of a condensation 
of fuzzy geometries to a continuum is physically relevant \cite{Glaser:2016epw} (also
addressed analytically in dimension-$1$ by \cite{KhalkhaliPhase}).
With this picture in mind, we estimate here candidates 
for such phase transition.  
 \begin{figure}[h!]
   \includegraphics[width=0.5\textwidth]{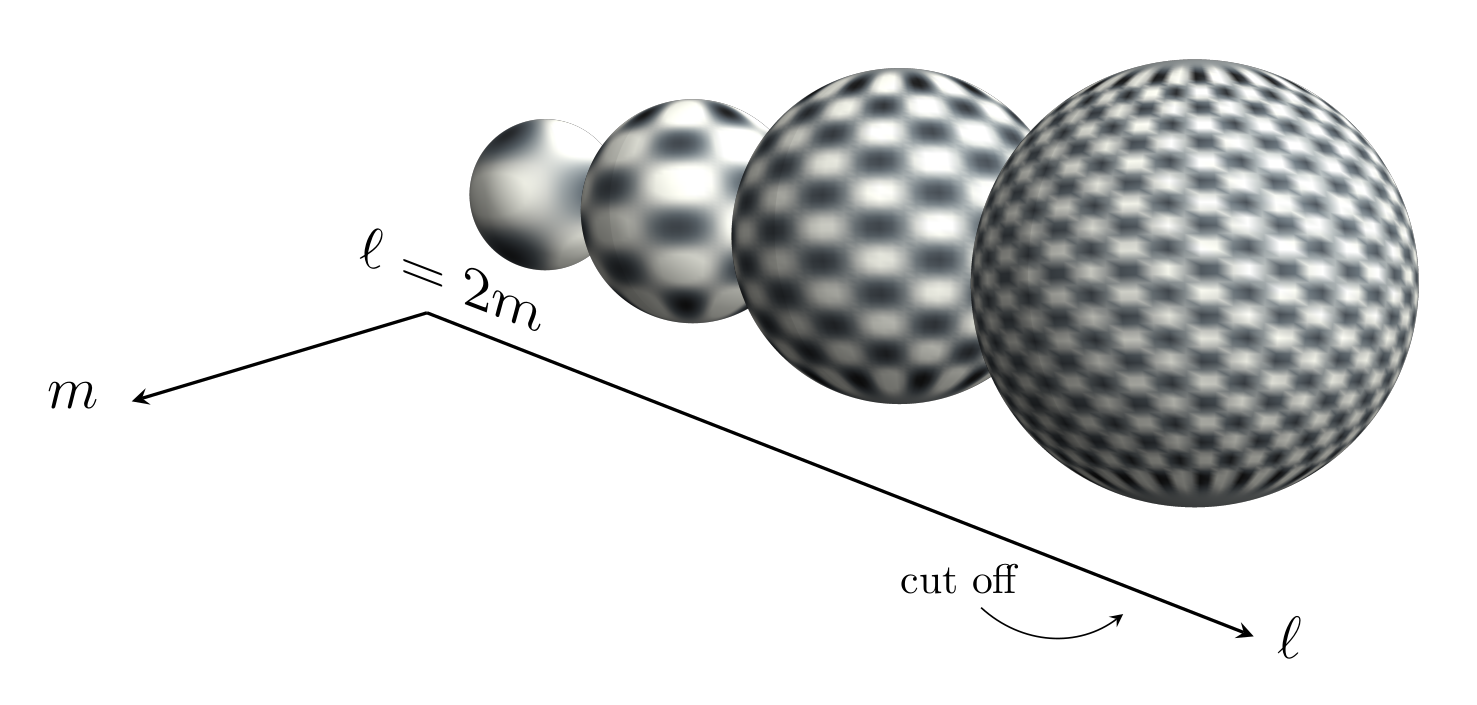}
\caption{A caricature of selected algebra elements of the fuzzy sphere. 
The real part of spherical harmonics $Y^{m=\ell/2}_{\ell}$ for $\ell=4,10,20,40$\label{fig:fuzzytr}}
 \end{figure}
 \par
The largest part of this paper develops 
the mathematical formalism that allows such exploration.
On top of well-known quantum field theory (QFT) techniques, the non-standard results 
this paper bases on can be divided into three classes: 
\begin{itemize}\setlength\itemsep{.4em}
\itemb  \textsc{The models} are originated in {Random NCG} \cite{BarrettGlaser}.   
Barrett's characterization of Dirac operators 
makes contact with certain kind of multimatrix models \cite{BarrettMatrix}. 
Their Spectral Action was systematically 
computed in \cite{SAfuzzy}, organized by chord diagrams,
which reappear here.

\itemb \textsc{The tool} is the Functional Renormalization Group. The main idea 
of the RG-flow parameter being the (logarithm of) the matrix size
appeared in \cite{Brezin:1992yc} and consists in reducing the $N+1$ square matrix $\varphi$  to effectively
obtain a $N\times N$ matrix field by integrating out 
the entries 
$\varphi_{a, N+1}, \varphi_{N+1,a}$ ($a=1,\ldots N+1$). Eichhorn and Koslowski provided 
the nonperturbative, modern formulation of the Brezin--Zinn-Justin idea.  
They put it forward for Hermitian matrix models in 
\cite{EichhornKoslowskiFRG} (preceded by a similar approach to scalar field theory
on Moyal space \cite{SfondriniKoslowski} and followed by an extension to tensor models \cite{StatusTM}). 
They did not present a proof and in fact it will be convenient to prove for multimatrix models 
the FRGE, as this equation actually dictates us the algebraic 
structure (needed for the so-called ``$FP\inv$-expansion'' \cite[Sec. 2.2.2]{FPinverse}) and exonerates us from making any choice.

\itemb \review{Although the Eichhorn-Koslowski approach orients us to find suitable truncations
and their scalings to take the large-$N$ limit were auxiliary, 
the mathematical structure we deal with here 
is constructed from scratch and does not rely on theirs (which turns out to be entirely replaced)}{}. 
\textsc{The language} that facilitates this is {abstract noncommutative algebra.} In order 
to state the RG-flow in ``coordinate-free'' fashion, we use 
Voiculescu's cyclic derivative \cite{Voi_gradient} 
and the noncommutative derivative defined by Rota-Sagan-Stein \cite{Rota}. 
 \end{itemize}
We do not assume familiarity with any of these 
references and offer a self-contained approach. 
 
\subsection{Organization, strategy and results}

In Section \ref{sec:toolkit} we develop the algebraic 
language needed for the rest of the paper.
We introduce a noncommutative (NC) Hessian and a NC-Laplacian on the free algebra, 
given in terms of noncommutative differential
operators defined by \cite{Rota} and \cite{Voi_gradient,VoiculescuFreeQ}.
A graphical method to compute this 
second order operator is provided.  Section \ref{sec:toolkit} prepares the algebraic 
structure that will turn out to emerge in the proof of 
Wetterich-Morris equation for multimatrix models. 
\par 

Section \ref{sec:MMM_NCG} briefly reviews fuzzy geometries 
and how their Spectral Action is computed in terms of elements
of the free algebra---in mathematics called \textit{words} or 
\textit{noncommutative polynomials} and in QFT-terminology \textit{operators}---that define 
a certain class of multimatrix models. For 2-dimensional fuzzy geometries,
we provide a characterization of allowed terms in the 
resulting action functional. \par 
In Section \ref{sec:FRGEderivation} the FRGE is proven to be governed by the NC-Hessian; 
in Section \ref{sec:FPandtruncations} we introduce 
truncations and projections in order to compute the $\beta$-functions.
Also there, the ``$FP\inv$ expansion'' is developed in the large-$N$ limit,
and the \textit{tadpole approximation}, corresponding to order one in that expansion,
is restated as a heat equation\footnote{That a Laplacian plays a role in the Functional 
Renormalization Group and that this 
has the form of is not a surprise \cite[Sec. 3.3]{Salmhofer}.} 
whose Laplacian is noncommutative (the one of Sec. \ref{sec:toolkit}). \par 

Once the formalism is ready, we do not directly proceed 
with fuzzy geometries, but in Section \ref{sec:CrossCheck} we briefly
reconsider the treatment of the FRGE for Hermitian matrix models.
A couple of points justify this interlude: 
\begin{itemize}\setlength\itemsep{.4em}
 \itemb It serves as a bridge from the index-computations
 in matrix models to index-free ones proposed in the present paper.
  
 \itemb By using a well-known result
 to be reproduced by the FRGE, we calibrate the infrared regulator (IR-regulator) that we
 shall use for the fuzzy geometry matrix models. With a  
 quadratic, instead of the already studied linear IR-regulator, 
 the fixed point is closer to the exact value 
 $-1/12$ for gravity coupled with conformal matter.
 
 \itemb Finally, since the number of 
 flowing operators for the Hermitian matrix model is relatively small, it is 
 helpful for the sake of clearer exposition
 to present a case whose techniques
 fit in a couple of pages to prepare
 the more complex fuzzy two-matrix models.
\end{itemize}

The actual application  
of the formalism appears in Section \ref{sec:2MM}. 
We treat there a class of two-matrix models 
 that lies in an orthogonal direction 
to the well-investigated two-matrix model that describes the Ising model 
\cite{KazakovIsing,StaudacherCombi,EynardLargeN2MM}, often just referred to as ``\textit{the} 2-matrix model'', due to its importance.
To wit, whilst in the Ising two-matrix model the (trace of) $AB$
appears as the only interaction mixing the two random 
matrices, $A$ and $B$, NCG-models forbid this very operator. 
Instead, these matrices 
interact via several elements of the free
algebra and its tensor powers, 
i.e. via (traces of) words 
\[ABAB, A^2B^2,A^3BAB,\ldots \quad \and  \quad
A\otimes A, A^2\otimes B^2,  A\otimes A^2BAB,\ldots \, \]
whose exact form has been investigated in \cite{SAfuzzy},
also for higher dimensions. The RG-flow we analyze does not 
take place inside the space of Dirac operators 
---in which coupling constants of the same polynomial degree are correlated---but we consider
the general situation in which 
the symmetry breaking by the IR-regulator kicks the 
RG-flow out to (couplings indexed by a larger subspace of tensor products of) the free algebra.
\par 
\begin{figure}[t!]
\includegraphics[width=.890\textwidth]{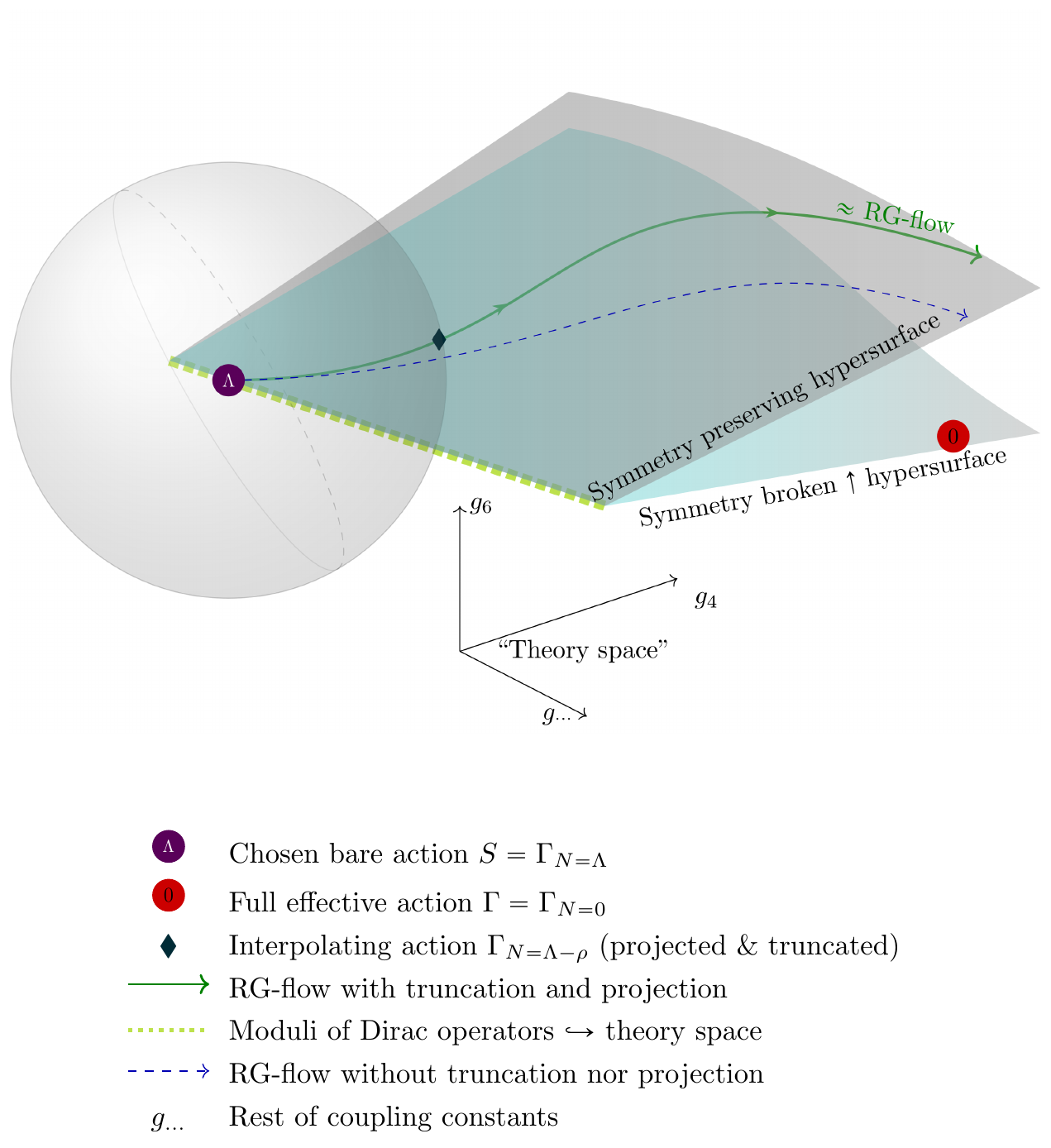}
\caption{Picture of the theory space 
and two hypersurfaces. The lower one, which considers the modified 
Ward-Takahashi identity (mWTI), is where the exact flow takes place.
The upper one is an approximation with finitely many
parameters. If $\rho$ is small, the approximation 
ignoring the mWTI, together with the 
truncation and projection for the approximated 
RG-flow ($\approx $\textit{RG-flow}) 
is assumed to not to be far apart from 
the actual interpolating action
\label{fig:piekna}}
\end{figure} For an arbitrary-dimensional
fuzzy geometry 
the \textit{bare action}---the starting point of the RG-scale $t=\log \Lambda$ (or \textit{energy scale} $\Lambda$)---is chosen in the space of Dirac operators inside the full \textit{theory space}, 
the space of running couplings. The exact RG-path
\textit{ends} at the precise effective \review{action at RG-scale\footnote{Actually,
it since $\Lambda$ will be identified with the matrix size
the lowest value for the RG-time $t$ is $0=\log 1=\log \Lambda_0$. But at this point
``we do not know this yet''.} $t=-\infty$}, 
which is too hard to determine at present. 
Making the RG-flow computable introduces two types of 
errors: on the one hand, deviations caused by projections that consider only operators with unbroken symmetries and, on the other hand, errors due to truncations introduced in order to keep the number of flowing parameters finite. 
This is depicted in Figure \ref{fig:piekna} in a pessimistic scenario, later improved 
in view of the results of Section \ref{sec:Compatibility}.
\par 
The large number of the NC-polynomial interactions, on top of the ordinary polynomials in each matrix, makes the projected and truncated RG-flow still computationally demanding\footnote{For instance, at sixth order the NCG-matrix model 
includes up to 48 operators in a double-trace even-degree
truncation. In contrast, in the same truncation, the Ising two-matrix model 
would include at most 19 operators, but the RG-flow
does not combine letters in the latter case.} and at this stage a further simplification is helpful. Namely, we look for critical exponents corresponding to 
solutions to the fixed point equations that obey
the duality $A \leftrightarrow B$, whenever the signature allows it. 
We find those solutions inside a hypercube in theory space (with coordinates $g_i$ obeying $|g_i|\leq 1$), which,  
even if it is not the full exploration, it exhausts the 
scope of the $FP\inv$-expansion. Further improvements
are discussed in Section \ref{sec:Conclusion},
together with the conclusion. 
To ease the reading, some oversized expressions involved in proofs
are located outside the main text\footnote{In the Annales Henri Poincaré version of this preprint, Appendices B, C, D and E correspond to Supplementary Material I, II, III, IV  respectively.} (see Appendices B, C,  D and E). Also Appendix \ref{sec:Glossary}
serves as a glossary and guide on the notation. 
\section{Noncommutative calculus}\label{sec:toolkit}
We address the noncommutative calculus in several (say $n$)
variables. The object of interest is the free algebra 
spanned by an \textit{alphabet} of $n$ letters $x_1,\ldots, x_n$. The elements of the free algebra
are the linear span of words in those $n$ letters, the product being concatenation. 
Although the physical theories we address are 
well described by the real version $\re_{\langleb n \rangleb}$ of it, we consider 
the complex free algebra $\Cn$. There exists in  $\Cn$ an \textit{empty word}, denoted by $1$, 
that behaves as multiplicative neutral. Other than $1$, the
letters of the alphabet do not commute. \par

Rather than in the generators $x_i$ in the abstract free algebra, 
we are interested in their realization as matrices\footnote{This
section is the mathematical background 
of the FRGE. So far, $N$ is still fixed.}, $x_i=  X_i\in \MN$
for each $i=1,\ldots,n$. 
In contrast with the convention of taking
self-adjoint generators, we have reasons to allow anti-Hermitian generators
and set instead
 \begin{subequations}\label{econditions}
 \begin{align} \label{esigns}
  X^*_{i}=\pm X_i  \qquad \mbox{ if }\quad e_i=\pm 1\, \quad(i=1,\ldots,n).
 \end{align}
 In this section the signs $e_i$ are input; 
 later these will be gained from the NCG-structure, which additionally 
imposes 
 \begin{align}\label{tracelessnesscond}
  \TrN  (X_{i}):= \sum_{a=1}^N (X_i)_{aa}=0  \qquad \mbox{ if }\quad  e_i=-1\,.
 \end{align}
 \end{subequations}
When the $n$ generators are 
$N\times N$ matrices, it will be convenient 
to denote the free algebra by $\CnN$.
Having fixed signs $e_i$ ($i=1,\ldots, n$), we let
\begin{align}\label{ModuliN}
\mathcal M_N = \big\{ (X_1,\ldots, X_n) \mid  \mbox{conditions } \eqref{econditions} \mbox{ hold for each } X_i\in \MN \big\}
\,,
\end{align}
with some abuse of notation concerning the omitted parameters. The tracelessness condition \eqref{tracelessnesscond} is of no relevance in this section, but 
important later.\par 
The empty word, which corresponds to the identity matrix $1_N\in\MN$, generates the \textit{constants}. The elements of the free algebra that are not generated 
 by the empty word are referred to as \textit{fields}:
 \begin{align} \label{splitting}
 \CnN=
 \overbrace{\C \cdot 1}^{\text{constants}}{\hspace{-9pt}}_N\, \oplus \,\,
 \langleb \underbrace{ X_{\ell_1} X_{\ell_2}  \cdots X_{\ell_k} \mid \ell_j=1,\ldots, n \,\,\and \,\,k\neq 0  }_{\text{fields}}\rangleb \,.
 \end{align}
 A similar
 terminology is employed for the analogous 
 splitting of the tensor product: \vspace{-12pt}
  \begin{align}\label{splittingtensor}
 \CnN\otimes  \CnN=
 \overbrace{\C \cdot 1_N \otimes 1_N}^{\text{constants}}_{\phantom{not fields}}\, \oplus \,\, \text{fields} \,,
 \end{align}
whose fields in this case are given by
\begin{align} \raisetag{1.2cm}
\label{fieldpart} \langleb   X_{\ell_1} X_{\ell_2}  \cdots X_{\ell_k} \otimes X_{\ell_1'} X_{\ell_2'}  \cdots X_{\ell_r'} \mid \ell_j',\ell_j=1,\ldots, n \and r+k\neq 0   \rangleb \,.
\end{align}
The free algebra is equipped with the trace 
of $\MN$:  $\TrN (Q)=\sum_{a=1}^N Q_{aa}$, $Q\in \CnN $. Instead of making this trace a state, normalizing it as usual also in probability, $\tr (1_N)=1$, we still stick to a trace satisfying $\TrN (1_N)=N$ 
in order to make power-counting arguments comparable with other references. 

\subsection{Differential operators on the free algebra}

We now elaborate on the next operators, due to 
 Rota-Sagan-Stein \cite{Rota} (in one variable to Turnbull \cite{turnbull}) 
and to Voiculescu \cite{Voi_gradient}. 
 The \textit{noncommutative derivative}---called
 also the \textit{free difference quotient} \cite{VoiculescuFreeQ,GuionnetFreeAn}---with respect to the $j$-th variable $x_j$, denoted by  $\partial^{x_j}$,  is defined 
on generators by
\begin{align} 
\nonumber
\partial^{x_j}  :\C_{\langleb n \rangleb } & \to \C_{\langleb n \rangleb } \otimes  \C_{\langleb n \rangleb } \\ 
  x_{\ell_1} \cdots x_{\ell_k}  & \mapsto \sum_{i=1} ^k  \delta_{\ell_i}^j \cdot 
  x_{\ell_1}\cdots x_{\ell_{i-1}}\otimes x_{\ell_{i
  +1}} \cdots x_{\ell_k} \label{NCderivative} \,. \raisetag{34pt}
\end{align}
The tensor product keeps track of the spot (in monomial)
the derivative acted on.  
Moreover, the \textit{cyclic derivative} $\Day^{x_j}$ with respect 
to the $j$-th variable is defined by 
\begin{align}
 \Day^{x_j}  = \tilde m \circ \partial^{x_j}  \quad\mbox{where } \tilde m :\C_{\langleb n \rangleb }\otimes 
\C_{\langleb n \rangleb }\to \C_{\langleb n \rangleb }\,, 
\quad \tilde m (A\otimes B)=BA.
\end{align} 

\begin{example}
In the free algebra generated by the Latin alphabet $\mathtt{A},\ldots,\mathtt{Z} $, one has 
\vspace{-.3cm}
\[\partial^{\mathtt{E}} (\mathtt{FREENESS})=
\mathtt{FR}\otimes \mathtt{ENESS}
+\mathtt{FRE}\otimes \mathtt{NESS}
+\mathtt{FREEN}\otimes \mathtt{SS}\,,
\] but notice that (if $1$ is the empty word)
$\partial^{\mathtt{S}} (\mathtt{FREENESS})=
 \mathtt{FREENE}\otimes \mathtt{S} 
 +\mathtt{FREENES}\otimes  1 
 $. For the cyclic derivative it holds: 
\begin{align*}
\Day^{\mathtt{E}}
(\mathtt{FREENESS})
&=\tilde m \big(\mathtt{FR}\otimes \mathtt{ENESS}
+\mathtt{FRE}\otimes \mathtt{NESS}
+\mathtt{FREEN}\otimes \mathtt{SS}\big)
\\
&=\mathtt{ENESSFR + NESSFRE+ SSFREEN}\,.
\end{align*}

\end{example}

Using the same rules for the abstract derivatives on $\Cn$ for
$\CnN$, one can make the following 
\allowdisplaybreaks[3]
\begin{proposition}\label{thm:indexfrei}
Let $Y=X_i$ be any of the generators of $\CnN$.
For any $Q\in \CnN$, the derivatives $\partial^{Y}$ and $\Day^Y$ enjoy the following properties:
\begin{enumerate}
 \itemb the abstract derivative is realized by the derivative 
with respect to a matrix:\vspace{-.4cm}
\begin{align}
\partial ^{Y}_{ab} &= \dervfunc{}{Y_{ba}}\,, 
\end{align} 
that is, letting $(U\otimes V)_{ab;cd}=U_{ab}V_{cd}$ $(U,V\in \CnN)$, one has
\[ 
\qquad[(\partial^{Y } Q)(X)]_{ab; cd}
=
 \dervfunc{}{Y_{bc}}[Q(X)]_{ad}  \qquad  \mbox{for } X=(X_1,\ldots,X_n) \in \mathcal M_N\,.
\]

\itemb The cyclic derivative 
equals the noncommutative derivative of the trace: 
\begin{align} 
\partial ^{Y} \Tr  Q &= \Day^{Y}  Q
\,.
\end{align}
\end{enumerate}

\end{proposition}
\begin{proof}
Let $Q \in \Cn$. Since the trace is linear, one can verify
the property on a monomial $Q(X)= X_{\ell_1} \cdots X_{\ell_k} $ and then obtain 
  \begin{salign}
\dervfunc{}{(X_i)_{bc}}  Q(X)_{ad}  & = 
\dervfunc{}{(X_i)_{bc}}  ( X_{\ell_1} \cdots X_{\ell_k})_{ad} \\
&=\sum_{j=1}^k
(X_{\ell_1}\cdots X_{\ell_{j-1}})_{ar}\dervfunc{(X_{\ell_j})_{rs}}{(X_i)_{bc}}
(X_{\ell_{j-1}}\cdots X_{\ell_{k}})_{sd}\\
&=\sum_{i=\ell_j}
(X_{\ell_1}\cdots X_{\ell_{j-1}})_{ab}
(X_{\ell_{j-1}}\cdots X_{\ell_{k}})_{cd}\\
&=\sum_{i=\ell_j}
(X_{\ell_1}\cdots X_{\ell_{j-1}} \otimes  X_{\ell_{j-1}}\cdots X_{\ell_{k}})_{ab;cd} \\
&=\big[(\partial^{X_i} Q)(X)\big]_{ab;cd} \,.
  \end{salign}
To obtain the second statement, notice that 
   $\partial ^{X_i}_{ab} \Tr Q $ is obtained from last
   equation by setting $a=d$ and summing, 
   \begin{salign}
    \partial ^{X_i}_{cb} \Tr [Q(X)]  & = 
   \sum_{a} \big\{(\partial^{X_i} Q)(X)\big\}_{ab;ca}
   \\& =\sum_{a}\sum_{i=\ell_j}
(X_{\ell_1}\cdots X_{\ell_{j-1}} \otimes  X_{\ell_{j-1}}\cdots X_{\ell_{k}})_{ab;ca} \\& = 
\big \{[(\tilde m \circ \partial^{X_i})Q](X)\big\}_{cb} = \big[(\Day^{X_i} Q)(X)\big]_{cb} \,,\end{salign}
where the last line follows just by the definition of cyclic derivative.
\end{proof}

\begin{definition}
The \textit{noncommutative-Hessian} (NC-Hessian) is the operator 
\begin{align}\label{NCLap}
\Hess: \im \,\TrN \to \Mn\otimes  \CnNtens{2}
\end{align}
whose $(ij)$-entry ($1\leq i,j\leq n$) in the first tensor factor is
\begin{align}
(\Hess \TrN P)_{ij} := 
(\partial^{X_i} \circ \partial^{X_j} \TrN P ) \in \CnNtens{2}\,.
\end{align}
Here, $\TrN: \CnN  \to \C$ is the ordinary trace of $ M_N(\C) \supset\CnN$. Alternatively,
\[\Hess  
= \begin{pmatrix}
   \partial^{X_1}\circ \partial^{X_1} &   \partial^{X_1}\circ \partial^{X_2} & \cdots &  \partial^{X_1}\circ \partial^{X_n} \\
   \partial^{X_2}\circ \partial^{X_1}  &  \partial^{X_2}\circ \partial^{X_2} & \cdots  &  \partial^{X_2}\circ \partial^{X_n} \\ 
  \vdots & \vdots & \ddots &  \vdots \\
    \partial^{X_n}\circ \partial^{X_1} &   \partial^{X_n}\circ \partial^{X_2} &  \cdots &     \partial^{X_n}\circ \partial^{X_n}
  \end{pmatrix}\,.
\]
It will be convenient to introduce a closely related Hessian, $\Hess_g$,
modified by the\footnote{Later it will be clear this terminology---by now we use quotation marks.}
``signature'' $g=\diag(e_1,\ldots,e_n)$, 
\begin{align}
(\Hess_g \TrN P)_{ij} := 
(e_i)^{\delta_{ij}}
(\partial^{X_i} \circ \partial^{X_j} \TrN P ) \in \CnNtens{2}\,,
\end{align}
so 
\[\Hess_g 
= \begin{pmatrix}
  e_1 \partial^{X_1}\circ \partial^{X_1} &   \partial^{X_1}\circ \partial^{X_2} & \cdots &  \partial^{X_1}\circ \partial^{X_n} \\
   \partial^{X_2}\circ \partial^{X_1}  & e_2 \partial^{X_2}\circ \partial^{X_2} & \cdots  &  \partial^{X_2}\circ \partial^{X_n} \\ 
  \vdots & \vdots & \ddots&  \vdots \\
    \partial^{X_n}\circ \partial^{X_1} &   \partial^{X_n}\circ \partial^{X_2} &  \cdots &    e_n \partial^{X_n}\circ \partial^{X_n}
  \end{pmatrix}\,.
\]
Tracing the NC-Hessian $\Hess_g$ with help of the signature  yields the \textit{noncommutative Laplacian} $\nabla^2$, that is the map
\begin{align}\label{NCLap}
\nabla^2: \im  \TrN \to \CnNtens{2}
\qquad \mbox{given by }
\nabla^2 := \sum_{i=1}^n e_i(\partial^{X_i} \circ \partial^{X_i} )\,.
\end{align}  
We abbreviate $\partial^{X_i} \circ \partial^{X_i}=(\partial^{X_j})^2=\nabla_j^2$,
 so $\nabla^2=\sum_{j=1}^n e_j\nabla_j^2$. 
\end{definition}
 
We remark that the Hessian matrix (of NC-polynomials in $\Cntens{2}$) is not symmetric. 
Clearly, the NC-Laplacian and the NC-Hessian 
vanish on degree $< 2$. On larger words, we compute them with aid of: 
 
\begin{proposition} \label{thm:NCLap}
Consider a monomial $ Q\in\CnN$,  $Q= X_{\ell_1}X_{\ell_2} \cdots X_{\ell_{k}}  $  
with $k\geq 2$. Then, for $i,j=1,\ldots,n$
\begin{align}\label{doubleNCder}
(\partial^{X_i} \circ \partial^{X_j}) \TrN Q=
\sum_{\pi=(uv)} 
\delta^{j}_{\ell_u} 
\delta^{i}_{\ell_v} \pi_1(Q) 
\otimes
\pi_2(Q)  \,,
\end{align}
where the sum runs over all (directed) pairings $\pi=(uv)$
between the letters of the word $Q$ distributed on a circle:
\begin{align} \label{CircleLaplacian}
\includegraphicsd{.3}{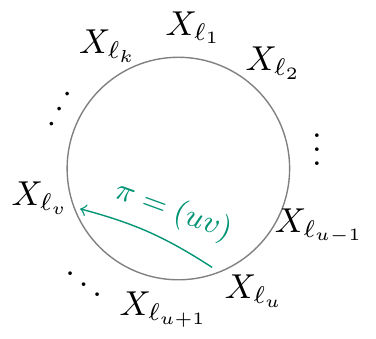}
\end{align}
In Eq. \eqref{doubleNCder},  $\pi_1(Q) $ and $\pi_2(Q)$ are the words 
between $ X_{\ell_u} $ and $X_{\ell_u} $. They fulfill that  
$\pi_2(Q) X_{\ell_u}  \pi_1(Q) X_{\ell_v}$ matches $Q$ 
up to cyclic reordering. 
\end{proposition}
As a particular case in that definition: for $\pi$ matching contiguous letters, that is
if $v=u\pm 1$, one has the empty word in between,
$\pi_{\big\{\substack{ 1 \\ 2}\big\}}(Q)=1_N$.
\begin{proof} Notice that by (2) of Claim \ref{thm:indexfrei},
\begin{align*}
\partial^{X_i}\partial^{X_j} \TrN Q
& =\partial^{X_i} \Day^{X_j} Q   
\\ &=
\partial^{X_i} \bigg(
\sum_u \delta^j_{\ell_u} X_{\ell_{u+1}}X_{\ell_{u+2}}  \cdots X_{\ell_{k}} X_{\ell_{1}} 
X_{\ell_{2}} \cdots 
X_{\ell_{u-1}} \bigg)
\\
&=
\sum_u \delta^j_{\ell_u}  \sum_{v \neq u}
\delta^i_{\ell_v}
X_{\ell_{u+1}}X_{\ell_{u+2}}   \cdots X_{\ell_{{v-1}}} \\[-5pt] &\qquad\qquad\qquad\otimes 
X_{\ell_{v+1}}
\cdots X_{\ell_{k}} X_{\ell_{1}} 
X_{\ell_{2}} \cdots 
X_{\ell_{u-1}} \\
&= 
\sum_u \delta^j_{\ell_u}  \sum_{v \neq u}
\delta^i_{\ell_v} \pi_1(Q) \otimes \pi_2(Q),\,\, \,\where \pi=(uv)
\,.\qedhere
\end{align*}
\end{proof}
Before we give some examples, notice that since for the NC-Laplacian both pairings $\pi=(uv) $ and $(vu)$ appear, one can replace the 
expression $\nabla^2 \TrN Q=\sum_{j=1}^n e_j\sum_{\pi=(uv)} 
\delta^{j}_{\ell_v}
\delta^{j}_{\ell_u}  \pi_1(Q) 
\otimes
\pi_2(Q) $ by a more symmetric one, 
\begin{align}\label{undirected}
\nabla^2 \TrN Q=\sum_{j=1}^n e_j\sum_{\substack{\pi=\{uv\} \\ \text{\scriptsize(undirected) }} } 
\delta^{j}_{\ell_v}
\delta^{j}_{\ell_u}
\big\{ \pi_1(Q) 
\otimes
\pi_2(Q)
+
\pi_2(Q) 
\otimes
\pi_1(Q)
\big\}\,. \raisetag{18pt}
\end{align}
These differential operators can be extended to products of traces using the 
same formulae that defines them in the single-trace case, but they require additional structure. 
Namely, the NC-Laplacian satisfies the rule 
\begin{align} \nonumber
 \nabla^2  \Tr^{\otimes 2} (P\otimes Q) 
 &=\nabla^2  (\Tr P \cdot \Tr Q ) \\
 &=(\nabla^2  \Tr P) \cdot \Tr Q \label{NCLap_prod}
 +(\nabla^2  \Tr Q) \cdot \Tr P \\
 & + \sum_j e_j \big\{\Day^{X_j} P \totimes  \Day^{X_j} Q +
 \Day^{X_j} Q \totimes  \Day^{X_j} P\big\}\,, \nonumber
\end{align} 
in terms of a tensor product $\totimes$ that does not receive the next natural 
matrix coordinates for monomials $U,W\in\Cfree{n}$,
\begin{subequations}
\begin{align}
(U\otimes W)_{ab;cd}:=
U_{ab} W_{cd}\,,
\label{coordsnormal} 
\end{align}
but twisted ones with respect to the transposition $\tau=(13)\in \Sym(4)$ of the four indices,
\begin{align}
(U\totimes W)_{a_1a_2;a_3a_4}:=
(U\otimes W)_{a_{\tau (1)} a_{\tau (2)};a_{\tau(3)}a_{\tau(4)}} \,,
\label{coordsstrangetau} 
\end{align}
or more clearly
\begin{align}
(U\totimes W)_{ab;cd} =
U_{cb} W_{ad}\,.
\label{coordsstrange} 
\end{align}

\end{subequations}
Before seeing how $\tau$ twists the product on $\Cntens{2}$, in the next section, 
notice how expression \eqref{NCLap_prod} follows directly
from a slightly more general one that we do prove:
\begin{proposition}
The NC-Hessian of a product of traces is 
\begin{align}\label{HessianLeibn}
 \Hess (\Tr P \Tr Q)=  \Tr Q \Hess (\Tr P) + \Tr P \Hess \Tr Q + \Delta (P,Q)
\end{align}
where the last term  is the matrix with entries 
\[[\Delta (P,Q)]_{ij}= \Day^{X_i} P \totimes \Day^{X_j} Q + 
\Day^{X_i} Q \totimes \Day^{X_j} P
\in \Cntens{2}\,.\]
\end{proposition}

The matrix just defined satisfies 
$\Delta (P,Q)= \Delta(Q,P)$ evidently---which is important since $P$ and $Q$ in  $\Hess (\Tr P \Tr Q)$
are interchangeable---but, like the Hessian, it is not symmetric in general, $[\Delta (P,Q)]_{ij} \neq [\Delta (P,Q)]_{ji}$ .  

It is convenient to split $\mtc P= \partial^{X_i} \partial^{X_j} \Tr P \in \Cn \otimes \Cn $ using (a convenient upper-index version of) Sweedler's notation $\mtc P= \sum \mtc P\hp1 \otimes \mtc P\hp2  $. The transition to the index notation
  can be expressed 
  as\footnote{Other choices are possible. However, if one applies 
  these operators to products of traces, as is the case treated here, 
  at least one of the products will show this braiding.}
  \[
     \partial_{cb}^{X_i}  \partial_{ad}^{X_j} \Tr P =  \sum  \mathcal P_{ab;cd} = \sum  \mathcal P\hp{1 } _{ab} \mathcal P\hp{2} _{cd} \,, \]
   which follows by direct computation (and is moreover supported by \cite[Eq. 4]{GuionnetFreeAn}).  

\begin{proof}
 The coordinates of the $(i,j)$-matrix block of a Hessian are given by 
 $\Hess (\Tr P)_{ij \mid ab;cd }:=
 (\Hess (\Tr P)_{ij})_{ ab;cd } $. We compute these for the product of traces:
 \begin{align*} 
  \Hess &(\Tr P \Tr Q)_{ij \mid ab;cd } \\
  &= \partial^{X_i}_{cb} \partial_ {ad}^{X_j}(\Tr P \Tr Q) \\[3pt]
  &= (\Hess (\Tr P)_{ij})_{ ab;cd } \Tr Q +   \Tr Q   (\Hess (\Tr P)_{ij})_{ ab;cd }
  \\[3pt] & \,\,+ \Day^{X_i }_{cb} P \Day^{X_j}_{ad} Q + \Day^{X_i }_{cb} Q \Day^{X_j}_{ad} P \\[3pt]
   &= ( \Tr Q  \Hess \Tr P +\Tr P  \Hess \Tr Q)_{ij|ab;cd} \\[3pt] & \,\,+( \Day^{X_i}P \totimes \Day^ {X_j} Q 
     + \Day^{X_i}Q \totimes \Day^ {X_j} P )_{ab;cd}\,. \qedhere    
 \end{align*}
\end{proof}
From the last proposition, one can easily show a similar rule holds
replacing everywhere by its version $\Hess_\sigma$ modified by $\sigma=\diag(e_1,\ldots, e_n)$ and the 
$\Delta$-matrix by $\Delta^{\sigma}(P,Q)$ which has 
diagonal entries $\Delta^{\sigma}_{ii}(P,Q)=e_i \Delta_{ii}(P,Q)$ and else those of $\Delta$.      

In the following, we sketch the action of the operator $\partial^{X_j}$ graphically. 
The convention is that the first letter of a word is the first after
the cut (arrow tail) and the last letter corresponds to the 
one before the cut (arrow head).
\begin{itemize}\setlength\itemsep{.4em}
 \itemb One can represent 
 the elements of $\im \TrN$ as words on circles. 
For the NC-derivative $\partial^{X_j}: \im \TrN\to \Cn$ one has 
\begin{align}
\qquad
 \includegraphicsd{.2}{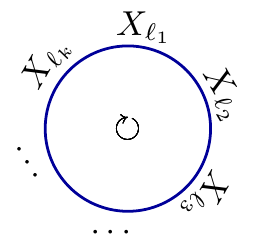}& \stackrel{}{\mapsto}\sum_{j\text{-cuts} }
\includegraphicsd{.23}{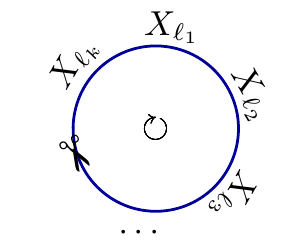} \\
& =
\sum_{j\text{-cuts} }\includegraphicsd{.555}{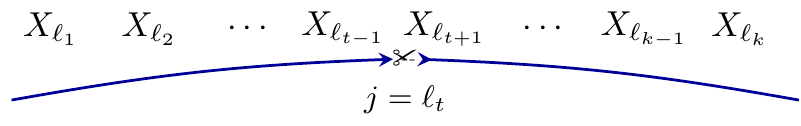} \nonumber
\end{align}
where the ends of the line in the last figure are joined (multiplied).
\itemb In the next representation, each arrow belongs to a different tensor factor. Thus, $\partial^{X_j}: \Cn \to \Cntens{2}$ acts as
\[
\includegraphicsd{.25}{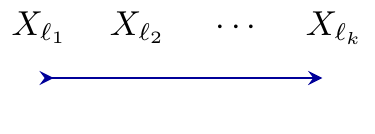} \stackrel{}{\mapsto}
\sum_{j\text{-cuts} }\includegraphicsd{.44}{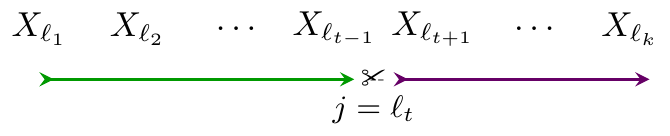} 
\]
Together, the two last pictures
give the graphical interpretation of the proof of the proposition above.
 \itemb Similarly, $\partial^{X_j}: \Cntens{k} \to \Cntens{k+1}$
 $j$-cuts at all places of each tensor-factor (line):
 \begin{align} \nonumber
 \partial^{X_j} \big(
\includegraphicswextra{.162}{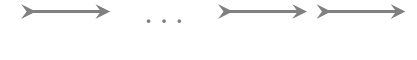}{6} \big)
& = \color{black}\sum_{r=1}^k \includegraphicswextra{.22}{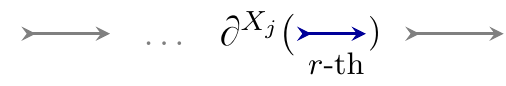}{5} \hspace{-46pt} \\
& = \sum_{r=1}^k  \sum_{j\text{-cuts} }\includegraphicsw{.15}{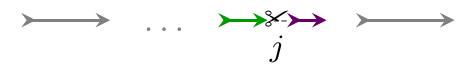}  
 \end{align}
 
 \end{itemize}
\begin{example} The next examples shall be useful below:
\begin{itemize}\setlength\itemsep{.4em}
\itemb On $\CnN$ it holds 
 $\nabla^2 \Tr  (X_iX_j)= \sum_{k=1}^n2  e_i  \delta_{i}^k \delta_{j}^k1_N\otimes 1_N = 2 e_i \delta_{j}^i 1_N\otimes 1_N $
 from the last statement, since only the empty word is between the two letters.

\itemb On $\C_{\langleb 1\rangleb, N}=\C\langleb X\rangleb$ with $X^*=X$,   $\nabla^2=( \partial^X)^2$ and ($m\geq 2$)
\begin{align}\label{Trm}
\nabla^2 \TrN \Big(\frac{X^m}{m}\Big)=\sum_{\ell=0}^{m-2} X^\ell\otimes X^{m-2-\ell}\,.
\end{align}
Now is evident that, even though $\C_{\langleb 1\rangleb, N}$
consists of ordinary polynomials, NC-derivatives are not ordinary.

 \end{itemize}

\end{example}

\begin{example}
We compute a NC-Hessian and a NC-Laplacian on 
$\Cfree{2}= \C\langleb A,B \rangleb$. With aid of Claim \ref{thm:NCLap} 
and setting  $g=\diag (e_1,e_2)= \diag (\ea,\eb)$:
\begin{align}
\Hess_g \{ \Tr(ABAB) \} & =
\begin{pmatrix}  \ea \partial^A \circ \partial^A &    \partial^A \circ \partial^B 
\\
   \partial^B \circ \partial^A  &   \eb \partial^B \circ \partial^B 
 \end{pmatrix}  \Tr(ABAB)\\
 & = 2
 \begin{pmatrix}
    \ea B\otimes B &    AB\otimes 1 +1\otimes BA  \\
   BA\otimes 1 + 1 \otimes AB  &  \eb  A\otimes A\\
 \end{pmatrix} \nonumber
\end{align}
which also explicitly shows the asymmetry of the Hessian.
To compute, say, the entry (12) of this matrix, which 
corresponds to the operator $ \partial^A \circ \partial^B $, 
one has four matches: distributing the word $ABAB$ on a circle as in \ref{CircleLaplacian},
with the arrow tail at any letter $B$, the tip of the arrow can pair the $A$ left (or clockwise)
to it or the $A$ to its right (counterclockwise). According to Claim \ref{thm:NCLap}
these contributions are, respectively, $1\otimes BA$ and $AB\otimes 1$  for each letter $B$ in the word, hence the factor 2.
The Laplacian is the trace of the Hessian, as block matrix, 
\[
\nabla^2 \Tr(ABAB) = \Tr_2\big(\Hess_g \Tr(ABAB)\big ) = 
2 \ea B\otimes B+ 2\eb A\otimes A\,.
\]
\end{example}


\subsection{The algebraic structure}
We consider sums of monomials which either have
the form $X\otimes Y$ or $X\totimes Y$ inside the same algebra:
\begin{align}
\label{defofA} 
\A_{n}=\Cntens{2 } \oplus \C^{\totimes 2}_{\langleb n \rangleb} \qquad  \mbox{ and } \qquad \A_{n,N}=\CnNtens{2} \oplus \C^{\totimes 2}_{\langleb n \rangleb,N} \,, 
\end{align}
where the second symbol emphasizes the matrix realization of the free algebra. 
On $\A_{n,N}$ there is a product $\times $ defined in coordinates by
\begin{align}\label{deftimes}
 [ (U\otimes_{\vartheta } W) \times (P \otimes_{\varpi} Q ) ]_{ab;cd} :=  
  (U\otimes _\vartheta W)_{ax;cy} (P\otimes_\varpi Q )_{xb;yd} \,,
\end{align}
where $\vartheta,\varpi$ represent the twist $\tau$ 
or its absence, and the sums are implicit. The twisted structure modifies the product according to: 
\begin{proposition} For monomials 
$U,W,P, Q\in \Cn$ one has
\begin{subequations}\label{productrules}
\begin{align}
(U\otimes  W)\times ( P\otimes Q  )&= UP \otimes WQ\,, \\
(U\otimes  W)\times ( P\totimes Q  )&= WP \totimes UQ\,, \\
(U\totimes  W)\times ( P\otimes Q  )&= UP \totimes WQ\,, \\
(U\totimes  W)\times ( P\totimes Q  )&= WP \otimes UQ \,. 
\end{align}
\end{subequations}

\end{proposition}

These rules can be remembered by identifying  tensor product of monomials
$U\otimes W$ with the block diagonal 
element $\diag (U,W)\in M_2(\Cn)$ and each twisted 
product $U\totimes W $ with the anti-diagonal $\jmath \, \diag(U,W)=\big(\begin{smallmatrix}
                                                   0&W \\ 
                                                   U& 0
                                                  \end{smallmatrix}\big)$  for $\jmath=\big(\begin{smallmatrix}
                                                   0&1 \\ 
                                                   1& 0
                                                  \end{smallmatrix}\big)$. Then, the rules \eqref{productrules} are just ar\ restatement of matrix multiplication in $ M_2(\Cn)$, but we do not state it a such since it does not work for polynomials. 
 But in fact eqs. \eqref{productrules} can be proven in coordinates:
\begin{proof}
We prove the second rule: for $a,\ldots, d=1,\ldots,N$, one has
\begin{align*}
((U\otimes  W)\times ( P\totimes Q  ))_{ab;cd} 
&= (U\otimes  W)_{am;co}  ( P\totimes Q  )_{mb;od}  \\
&= U_{am}  W_{co}   P_{ob}  Q_{md} \qquad  \mbox{(implicit $o,m$ sum)}  \\ 
&=  (U_{am}  Q_{md} ) (W_{co}   P_{ob})
 =  (WP)_{cb} (UQ)_{ad}   \\
&= (WP \totimes UQ)_{ab;cd}
\end{align*}
and that rule follows. 
The first rule \eqref{productrules} is obvious, 
the two left unproven are verified in similar way.
\end{proof}
As a caveat, notice that 
\[
(1\totimes 1)\times  (P\otimes Q)= P \totimes Q 
\quad \mbox{but}\quad 
(P\otimes Q) \times (1\totimes 1)= Q \totimes P\,.
\]
For monomials $P,Q, U,W \in \Cn$, we let also
\[
[(U \otimes_{\vartheta} W) \star ( P\otimes_{\varpi} Q) ]_{ab;cd}: = 
(U\otimes_\vartheta W)_{ab;xy} (P \otimes_{\varpi} Q)_{yx;cd}\,,\]
where $\varpi,\vartheta$ stand for either $\tau$ or an empty label.
\begin{proposition} \label{thm:rightproductformula}  It follows that
 \begin{subequations}
\begin{align}
(U \totimes W) \star ( P\totimes Q)  &=  PU \totimes WQ \,, \\
(U \otimes W) \star ( P\totimes Q)  &=U  \otimes PWQ  \,, \\
(U \totimes W) \star ( P\otimes Q)  &= WPU  \otimes Q \,, \\
(U \otimes W) \star ( P\otimes Q) &= \Tr (WP) U\otimes Q
\end{align}\label{rightproductformula}
\end{subequations}
\end{proposition} \vspace{-1cm}
\begin{proof}
We prove only the first one, the other proofs being similar:
 \begin{align} \nonumber
((U \totimes W )\star (P \totimes Q))_{ab;cd} & =  
(U \totimes W)_{ab;xy} ( P\totimes Q)_{yx;cd}
\\  &=  P _{cx} U_{xb} W_{ay}  Q_{yd} = (P U)_{cb} (W  Q)_{ad} \label{rightprodincoords} \\ 
\nonumber &=  (P U \otimes W  Q)_{cb;ad}  \\
&=
(PU\totimes WQ)_{ab;cd}\,. \nonumber  \qedhere
\end{align}

\end{proof}
One can replace the the new product by $\times$, namely using
\begin{align} \label{asdfasdf}
 (U\totimes W)\star  (P\totimes Q) =  (P\totimes W) \times (U\otimes Q)\,,
\end{align}
which holds due to 
\begin{salign}
 [ (U\totimes W)\star  (P\totimes Q)] _{ab;cd } &=( PU \totimes WQ )_{ab;cd } = P_{cx}U_{xb}W_{ay}Q_{yd} \\&= 
(U\totimes W)_{ab;xy}(P\totimes Q)_{yx;cd}\,.
\end{salign}%
Notice that in \eeqref{asdfasdf} $\tau$ no longer acts on the matrix indices and it has been transferred  to the factors:  
\[
(Y_1\totimes Y_2) \star (Y_3 \totimes Y_4) = (Y_{\tau(1)} \totimes Y_{\tau(2)}) \times ( Y_{\tau(3)}  \otimes Y_{\tau(4)})\,.
\]
Since  $( PU \totimes WQ )_{ab;cd }=(PU \otimes WQ)_{cb;ad}$, another useful expression for the sequel is
\begin{align}\label{withouttau} 
[ (U\totimes W)\star  (P\totimes Q)] _{ab;cd } = 
(U \otimes W)_{xb;ay} (P\otimes Q)_{cy;xd}\,.
\end{align}
Also, while the product $\times$ loses the twist,  $(1\totimes 1)^{\times 2}  = (1 \otimes 1  )$,
the $\star$ product preserves it $(1\totimes 1)^{\star 2} = (1\totimes 1) $ 
and in fact $(1\totimes 1)$  is the unit element:
\begin{subequations}\label{unittau}
\begin{align}
(1\totimes 1) \star (P\totimes Q) &= P\totimes Q \,,\\ (U\totimes W) \star (1\totimes 1) &= U\totimes W \,,
\\
(1\totimes 1) \star (P\otimes Q) &= P\otimes Q\,,\\   (U\otimes W) \star (1\totimes 1) &= U\otimes W\,,
\end{align}\end{subequations}%
which follows from Proposition \ref{thm:rightproductformula}. Although it might be clear from the definition of $\star$ that on
$ \Cn^{\totimes  2}$ it is associative---since there the first factor right multiplication and in the second ordinary matrix multiplication---it is reassuring
to see that it is also associative if untwisted elements are implied: 
\begin{proposition}\label{thm:associativity}
 The product $\star$ is associative on $\A_n$. 
\end{proposition}
\begin{proof}
 Let $A,B,C,D,U,W,P,Q,T,S,X,Y \in \Cn$. 
 Using Proposition \ref{thm:rightproductformula} one verifies straightforwardly that either bracketing, 
 \[
 \big \{ ( U\totimes  W + P\otimes Q) \star ( T\totimes S + X\otimes Y)   \big \}\star ( A\totimes  B + C\otimes D) 
 \]
 or 
 \[
 ( U\totimes  W + P\otimes Q) \star \big\{ ( T\totimes S + X\otimes Y)   \star ( A\totimes  B + C\otimes D) \big \}\,,
 \]
 yields due to the cyclicity of the trace the same result, namely:
 \begin{align*}
  &ATU\totimes WSB  + P\otimes (ATQSB) + WXU\otimes AYB  \\ 
  +& \Tr(QX) \cdot (P\otimes AYB) + WSCTU\otimes D + \Tr (TQSC) \cdot  (P\otimes D)\\
  +& \Tr (YC) \cdot  WXU\otimes D + \Tr(XQ)\cdot \Tr(YC)\cdot  (P\otimes D)\,.\qedhere
 \end{align*}

\end{proof}
For the sequel, more important that the Hessian is its twisted version
\begin{definition}
 The \textit{twisted NC-Hessian} $\Hess^\tau_\sigma$ is given by
 \[\Hess^\tau_\sigma := (1\totimes 1) \times \Hess_\sigma\,.
 \]
\end{definition}
In other words, by Proposition \ref{productrules}, $\Hess^\tau_\sigma$ is obtained from 
$\Hess_\sigma$ after exchanging the products $\totimes$ and $ \otimes$. 

\begin{example}\label{ex:MultHessians}
We exemplify computing the product of $\Hess_\sigma^\tau ( AABB) $, namely 
\fontsize{8.5}{14.0}\selectfont  
\begin{align*}
  \bigg(
\begin{array}{cc}
 \ea ({1}\totimes BB+BB\totimes {1}) & {1}\totimes AB+BA\totimes {1}+A\totimes B+B\totimes A \\
 {1}\totimes BA+AB\totimes {1}+A\totimes B+B\totimes A & \eb ({1}\totimes AA+AA\totimes {1}) \\
 \end{array}\bigg) \,,
\end{align*}   \fontsize{11.49}{14.9}\selectfont  
with $\Hess_\sigma^\tau \big[\Tr A \Tr(A B B)\big]$,\fontsize{8.5}{14.0}\selectfont  
\begin{align*}
  \Bigg(
\begin{array}{cc}
 \ea ({1}\otimes BB+BB\otimes {1}) & \Tr A (B\totimes {1}+{1}\totimes B)+{1}\otimes AB+{1}\otimes BA \\
 \Tr A (B\totimes {1}+{1}\totimes B)+AB\otimes {1}+BA\otimes {1} & \eb \Tr A (A\totimes {1}+{1}\totimes A) \\
\end{array}
  \Bigg)\,.
\end{align*}  \fontsize{11.49}{14.9}\selectfont  
The diagonal\footnote{
The $*$ entries of products of two Hessians are uninteresting in this paper (unless one wants to compute to third order the RG-flow). } of $\Hess_\sigma^\tau [\Tr A \Tr(A B B) ] \star \Hess_\sigma^\tau ( AABB)= \big( \begin{smallmatrix} \mathcal P & * \\ * & \mathcal Q \end{smallmatrix}\big)$, which is computed entrywise with $\star$, is given by  (recall $\ea^2=\eb^2=1$)
\begin{align*}
\mathcal P&=
\Tr A  \{1\otimes  B B A + A B B \otimes 1+A\otimes  B B +2 B\otimes  B A +2  A B \otimes B \\ & + B B \otimes A \}+1\totimes  A
   A B B +2 \cdot 1\totimes  A B A B +2 \cdot 1\totimes  A B B A  \\ &  +2 \cdot 1\totimes  B A B A +1\totimes  B B A A +2 \cdot 1 \totimes  B B B B +2  B B \totimes  B B \,,
%
\end{align*} 
and
\begin{align*}
 \mathcal Q&=\Tr A  \{1\otimes  B A B + B A B \otimes 1+A\otimes  B B +B\otimes  A B +B\otimes  B A  \\ &  + A B \otimes B+ B A \otimes B+ B
   B \otimes A+1\otimes  A A A + A A A \otimes 1 \\ & +A\otimes  A A + A A \otimes A\}+2  A B \totimes  A B  \\ &  +2  A B \totimes  B A +2  B
   A \totimes  A B +2  B A \totimes  B A \,.
\end{align*}
The $M_n(\C)$-trace (here for $n=2$, $\mathcal P +\mtc Q$) of products of Hessians---or rather of their anti-commutator---
will be shown to be fundamental for the RG-flow.
The absolute (not only cyclic) order in the letters of the 
expressions for the twisted or untwisted Hessians
of cyclic NC-polynomials absolutely matters. 
If one continues taking products of Hessians the 
order of the matrix factors does matter too (which is why 
one gets bulky expressions now). Only
at the final stage, when we take traces, we can 
cyclically permute. 
\end{example}


\section[Random NCG and multimatrix models]{Random noncommutative geometries and multimatrix models} \label{sec:MMM_NCG}

We briefly recall the foundations of fuzzy geometries, 
known to be rephrasable in terms of matrix algebras \cite{Ramgoolam:2001zx}, in 
Barrett's matrix geometry setting \cite{BarrettMatrix}.
The original definition is given in terms of spectral triples,
but in that definition the axioms 
implying the Dirac operator can be directly 
replaced by a characterization these boil down to.  

\subsection{Fuzzy geometries as spectral triples}\label{sec:fuzzy}
Given a signature $(p,q)\in\Z_{\geq 0 }^2$, a \textit{fuzzy $(p,q)$-geometry} 
consists of a quintuple
\[(M_N(\C), \,V\otimes M_N(\C), \,D,\, J, \,\gamma) \,\] whose elements 
we describe next. 
The inner product space $V$ is given the structure of $\mtc C\ell (p,q)=\mtc C\ell (\re\hp{p,q})$-module. The action $\text{\underline{c}}$ of the Clifford 
algebra on the basis elements $\theta^\mu$ of $\re\hp{p,q}=(\re^{p+q},\diag(+_p,-_q))$, where the subindex in 
each sign means its repetition that many times,
yields \textit{gamma-matrices} $\gamma^\mu=\text{\underline{c}}(\theta^\mu)$. We assume that they satisfy, 
\begin{subequations}
 \label{eq:ConventionGammas}
\begin{align} 
(\gamma^\mu)^2&=+1, \quad \mu=1,\ldots, p,  &&\gamma^\mu  \mbox{ Hermitian}, \\
(\gamma^\mu)^2&=-1, \quad \mu=p+1,\ldots, p+ q, && \gamma^\mu \mbox{ anti-Hermitian}\,,
\end{align}
\end{subequations}
that is, one has Hermitian or anti-Hermitian
gamma-matrices according to whether $\mu$ is a \textit{time-like} 
($1\leq \mu \leq p$) 
or a  \textit{spatial} index ($p< \mu \leq p+q$). This
in turn yields the 
\textit{chirality} $\gamma=(-\ii)^{s(s+1)/2} \gamma^1\cdots \gamma^{p+q}$, being 
 $s:=q-p$ the \textit{KO-dimension}. The inner product of $V$ together with the Hilbert-Schmidt
inner product  on $M_N(\C)$ endow $\H= V\otimes M_N(\C)$ with
the structure of Hilbert space the matrix algebra 
acts on in the natural way, ignoring $V$. 
Moreover, the KO-dimension determines
three signs $\epsilon, \epsilon',\epsilon''\in \{-1,+1\}$ via
 \[ \vspace{4pt}
 \begin{tabular}{ccccccccc}
$s\equiv q-p \,\,\mtr{( mod }\, 8)$ & 0 & 1 & 2 & 3 & 4& 5 &6 &7  \\[1pt] \hline\color{black}
 $\epsilon$ & $+$ & $+$ & $-$ & $-$ &$-$&$-$& $+ $ &$+ $ \\
 $\epsilon'$ & $+$ & $-$ & $+$ & $+$ &$+$&$-$&$+$&$+$ \\
 $\epsilon''$ & $+$ & + & $-$ & + &$+$& +&$-$& +\\[.21pt]
 \hline 
 \end{tabular}\]
The operator
  $J= C\otimes ($complex conjugation$)$ on $\H$ defines 
  the \textit{real structure}. Here $C:V\to V$ is anti-unitary and satisfies $C^2=\epsilon$ and $C\gamma^\mu=\epsilon'\gamma^\mu C$ for each $\mu=1,\ldots,p+q$. 
Last, but most importantly, $D$, the \textit{Dirac operator}, is a self-adjoint operator on $\H$ that satisfies the \textit{order-one condition}
$[[ D, Y' ]\,, J Y J\inv]= 0$ for each $Y,Y'\in \MN$. 
The signs in the table above imply, as part of the 
definition,
  \begin{align*}
  J^2 =\epsilon \, , \qquad 
  JD =\epsilon' DJ \,,  \qquad
  J\Gamma  =\epsilon '' \Gamma J\,.
  \end{align*}  
After the axioms are solved \cite{BarrettMatrix},
for an even dimension $q+ p$ (thus even KO-dimension), the Dirac operator has the form 
 \begin{align} \label{characDirac} \raisetag{50pt}
D 
&=
\sum_{\mu}^{\phantom d}  \gamma^\mu \otimes k_\mu
+\sum_{\mu,\nu,\rho} \gamma^{\mu}\gamma^\nu \gamma^\rho \otimes k_{\mu\nu\rho}+\ldots
\\ & 
+
\sum_{\mu,\nu,\rho} 
\gamma^{ \widehat{\mu\nu\rho} }\otimes k_{\widehat{\mu\nu\rho} }
+
\sum_{\hat{\mu}} \gamma^{\hat \mu} \otimes k_{\hat \mu}\,, \nonumber
    \end{align}
where 
\begin{itemize}\setlength\itemsep{.4em}
 \itemb $\gamma^\alpha=\gamma^{\mu_{1}}\gamma^{\mu_{2}}\cdots \gamma^{\mu_{i_{2r-1}}}$ means the product 
of all indices included in an incresingly-ordered multi-index $\alpha=(\mu_{1}\cdots 
\mu_{i_{2r-1}})$. The hatted indices are 
those omitted from the list $\{1,2,\ldots, p+q\}$. Notice that the sum runs only over  
multi-indices of odd cardinality; and
\itemb for any $Y\in \MN$, $k_\alpha$ are commutators or anti-commutators 
determined by $\alpha$ via $k_\alpha (Y)= X_\alpha Y + e_\alpha Y X_\alpha$,  being $X_\alpha\in \MN$  self-adjoint if $e_\alpha=+1$ and traceless anti-Hermitian if $e_\alpha=-1$.\end{itemize}
 For the first $p+q$ values of $i$,
 $e_i$ can be read off from $\diag(e_1,\ldots,e_{p+q})$, the signature;
 however, if $p+q\geq3$, the number $n$ of matrices that parametrize $D$ exceeds $p+q$. 
 This is also true for odd $p+q$, 
 for instance, in signature $(0,3)$ the Dirac operator can be written as
\[D=\{H,\balita\} + \iay [L_1,\balita]+ \jay [L_2,\balita]+\kay [L_3,\balita]   \,, \]
where $L_i\in \sun$ for each $i$
and $\iay, \jay , \kay$ are the quaternion units
as a realization of the pertaining gamma-matrices. In odd dimensions,
the chirality is trivial, which is why the anti-commutator term 
with a the Hermitian $N\times N$ matrix $H$ 
has a trivial coefficient, instead of a product of three different 
gammas matrices. 

The complete
 criterion \cite[App. A]{SAfuzzy} that fully determines the signs in \eqref{esigns} for even-dimensional fuzzy geometries implies
 multi-indices $\alpha$, namely
  \begin{align} \label{general_es}
e_\alpha= (-1)^{u+r-1}\, \qquad 2r=\#(\alpha)+1,\quad u=\#\{ \text{spatial indices in $\alpha$}\}.
\end{align} 
 
After a signature $(p,q)$ and the matrix size $N$
are chosen, notice that the items  $(M_N(\C), V\otimes M_N(\C), J, \gamma)$, called also a \textit{fermionic system}, are fixed. 
\par

We let $\mathcal M_N ^{p,q}$ be the space of all the Dirac operators that complete the four objects into a fuzzy geometry. This spectral triple is finite-dimensional
but does not fall into the classification 
made by Krajewski and Paschke-Sitarz \cite{KrajewskiDiagr,PaschkeSitarz}.
Using eqs. \eqref{characDirac} and \eqref{general_es}
one can obtain $\mathcal M_N ^{p,q}$
in terms of $\sun$ and $\mathbb H_N$,
the Hermitian matrices in $\MN$. For instance ${\mathcal M_N^{0,4}} =
\mathbb{H}_N^{\times 4} \times \mathfrak{su}(N)^{\times 4} $ for the 
Riemannian $4$-geometry and  $\mathcal M_N^{1,3} = \mathbb{H}_N^{\times 2} \times \mathfrak{su}(N)^{\times 6} $ for the Lorentzian case. 
However, the for formalism below it suffices to know
the space of Dirac operators for 2-dimensional fuzzy geometries:
 \[
{\mathcal M_N^{p,q}} =\begin{cases}
                    \mathfrak{su}(N) \times \mathfrak{su}(N)& (p,q)=(0,2) \\
                       \mathbb{H}_N \times \mathfrak{su}(N) & (p,q)=(1,1) \\
                       \mathbb{H}_N  \times  \mathbb{H}_N  & (p,q)=(2,0)
                     \end{cases}
\]
When we work in fixed signature, we write $\M_N={\mathcal M_N^{p,q}}$, as above in Eq. \eqref{ModuliN}. 
\subsection{The Spectral Action for fuzzy geometries}

We review how to compute the Spectral Action $\Tr f(D)$,
in order to see its relation with chord diagrams, 
simultaneously setting the terminology for Section \ref{sec:2MM}.
We restrict the discussion below to $2$-dimensional geometry
with otherwise arbitrary signature. 
\par 
As remarked in \cite{BarrettGlaser}, the computable 
spectral actions $\Tr f(D)$ require $f$, which in 
the original Connes-Chamseddine formulation is a
bump function around the origin, rather to be a 
polynomial, with $f(x)\to +\infty$ as $|x|\to \infty$. 
We thus restrict to positive, even powers of 
the Dirac operator, $\Tr D^m$, which according 
to \cite{SAfuzzy}, can be computed from \textit{chord diagrams} (C.D.)
of $m$ points. A chord diagram consists of a 
circumference with $m$ marked points and $m/2$ arcs joining them. 
These diagrams encode traces of products of gamma matrices.
For $2$-dimensional geometries, the description is relative simple, as
 no multi-indices are required:
\begin{align}
 \label{SAchords}
 \frac12 \Tr D^m = \sum_{\substack{\chi \\ \text{\footnotesize ($m$-point C.D.})}}
 \mathfrak a(\chi)\,, \quad  
 \end{align}
where the \textit{value} $  \mathfrak a(\chi)$ of the diagram $\chi$ is defined by  
\begin{align} \label{valueofCD}
 \mathfrak a(\chi)=\sum_{\mu_1,\ldots, \mu_m=1,2} &\chi^{\mu_1\cdots \mu_m} 
  \Bigg\{ \sum_{\Upsilon\in \mathscr P_{m}} \mathrm{sgn}[\mu(\Upsilon)] 
  \\ 
  &\cdot 
  \TrN \bigg(\prod^{\to}_{r\in \{1,\ldots, m\}\setminus\Upsilon} X_{\mu_r} \bigg) \times \TrN \bigg(\prod^{\leftarrow}_{r\in \Upsilon} X_{\mu_r}\bigg)
  \Bigg\} \,. \nonumber
\end{align}
Herein, for an $m$-point chord diagram and for each $\mu_1,\ldots,\mu_m=1,2$, one defines 
\begin{align}
\label{tensorchi}
\chi^{\mu_1 \dots \mu_{m} }
=
(-1)^{\#\{\text{simple crossings of chords in }\chi\}} 
\prod_{\substack{v,u =1\\ v\sim_\chi u} }^{m} \big( e_{\mu_v}\delta^{\mu_{v}\mu_{u}}\big) \,\,,
\end{align}
where $v\sim_\chi u$ means that the point $u$ and $v$
are joined by a chord of $\chi$, and $e_{\mu_v}$ are the signs in 
the signature $\diag(e_1,e_2)$ of the fuzzy $2$-dimensional geometry. The rest of the elements is given by:
\begin{itemize}\setlength\itemsep{.4em}
       \itemb  $\mathscr P_m$ is the power set of $\{1,2,\ldots, m\}$ 
       \itemb  for any $\Upsilon=\{i,j,\ldots,k\} \in \mathscr P_m$, $\mu(\Upsilon)$ is the ordered set $(\mu_i,\mu_j,\ldots, \mu_k) $ and $\mathrm{sgn}[\mu(\Upsilon)]= \prod_{r\in \Upsilon} e_{\mu_r}  $, which is a sign 
       \itemb $X_1=A$, $X_2=B$ are the (random) matrices
       \itemb the arrows on the product indicate 
the order in which it is performed; the right arrow preserves the 
order of the set one sums over and the left arrow inverts it.
       \end{itemize}

A quick way to see that the Spectral Action is real, as it should be,
bases on the observation that
for each word $w$ originated by a chord diagram $\chi$,
its adjoint $w^*$ is originated by the 
mirror image of $\chi$, denoted by $\chi^*$. But this being a 
chord diagram, it also appears summed in Eq. \eqref{SAchords}. 
\begin{align} \raisetag{55pt}
\text{If } 
\includegraphicsd{.135}{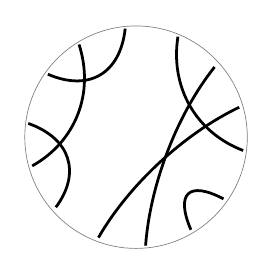} \text{originates $w$}\quad \Rightarrow\quad \Bigg\{\!\!\includegraphicsd{.135}{Chords14BW}\!\!\Bigg\}^* = \reflectbox{\includegraphicsd{.135}{Chords14BW}}
\text{originates $w^*$}\,.
\end{align}

When the running indices in Eq. \eqref{valueofCD} take a 
particular value, we color the chords of the corresponding chord diagram: 
green, if at the ends of the chords there is a matrix $A$ 
and violet\footnote{In the printed version green is just light gray and violet is black.} if it is $B$. 
To a fixed word, say $B^4 A B^2 A^3$, generally many diagrams contribute, 
\begin{align}
\includegraphicsd{.14}{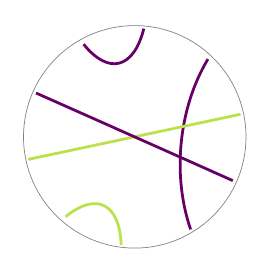} &+  \includegraphicsd{.14}{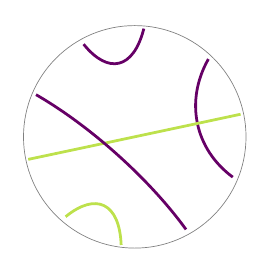}+\includegraphicsd{.14}{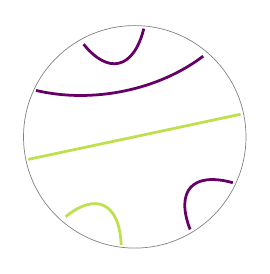} +\ldots +\includegraphicsd{.14}{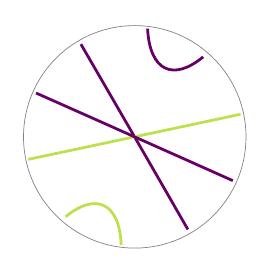}   \nonumber \\
& +\includegraphicsd{.14}{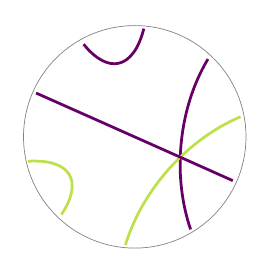} + \includegraphicsd{.14}{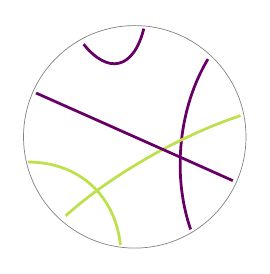} \label{manydiags} 
\end{align}
and we now show that for certain words their sum cannot vanish. (This 
will lead to the well-definedness of the chosen
truncation schemes in the FRGE.)

\begin{lemma}\label{thm:coeffssingle}
Given any word $w \in \Cfree{2}=\C\langleb A,B\rangleb$ of even degree in each of its generators, 
$\deg_A(w),\deg_B(w)\in 2 \Z_{\geq 0}$, it holds 
\[[\TrN w](\Tr D^m )\neq 0\,, \qquad \text{for } m=\deg_A(w)+\deg_B(w)\,. \]
That is, $\TrN w$ has non-zero coefficient in the Spectral Action
 for a suitable power of the Dirac operator $D$.
\end{lemma}

\begin{proof}
We use a chord diagram argument. The general situation is that
not only one chord diagram gives rise to $w$. Although the 
existence of such diagram having $m$ chords is trivial to exhibit,  
there exists the risk that all those 
diagrams add up to zero. We now verify 
that this is impossible. \par 

Suppose that $[\TrN w]\mathfrak a(\chi)$ does not vanish. This does not fix $\chi$ but leaves still a freedom of exchanging all green chord ends, corresponding to $A$, among 
themselves and the same for and all violet ones,
which correspond to $B$. (As a matter of illustration, for the word $B^4 A B^2 A^3$ above 
the first line shows such moves among $B$-chords
and the lower the $A$-chords.) All the diagrams $\chi'$ with $[\TrN w]\mathfrak a(\chi')\neq 0$ are obtained by either of these moves 
applied to the initial $\chi$, hence the number of 
such $\chi'$ is
\begin{align*}
\#\{\text{$\deg_A(w)$-point C.D.'s}\} \times \# \{ &\text{$\deg_B(w)$-point C.D.'s}\}
\\ &= (\deg_A(w)-1)!! \times (\deg_B(w)-1)!!\,,
\end{align*}
which by assumption is the product of odd numbers,
thus itself odd. Notice that the value of the diagrams $\chi'$ might
only differ from $\chi$ by a sign, which is determined by the crossings of the chords. Indeed,
 in Eq. \eqref{valueofCD}
the terms inside the curly bracket are fixed by hypothesis,
and in \eqref{tensorchi}
the product $ \prod_{\substack{v,u=1;\, v\sim u} }^{m} \big( e_{\mu_v}\delta^{\mu_{v}\mu_{u}}\big)$ is the same. This implies, that 
all the diagrams contributing to $w$ never can cancel,
as the sum $\sum_{r=1}^{2 \ell -1} \varepsilon_r$ never does (for $\ell\in\N, \varepsilon_r\in \{-1,+1\}$). Therefore the coefficient $[\TrN w](\Tr D^m )$ of $\TrN w$ in $\Tr D^m$ is by Eq. \eqref{SAchords} non-zero.
\end{proof}

Given an $m$-point C.D., a non-trivial \textit{partition} is a set $\Upsilon \in \mathscr P_m$
which is neither empty, nor is complement is, $\Upsilon^{\mathrm c}\neq \emptyset$. 
Such an $\Upsilon$ splits a diagram into product
of traces of two words. These words can be read
off from the diagram, according to \eqref{valueofCD} one counterclockwise 
the other clockwise. For instance,
\[
\includegraphics[width=33mm]{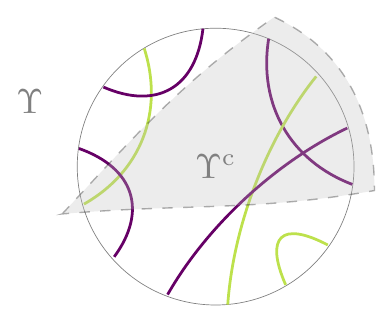}
\]
produces $\TrN (BAB^2A) $ from $\Upsilon^{\mathrm c}$ (denoted by a shaded region)
and $\TrN(BAB^4A^3)$  from $\Upsilon$. 
These non-trivial subsets $\Upsilon$ play the main role in the next  
\begin{lemma} \label{thm:coeffsbitrace}
Let  $\Cfree{2}=\C\langleb A,B\rangleb$, as before. 
The coefficient of the double-trace $\TrN(w_1) \TrN(w_2)$
in the Spectral Action is non-trivial 
for any word  $w=w_1\otimes w_2 \in \Cfree{2}\otimes \Cfree{2}$ satisfying
$\deg_A(w_1)+\deg_A(w_2),\deg_B(w_1)+\deg_B(w_2)\in 2 \Z_{\geq 0}$: 
\[[\TrN^{\otimes 2}w](\Tr D^m )\neq 0\,, \qquad \text{for } m=\textstyle\sum_{r=1,2}\deg_A(w_r)+\deg_B(w_r)\,. \]
\end{lemma}
 \begin{proof}  
We use a similar, albeit longer, argument to the single-trace case of Lemma \ref{thm:coeffssingle}. 
Since otherwise the statement reduces to Lemma \ref{thm:coeffssingle} above,
we assume that neither $w_1$ nor $w_2$ is the trivial word $1_N$. This means 
that only non-trivial partitions ($\Upsilon, \Upsilon^{\mathrm c}\neq \emptyset$)
can generate $w$, that is, if $\mathfrak a (\chi)\neq 0$, then $w$ is listed in
\begin{align} \label{midproof} \raisetag{62pt}
\sum_{\mu_r=1,2} \chi^{\mu_1\cdots \mu_m} 
\Bigg\{\sum_{\substack{\Upsilon \in \mathscr P_{m}\\ 
\Upsilon,\Upsilon^ c \neq \emptyset }}  \mathrm{sgn}[\mu(\Upsilon)] 
  \cdot 
  &\TrN \bigg(\prod^{\to}_{r\in \{1,\ldots, m\}\setminus\Upsilon} X_{\mu_r} \bigg) \\ &  \cdot  \TrN \bigg(\prod^{\leftarrow}_{r\in \Upsilon} X_{\mu_r}\bigg)
  \Bigg\}  \,, \nonumber 
\end{align}
with the condition that the first 
and the second traces yield simultaneously $\TrN w_1$ and $\TrN w_2$, in either correspondence. To wit, we have the following cases:
\vspace{4pt}
\begin{itemize}\setlength\itemsep{.4em}
 \itemb \textit{Case I.} If $\{ \TrN (w_1^*), \TrN (w_2^*)\}=\{\TrN (w_1), \TrN (w_2)\}$
 as sets. 
 \vspace{6pt}
 \itemb \textit{Case II.} If the trace of both adjoint words are
 different, that is if 
 $ \TrN (w_1^*) \neq  \TrN (w_1), \TrN(w_2)$ as well as $\TrN (w_2^*) \neq  \TrN (w_1)$, $\TrN(w_2). \vspace{6pt}$  
 \itemb \textit{Case III.} 
  For $\{r,l\}=\{1,2\}$, if
  one coincides, $\TrN (w_r^*) = \TrN (w_v)$,  
then the other does not, $\TrN (w_{l}^*) \neq  \TrN (w_u)$, $\{u,v\}=\{1,2\}$. 

\end{itemize}\vspace{4pt}
 In the first case, if $\Upsilon \in \mathscr P_m$ originates 
these words, so does $\Upsilon^{\mtr c}$,
and their contribution to the previous sum 
is doubled, for $\mathrm{sgn}[\mu(\Upsilon)]=\mathrm{sgn}[\mu(\Upsilon^{\mtr c})]$.  Hence, in Case I, we can sum over 
half of the elements encompassed inside the curly brackets in \eqref{midproof}. 
Since we excluded the trivial partitions, 
the total of sets in that sum is
$\#(\mathscr P_m)-2=2^m-2= 2\cdot (2^{m-1}-1)$. 
By hypothesis, the result of 
\eqref{midproof} is twice the sum over 
half elements, which is $2^{m-1}-1$. 
But in Cases II and III,
we also can do so, since $\Upsilon^{\mtr c} $ 
does not reproduce the word $w_1\otimes w_2$,
so we can ignore the half of the sets \eqref{midproof}. 
In any case, the sum is a multiple of 2 (Case I) of, or 
directly (Cases II-III), a sum over $2^{m-1}-1$ elements,
which is an odd number, since $w_1\otimes w_2$ is not the trivial
word and thus $m > 1$.
Since the three cases are the only possibilities 
given the two words, the partial conclusion 
is that the set $\Upsilon$ in \eqref{midproof} runs over an odd number of independent elements.

Again, finding a C.D. $\chi$ that generates $w$ is 
not hard: one puts together the letters $w_1w_2^*$
and joins by chords, matching letters. And again, 
this diagram is ambiguous up to 
a factor of $
(\deg_A(w_1w_2^*)-1)!! \times (\deg_B(w_1w_2^*)-1)!!$.
Considering the initial paragraph, the total number of 
terms is \begin{align}
\label{degsbitr}
(2^{m-1}-1) \cdot (\deg_A(w_1w_2^*)-1)!! \times (\deg_B(w_1w_2^*)-1)!! \in 2\N +1 \,,
\end{align} 
since the sum of all such diagrams 
is the product of \eqref{degsbitr}
with the non-redundant odd number. 
By the same token as before,
the sum over all the signs listed in Eq. \eqref{degsbitr} 
cannot vanish. 
\end{proof}

Concrete expressions for $f(z)=\frac14\big(\frac{z^2}{2} + \frac{z^4}4+\frac{z^6}6\big)$ 
are given below\footnote{The common 
$1/4$ factor results from removing redundant 
partitions by a set $ \Upsilon$ and its complement 
$\Upsilon$ in Eq. \protect\eqref{valueofCD}, and from $1/\dim_\C(V)=1/2$. }. From now on, we agree to write down rather the 
operators  $\Tr^{\otimes\,2}$ should be applied to, in order to get the actual monomials 
in the action. 
As for the signs, it is convenient to set $\ea:=e_1$ and $\eb:=e_2$.

\setlength{\leftmargini}{1em}
\begin{itemize}
\itemb \textit{Quadratic operators:}
 \begin{align} \label{QuadOps}
   1_N \otimes  \Big(\frac{\ea  }{2}   A^2 +\frac{\eb   }{2}   B^2 \Big)
   +  \frac{1}{2} (A \otimes A )+     \frac{1}{2} (B \otimes B )\,.
 \end{align}

\itemb \textit{Quartic operators:}
\begin{salign}
  & 1_N \otimes  \Big(\frac{1}{4} A^4+ \frac{1}4 B^4   +\ea \eb  A^2B^2  -\frac12 \ea \eb   ABAB \Big) \\  \label{QuartOps}
 &+   AB \otimes AB +2\ea \eb    A^2 \otimes B^2  + (\ea  A^3 + \eb  AB^2  ) \otimes A \numerada\vphantom{ \Big)} \\ \vphantom{ \Big)} &  
 +(\ea   A^2B + \eb  B^3  ) \otimes B   +3 
  A^2\otimes A^2 + 3 B^2\otimes B^2  \,.
 \end{salign}

\itemb \textit{Sextic operators:} The part bearing a $1_N$ factor is: 
 \begin{subequations}\label{sexticnormal}
 \begin{salign} \numerada
&1_N\otimes \big\{  \ea  A^6+ 6\eb   A^4 B^2-6 \eb   A^2(AB)^2+3\eb  (A^2B)^2     \\
&  \hspace{19.5pt} +
\eb  B^6+ 6\ea   A^2 B^4 - 6 \ea   B^2(BA)^2 + 3\ea   (B^2A)^2    \big\} \,,
\end{salign}and bi-trace terms are:\begin{salign}
 \vphantom{ \frac{10}3}&\phantom{+o}    A \otimes (  2 A^5 +2 A B^4 
+6 \ea  \eb  A^3B^2  -2 \ea  \eb A^2BAB)   \\ 
 \vphantom{ \frac{10}3}&+     B \otimes (  2 B^5 +2 B A^4 
+6 \eb  \ea  B^3A^2  -2 \eb  \ea B^2ABA) 
\\  \vphantom{ \frac{10}3}
& +  8 A B   \otimes  [   \ea  A^3B+
   \eb   B ^3 A  ] \\ & \vphantom{ \frac{10}3} \numerada
+ A^2    \otimes \big\{\eb   [8  A ^2 B ^2 -2 
B A B A  ]  + \ea   [5 A ^4    +  B  ^4 ] \big\} \\
 & \vphantom{ \frac{10}3}
+  B^2    \otimes \big\{\ea   [8  B ^2 A ^2 -2 
A B A B  ]  + \eb   [5 B ^4    +  A  ^4  ] \big\}\\
&+  \frac{10}3 (A^3 \otimes A^3 ) + 4 \ea  \eb  (A B^2 \otimes A^3) + 6 (A^2 B \otimes  A^2 B)  \\ 
&+   \frac{10}3 (B^3 \otimes B^3 ) +4  \eb  \ea  (B A^2 \otimes B^3) +6 (B^2 A \otimes B^2 A) \,.
\end{salign} 
\end{subequations}
\end{itemize}
Notice that neither
\begin{align*}
 A\cdot A\cdot A\cdot A\cdot A\cdot B,\,\qquad & A\cdot A\cdot A\cdot B\cdot B\cdot B, & A\cdot
   A\cdot B\cdot A\cdot B\cdot B,\\ A\cdot A\cdot B\cdot B\cdot A\cdot B,\,\qquad & A\cdot
   B\cdot A\cdot B\cdot A\cdot B, &  A\cdot B\cdot B\cdot B\cdot B\cdot B.
\end{align*}
nor their cyclic permutations are allowed. The same holds for  
any nontrivial partition of these into two tensor factors (e.g., $A\cdot A \otimes A\cdot A\cdot A\cdot B$), as they are not compatible 
with chord diagrams, in the sense mentioned at the beginning of this subsection. 
We also remark that $\totimes$-products do not appear in the Spectral Action. 

\section{Deriving the Functional 
Renormalization Group Equation}\label{sec:FRGEderivation}

We are interested in a nonperturbative approach
and pursue the RG-flow governed by Wetterich-Morris equation (or FRGE). Polchinski equation\footnote{For Polchinski equation, 
a review \cite{krajewskireiko} on tensor models might include 
complex matrix models as a rank-$2$ case. This can 
be used as starting point for a perturbative approach for these
multimatrix models. } \cite[Eq. 27]{Polchinski:1983gv} can be more
suitable in a perturbative approach.

 We start with the \textit{bare action} $S[\Phi]$ that describes
 the model at an ``energy'' scale $\Lambda \in \N$ (\textit{ultraviolet cutoff}). Let $\Phi$ be an $n$-tuple of matrices $\Phi= (\varphi_1,\ldots, \varphi_n) \in  \M_{\Lambda}$, 
 but the following discussion can be easily be made more general 
 taking $\Phi \in  \MN^{n}$. 
 Motivated by fuzzy geometries, the bare action $S$ is assumed to be a functional of the form 
 \begin{align}
  S[\Phi]=\Gamma_{\Lambda}[\Phi]= \Lambda \cdot  \Tr  P + \suma_{\alpha}(\Tr  \otimes \Tr ) (\Psi_\alpha \otimes \Upsilon_\alpha)\,, 
 \end{align}
being $P$ and each $\Psi_\alpha$ and $\Upsilon_\alpha$ in the finite sum 
a noncommutative polynomial in the $n$ matrices,
$P, \Psi_\alpha, \Upsilon_\alpha\in\re_{\langleb n \rangleb}=\re\langleb \varphi_1,\ldots,\varphi_n\rangleb$. The trace $\Tr=\Tr_
\Lambda$ is that of $M_\Lambda(\C)$. \par

Our derivation of Wetterich-Morris Equation for multimatrix models 
is inspired by the ordinary QFT-derivation (e.g., \cite{Gies}) for the first steps. Let  
\begin{align}\label{measure}
\exp (\mtc W[J]):= \mtc Z[J]:=  \int_{\M_\Lambda} \ee^{-S[\Phi]+ \Tr (J\cdot\Phi)} \dif \mu_\Lambda(\Phi)\,,
\end{align}
being $J=(J^1,\ldots,J^n) \in \M_N$ 
an $n$-tuple of matrix sources $J^i$, and $J\cdot\Phi = \sum_{i=1}^n J^i \varphi_i $
the sum of the $n$ matrix products. 
Here $ \dif \mu_\Lambda(\Phi)$ is the product Lebesgue measure on $\mathcal M_\Lambda$,
for which the notations $\int _\Lambda [\dif\varphi]  (\balita)$ and $\int_\Lambda \Dif\varphi  (\balita) $ are also common, mostly in physics.\par
The fundamental object is the effective action $\Gamma$, 
obtained by the Legendre transform of the free energy $\mtc W[J]$, 
\begin{align}
\Gamma[X]=\sup_J \Big( \Tr ( J \cdot X) -\mtc  W[ J]  \Big)\,.
\end{align}
Here $X$ denotes the $n$-tuple $X=(X_1,\ldots,X_n)$ of
\textit{classical fields}
$X_i:=\partial^{J^i} \mathcal W[J]  = \langleb \varphi_i\rangleb$. 
The supremum creates the dependence $J=J[X]$ and yields a 
functional depending only on $X$. 
Notice that since $J\in \M_N$, 
each source obeys the 
same (anti)-Hermiticity relation
$(J^i)^*=e_iJ^i$ as $\varphi_i$, for 
each $i=1,\ldots,n$. 
As a consequence, $(J\cdot \Phi)^*=(\Phi\cdot J)$ and the classical 
fields obey the expected rules:
\[X_i^*=\big(\partial^{J^i} \mathcal W[J] \big)^*
=\partial^{(J^i)^*} \overline{(\mathcal W[J])}=e_i X_i\,.\]

The effective action $\Gamma[X]$ 
contains all the quantum fluctuations at all energy scales. In practice, one uses an 
interpolating average effective action that incorporates 
only the fluctuations that are stepwise integrated out; 
the average effective action $\Gamma_N[X]$ results 
after integration of the modes having an energy larger than $N$ 
(i.e. matrix indices larger than $N$), while lower degrees of freedom not yet integrated. 
\review{The parameter $N$ serves as a threshold splitting the modes in high and low; 
the latter sit in the $N \times N$ block.}
Lowering $N$ makes $\Gamma_N[X]$ to approximate the full effective action $\Gamma$.
\par 
The progressive elimination of degrees of freedom is obtained by adding a mass-like term 
\begin{align} \label{IRreg}
(\Delta S_N)[\Phi]=
\frac12 \sum_{a,b,c,d=1}^\Lambda  \sum_{i=1}^n e_i(\varphi_i)_{ba} (R_{N}^\tau)_{ab;cd}  (\varphi_i)_{dc}\,, \qquad ( \Lambda \geq N \in \N).
\end{align}
This regulator has been adapted from that of Eichhorn-Koslowski to the multimatrix case\footnote{The next treatment holds for $1_n\otimes R_N \to \omega \otimes (1\totimes 1)$
with $\omega \in \Mn$ diagonal, but we stay with the easiest choice.}.
Typically the function $R_N^\tau:\{1,\ldots,\Lambda\}^4\to \re$ restricts the sum to some $N$-dependent region, but the sum-limits in Eq. \eqref{IRreg} 
allow for a freedom of regulators $R_N^\tau$.
Here, $R_N^\tau$ is not meant as a matrix: in particular its $k$-th power
$(R_N^\tau)^k$ does not imply $k-1$ sums but rather the $k$-th power pointwise. This can be guaranteed by
assuming 
\begin{subequations}\label{symmsofRN}
\begin{align}
 (R_N^\tau)_{ab;cd}= r_N(a,c) (1 \totimes 1 )_{ab;cd} =  
 r_N(a,c) \delta_c^b \delta^d_a\end{align}
for a $\re$-valued function $r_N$, and to 
satisfy 
\begin{align}
  (R_N^\tau)_{ab;cd}= (R_N^\tau)_{ba;dc} \qquad \text{and}\qquad 
  (R_N^\tau)_{ab;cd}= (R_N^\tau)_{dc;ba}\,,
\end{align}
\end{subequations}%
which hold by imposing $r_N(a,c)=r_N(c,a)$ for all $a,c$. 
Since $\tau$ implies a twist in the product, we stress that 
$R_N^\tau$ is not a multiple of the identity, only 
\begin{align}\label{twistedR}
 (R_N)_{ab;cd}:= r_N(a,c) \delta_{ab} \delta_{cd}= r_N(a,c) ( 1 \otimes 1)_{ab;cd} = r_N(c,a) ( 1 \otimes 1)_{cd;ab}
\end{align}
is. The choice of $R_N$ is arbitrary up to
the following three conditions\footnote{This is customary to state in FRGE-papers. 
 This condition deserves a mathematical study itself, 
 in order to find a precise characterization. This is left as a perspective and commented on later. }:
 \begin{enumerate}\setlength\itemsep{.4em}
 \itemb $(R_N)_{ab;cd} > 0$ for low modes, i.e. $\max \{a,b,c,d\}/N\to 0$ 
 \itemb $(R_N)_{ab;cd} \to 0$ for high modes, i.e. $N/\min \{a,b,c,d\}\to 0$  
 \itemb $(R_N)_{ab;cd} \to \infty$ as $N\to \Lambda \to \infty$ 
\end{enumerate}%
\noindent
which have the following effect, respectively:
\begin{enumerate} \setlength\itemsep{.4em}
 \item[(1)] the infrared (IR) regulator suppresses low modes: as a result these are not integrated out, unlike high modes, which do contribute to the average effective action $\Gamma_N$
 \item[(2)] is an initial condition for low $N$, i.e. ensures that one
  eventually recovers the full quantum effective action
  by lowering $N$
 \item[(3)] is an initial condition for large $N$ and ensures
that one can recover the bare action $S$ as $N\to \Lambda \to \infty$ via the saddle-point approximation.
\end{enumerate}
\begin{figure}[h!]
\includegraphics[width=.70\textwidth]{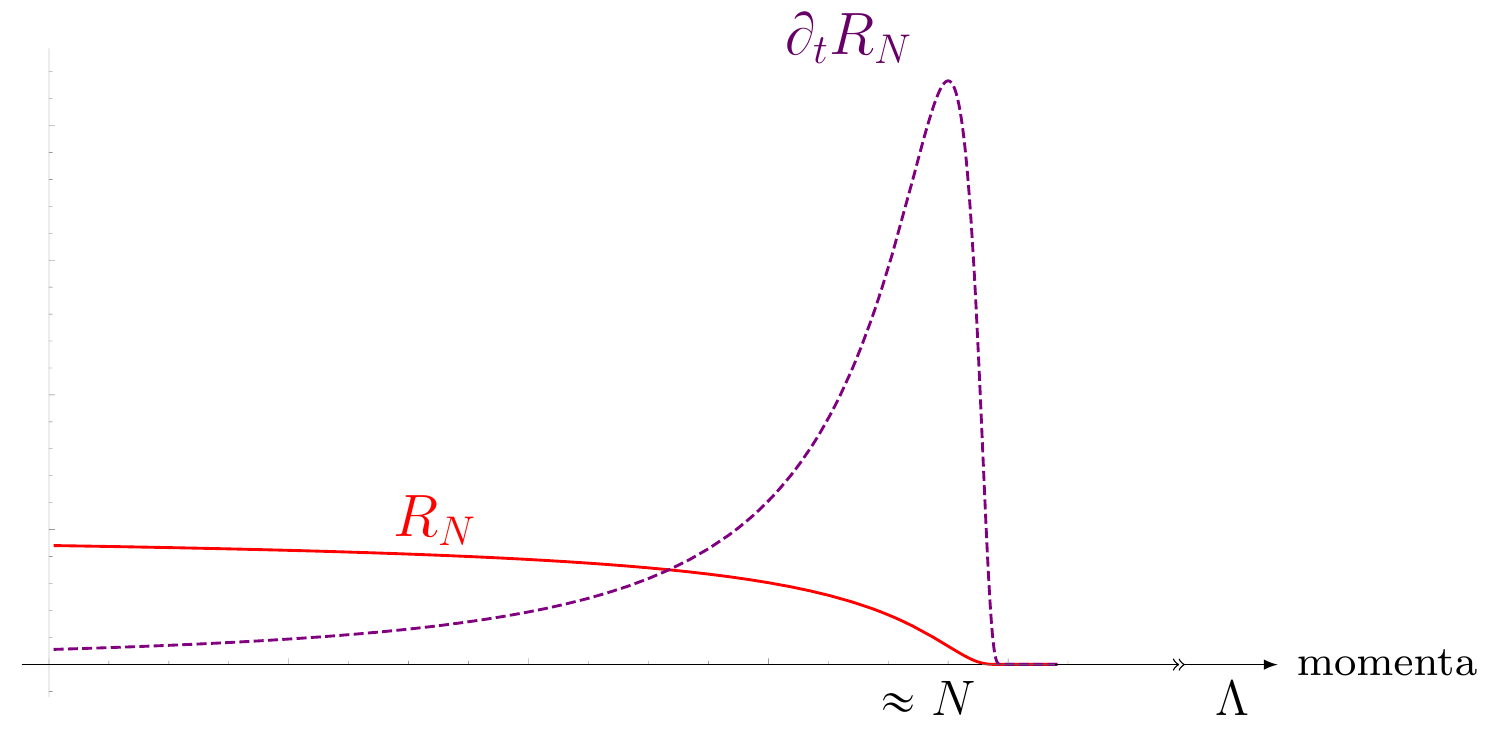}
\caption{The idea behind the regulator $R_N$ and its logarithmic 
derivative, here illustrated with 
a `bump function': $R_N$ protects
the IR degrees of freedom,
while those higher than $N$ are integrated out. 
\reviewinequation{Thus $N$ is the ``momentum threshold'' that 
splits modes into high- and low-momenta.}}%
\end{figure}%
Thus, incorporating $\Delta S_N$ to the action IR-regulates the functional
\begin{align}
 \exp \big(\mtc W_N[J]\big):= \mtc Z_N[J]:=  \int_{\M_\Lambda} 
 \ee^{-S[\varphi]-\Delta S_N[\varphi]+ \Tr (J\cdot \varphi)} \dif \mu_\Lambda(\Phi)
\end{align}
in terms of which one can obtain the \textit{interpolating average effective action}
\begin{align}
\Gamma_N[X]:= \sup_J \Big( \Tr(J\cdot X) -\mathcal W_N[J]\Big)- (\Delta S_N) [ X]
\,.
\end{align}
In practice, one uses the FRGE in order to determine it, instead
of performing the path-integral. This equation is usually displayed in 
physics in terms of a \textit{supertrace} $\STr$ we next define 
on the \textit{superspace} $\Mn \otimes \A_{n,\Lambda}=M_n(\A_{n,\Lambda})$.  Typical 
elements there form an $n\times n$ matrix $\mathcal T $ with entries 
\[
(\mathcal T_{ij}) =  \sum \mathcal T\hp 1_{ij}  \otimes  \mathcal T \hp 2 _{ij} \,.
\]
for some matrices $\mathcal{T}\hp{1}_{ij}, \mathcal T\hp 2_{ij} \in\C_{\langleb n \rangleb ,\Lambda}$,
whose four remaining entries we separate using a vertical bar, to avoid confusion: 
 \[
 \mathcal T = (\mathcal T_{ij|ab;cd})_{\substack {i,j=1,\ldots,n \\ a,b,c,d= 1,\ldots, \Lambda}} \in M_n(\CntensL{2})=M_n(\A_{n,\Lambda})\,.
 \]
We let also $\mathbf{1}= 1_n\otimes 1_\Lambda\otimes 1_\Lambda$, lest our notation becomes very loaded (which is a neutral element if $\A_n$ is endowed with $\times$)
but also notice that according to eqs. \eqref{unittau} only $\mathbf 1_\tau= 1_n\otimes 1_\Lambda\totimes 1_\Lambda$
 acts as a unit with respect to the $\star$-product. The supertrace is given by 
\begin{subequations}\label{defofSTr}
 \begin{align}
\STr&=\Tr_n\otimes \Tr_{\A_n}: M_n(\A_n) \to \C\,\,\\ 
\STr(\mathcal Q ) &=
  \sum_{i=1}^n \sum_{a,b=1}  ^\Lambda
  \mathcal Q_{ii \mid aa;bb} =  
\sum_{i=1}^n \sum_{a,b,c,d
=1 }  ^\Lambda
  \mathcal Q_{ii \mid ab;cd} (\delta _a^b \delta _{c}^d)
  \,.
\end{align}
\end{subequations}
Since knowing the matrix size will be useful, we use $\TrL^{\otimes 2}$
sometimes instead of $\Tr_{\A_2}$, but as the next $n=2$ example shows, it is important 
to be careful with twisted products  whose factors are merged inside a same trace:
\begin{align*}\STr\bigg( \begin{matrix}
                          1 \otimes A^4 & * \\
                           * & B^2 \totimes B^2  
                          \end{matrix}\bigg) & =\Tr_{\A_2} ( 1 \otimes A^4 + B^2\totimes B^2    ) \\ &= \Lambda \Tr (A^4) + \Tr  (B^4)  \,.\end{align*}
\begin{proposition} \label{thm:FRGE}The interpolating effective action $\Gamma_N$ of a matrix model 
with $X=(X_1,\ldots, X_n)\in \M_N^{p,q}$ satisfies for each $N\leq \Lambda$ Wetterich-Morris equation, which reads
\begin{align} \tag{FRGE}
\label{Wetterich}
\partial_t \Gamma_N[X]&=\frac{1}{2} \STr \bigg(   \frac{  \partial_t{R_N^\tau}}{\Hess_\sigma^\tau \Gamma_N[X] + R_N^\tau}\bigg)\,,
\end{align}
being $t=\log N$ the RG-flow parameter 
and $\sigma=\diag(e_1,\ldots,e_n)$ with $  X^*_{i}=\pm X_i\,$ iff  $e_i=\pm 1$. These signs 
are determined by the signature $(p,q)$ of the fuzzy geometry
that originates the matrix model---which for dimensions $p+q \leq 2$ coincides with $g=\diag (e_1,\ldots, e_{p+q})$---and else are given by Eq. \eqref{general_es}.  The quotient 
of operators is meant with respect to the $\times$ product.
\end{proposition}
Also $n=2$ if $p+q=2$ and $n=8$ if $p+q=4$, with general rule 
$n=2^{p+q-1}$ as far as $p+q$ is even \cite{SAfuzzy} and 
$R_N$ is economic notation for $1_n\otimes  R_N $.  After the proof, we provide the strategy to compute the RHS. 
\review{The quantity in the ``denominator'' of the FRGE requires some 
care; its well-definedness is addressed in Section \ref{sec:FP}.}
 \begin{proof}Directly from the definition 
 of the interpolating action one has 
 \begin{align} \nonumber
\partial_t \Gamma_N [X] & =(\partial_t \Gamma_N) [X] = \partial_t
  \Big\{ 
\sup_J \big( \Tr (J\cdot X) -\mathcal W_N[J]\big)- (\Delta S_N) [ X]
 \Big\} \\ \nonumber
 &=  -\partial_t \mathcal W_N [J]-
\partial_t(\Delta S_N) [ X] 
\label{mittlere}
\\
&= -\frac1
{\mathcal Z_N[J]}\int (-\partial_t \Delta S_N) \ee^{-S-\Delta S_N+\Tr(J \cdot \varphi)} \dif \mu_\Lambda (\Phi) 
\\
& \reviewinequation{\quad- \frac12 \sum_{a,b,c,d=1}^\Lambda  \sum_{i=1}^n  
e_i (X_i)_{ab} (\partial_t R_N^\tau)_{ab;cd} (X_i)_{cd} \,.}\nonumber
 \end{align}
Recalling that  $X_i= \mathcal Z_N\inv  \partial^{J^i}\mathcal Z_N $
one can use 
\reviewinequation{
\begin{align*}\frac{\delta^2 \mathcal W_N[J]}{\delta J^i_{ba}\,\delta J^i_{dc}}& =
- \langleb (\varphi_i)_{ab} \rangleb  \langleb (\varphi_i )_{cd}\rangleb +\frac1{\mathcal Z_N[J]} 
 \frac{\delta^2 \mathcal Z_N[J]}{\delta J^i_{ba}\,\delta J^i_{dc}}\qquad \text{(no $i$ sum)}\\
 & = 
-(X_i)_{ab} (X_i)_{cd} 
\\ & 
+\frac1{\mathcal Z_N[J]} \int
 (\varphi_i)_{ab} (\varphi_i)_{cd} \cdot \ee^{-S-\Delta S_N+\Tr(J \cdot \varphi)} \dif \mu_\Lambda (\Phi)\,
\end{align*}}
in order to re-express $\partial_t(\Delta S_N)$ appearing in the integrand in 
the first term, 
$\mathcal Z_N[J]\inv \int (-\partial_t \Delta S_N) \ee^{-S-\Delta S_N+\Tr(J \cdot \varphi)} \dif \mu_\Lambda (\Phi) $, of Eq. \eqref{mittlere} to obtain 
\begin{align}\label{jakies_row}\partial_t \Gamma_N [X]  =  \frac12 \sum_{a,b,c,d}^\Lambda  \sum_{i=1}^n \bigg(\frac{\delta^2 \mathcal W_N[J]}{\delta J^i_{ba}\,\delta J^i_{dc}} \bigg) \cdot e_i\cdot   (\partial_t R_N^\tau)_{ab;cd}\,.  
 \end{align}
The rest relies on the use of the superspace chain rule  
\begin{align} \label{decisive}
 \delta_{ij}\delta_{ux}
\delta_{vy}
&=\dervfunc{ (X_i)_{uv}}{(X_j)_{xy}}
=\sum_{k=1}^n
\sum_{l,m}
\dervfunc{ (X_i)_{uv}}{J^k_{lm}}
\dervfunc{J^k_{lm}}{ (X_j)_{xy}} \\
&=\!\!\!\sum_{\substack{k=1,\ldots, n \\ l,m=1,\ldots, \Lambda }  } \!\!\! \nonumber
 \big\{ \partial^{X_j}_{yx} \partial^{X_k}_{lm} \Gamma_N[X] +e_k  \delta_{jk} (R_N^\tau)_{lm;yx}\big\}
\cdot 
\bigg( \dervfunc{^2 \mathcal W_N [J]}{ J^k_{lm}\delta J^i_{vu}} \bigg)
\,.
\end{align}
Passing from the first to the second line is implied by taking 
the derivative with respect to $X_j$ of the IR-regulated
quantum equation of motion, that is of
\begin{align*}
\partial^{X_k}_{ab}\Gamma_N &=  
\partial^{X_k}_{ab} \Big( 
\Tr (X\cdot J)-\mathcal W_N[J]
-\Delta S_N[X]\Big)
\\
&=J^k_{ab} + \Tr ( X \cdot \partial^{X_k}_{ab} J )
-\partial^{X_k}_{ab} \mathcal W_N[J] 
- e_k\Tr (   (R_N^\tau)_{ab; \cdot \cdots } X^k)\\
&=J^k_{ab}  
-e_k \Tr (  (R_N^\tau)_{ab;\cdot \cdots } X^k)
\,. 
\end{align*}
In the second line $\partial^{X_k}_{ab} J$ 
is a matrix (for fixed $a,b$) and the trace 
$\Tr ( X \cdot \partial^{X_k}_{ab} J )$, which equals $\partial^{X_k}_{ab} \mathcal W_N[J]$ by the chain rule, is taken
with respect to those tacit indices of $J$. 
In the other trace-term, the shown indices $a,b$ are excluded,
so traces are taken for the remaining ones (the dots in $R_N^\tau$); the symmetries \eqref{symmsofRN} of $R_N^\tau$ have been used too. Hence, indeed \begin{align*}
\dervfunc{ J^k_{pq}}{X^j_{xy}} &=\partial^{X_j}_{yx} \partial^{X_k}_{pq} \Gamma_N[X] +e_j  \delta_{jk} (R_N^\tau)_{pq;yx} \\
&= \big( {e_k^{\delta_{jk}}} \{ \Hess \Gamma_N[X]\}_{jk} + e_j \delta_{jk} R_N \big )_{px;yq} \,, 
\end{align*}
after \eeqref{twistedR} and the index symmetries implied by it. Denoting by $\cdot_n$ the product in the $\Mn$ tensor factor (of the superspace), 
one can moreover replace  $(\Hess^J   \mathcal W_N)_{ki}= \partial^{J^k} \partial^{J^i} \mathcal W_N[J]$ by the inverse\footnote{See discussion after the proof.} of 
\[\Hess  \Gamma_N + (\sigma \otimes  R_N   )  =  
\sigma \cdot_n (\Hess_\sigma  \Gamma_N +  1_n \otimes  R_N  ) 
\,, \]
 after using $\sigma=\diag (e_1,\ldots,e_n)$ and the fact that $1/e_i =e_i$ (since $e_i=\pm$).
 One has
 \begin{align} \label{intermediajakas}
 \partial_t \Gamma_N [X] &=  \frac12 \sum_{a,b,c,d}^\Lambda  \sum_{i=1}^n \bigg(e_i
 \frac{\delta^2 \mathcal W_N[J]}{\delta J^i_{ba}\,\delta J^i_{dc}} \bigg)    (\partial_t R_N^\tau )_{ab;cd}\\
 &=  \frac12 \sum_{a,b,c,d}^\Lambda  \sum_{i=1}^n 
 ( \Hess^J_\sigma \mathcal W_N[J])_{ii|cb;ad}  (\partial_t R_N^\tau )_{ab;cd} \,.\nonumber
\end{align}
The result follows from Eq. \eqref{jakies_row}, after realizing that 
  the LHS of \eqref{decisive} is $ \delta_{ij}\delta_{ux}
\delta_{vy}=(1_n\otimes 1 \totimes 1 )_{ij|yx;uv}=(\mathbf 1_\tau)_{ij|yx;uv}$. 
In order to invert\footnote{
One could feel tempted to state   
\[
 \big\{\partial^{X_j}_{ab}\partial^{X_k}_{xy} \Gamma_N[X]+ e_{j}\delta_{jk} (R_N)_{ab;xy}\big\}\inv \stackrel{!}{=}
 \partial^{J^k}_{yx} \partial^{J^i}_{cd} \mathcal W_N[J ] \,.
\]
Although this expression is probably clearer than \eeqref{rightproduct},
first one has to invert in superspace, and only then, take the matrix entries. 
} the Hessian of $\mathcal W$, we use 
\begin{align}\label{rightproduct} 
  & \{\Hess_\sigma \Gamma_N [X] + R_N \}_{ij\mid xb;ay} (\Hess^J_{\sigma}\mathcal W[J])_{jk\mid cx;yd }  &\\
 &=\{\Hess_\sigma \Gamma_N [X] + R_N \}_{ij\mid \tau(ab;xy)} (\Hess^J_{\sigma}\mathcal W[J])_{jk\mid \tau( yx;cd )} \nonumber \\ &=(1_n \otimes 1  \totimes 1 )_{ik \mid ab;cd} \, \nonumber 
\end{align} 
where the $\star$ product and the twisted Hessian can now be recognized. 
Therefore 
\begin{align}
  \partial_t \Gamma_N [X]  &= 
\frac12 
  \Tr_n \bigg\{\sum_{a,b,c,d}^\Lambda  \nonumber
  \big( [\Hess_\sigma^\tau \Gamma_N[X] +1_n\otimes R_N^\tau ]\inv \big)_{ab;cd} \\
  & \hspace{4.6cm}\times  (\partial_t R_N^\tau  )_{ab;cd}\bigg\}  \\ &
 = \frac12 \Tr_n\otimes \Tr_{\A_n}  \big\{ (\Hess_\sigma^\tau  \Gamma_N +   R_N^\tau   )\inv  \times   (\partial_t R_N^\tau  )  \big\}  
 \, . \nonumber  
 \end{align}
 We renamed indices and we used the symmetry $(R_N^\tau)_{ab;cd} = (R_N^\tau)_{ba;cd}$.
\end{proof}
%
 

The RHS of the FRGE is usually interpreted in terms of a ribbon loop \raisebox{-.29\height}{\includegraphics[width=20pt]{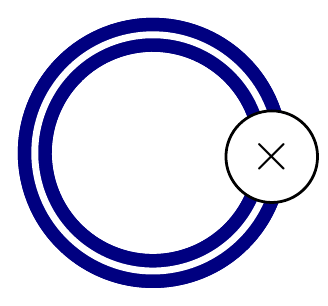}}, 
the thick ribbon being the full propagator.
 For the present FRGE
this picture is obtained by interpreting the ribbon as the 
supertrace $\Tr_n\otimes \TrL^{\otimes 2}$,
\begin{align}
 \partial_t\Gamma_N  = \!\! \raisebox{-.08\height}{\includegraphicsd{.398}{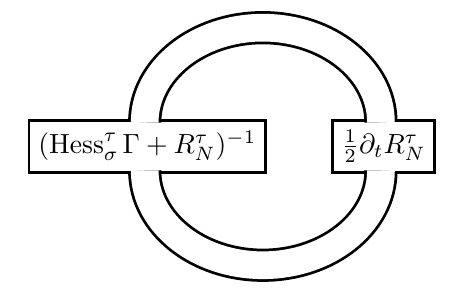}} 
\label{listones}
\end{align}
The source marked with a crossed circle is the RG-time derivative term.
In order to stress the meaning of the last equation,
we consider an ordinary Hermitian $n$-matrix model. 
Proposition \ref{thm:FRGE} then restricts to signature 
$(n,0)$, so each $e_i=1$, $i=1,\ldots,n$. 

\begin{cor}[FRGE for Hermitian multimatrix models]
Wetterich-Morris equation for Hermitian $n$-matrix models is given by
\begin{align}\label{FRGEHermitian}
\partial_t \Gamma_N[X]&=\frac{1}{2} \STr \bigg(   \frac{  \partial_t{R_N^\tau}}{\Hess^\tau  \Gamma_N[X] + R^\tau_N}\bigg)\,.
\end{align}

\end{cor}
\begin{proof}
 It is immediate from Proposition \ref{thm:FRGE}, 
since for Hermitian matrices one has $\sigma=1_n$.
\end{proof}


\section{Techniques to compute the renormalization group flow}\label{sec:FPandtruncations}
 
The next sections explain how to compute the RHS of the 
FRGE.

\subsection{Projection and truncations}  \label{sec:truncations}
\review{The RG-flow generates the infinitely many operators that 
the symmetries allow.} Feasibility forces us
first to \textit{project} each matrix $X_i$ to a $N\times N$ matrix $X_i\hp N$ and
then truncate $\Gamma_N[X\hp N]$ to Ans\"atze implying
finitely many operators $\mtc O_I$ indexed by words $I$
of the free algebra. Since 
this projection will be assumed for the rest of this paper,
for the sake of lightness we agree to write $X\hp N$ as $X$. 
Some truncation schemes are:
\begin{itemize}\setlength\itemsep{.4em}
\itemb Single trace truncation:
\[
\Gamma_N[X]= N \sum_I \bar{\mathsf{g}}_I(N)\TrN (\mathcal O_I(X))\,.
\]
\itemb Bi-tracial truncation:
\begin{align}
\Gamma_N[X]&=   N \sum_I\bar{\mathsf{g}}_I(N)\TrN (\mathcal O_I(X)) \\[-16pt]
&\,\,\,\,+ \sum_{I,I'} \bar{\mathsf{g}}_{I|I'}(N)(\TrN\otimes \TrN )(\overbrace{\mathcal O_I(X)\otimes \mathcal O_{I'}(X)}^{\mathcal O_{I|I'}(X)})\,.
\end{align}
\itemb Degree-$k$ truncation: 
\[
\Gamma_N[X]= 
  \sum_{\substack{ I_1,\ldots, I_{\alpha} \\ 
  \sum_\nu \mtr{deg } \mathcal{O}_{I_\nu}(X)\leq k}  } 
\frac{(\bar{\mathsf{g}}_{I_1|I_2|\ldots| I_\alpha})(N)}{N^{k-1}} 
 \TrN^{\otimes j}  \bigg( \bigotimes _{\nu=1}^\alpha \mathcal O_{ I_\nu} (X) \bigg)\,,
\]
\end{itemize}
where $\bar{\mathsf{g}}_{\ldots}(N)$ are the coupling constant, to be later 
renormalized to ${\mathsf{g}}_{\ldots}(N)$, the physical value.

We warn that this choice will be taken together 
with the assumption of $N$ being large. The price
to be paid is the un ability to recover the full 
effective action (which otherwise would be
obtained by $\lim _{N\to 1} \Gamma_N $) not only because $N$ is large,
but also because we compute in a projection. 
 
 \subsection{The $FP\inv$ expansion in the large-$N$ limit}\label{sec:FP}
 
 Based on the procedure introduced in \cite{EichhornKoslowskiFRG} for Hermitian 
 matrix models---which soon will be modified---we split the full propagator, for us $
 \Hess_\sigma\Gamma_N[X] + R_N =P \oplus F[X]$,
 into field-dependent and field-independent parts. 
 In our multimatrix case, with signs $\sigma=\diag(e_1,\ldots, e_n)$ 
 given by Eq. \eqref{esigns},  we get $ F[X]:=\Hess_\sigma \Gamma_N [X] -(\Hess_\sigma \Gamma_N\big|_{X=0}) $ and $P:= R_N +(\Hess_\sigma \Gamma_N\big|_{X=0})$.
We now simplify the treatment assuming that 
\[
 Z_i\equiv Z_j=:Z, \hphantom{a}\mbox{ when } e_i =e_j :=e \mbox{ for all $i,j$}, 
\]
for the rest of the paper.  This is not the most general case and particularly excludes for the time being mixed signatures left 
for later study; however, this simplification has the advantage of leading 
to a $P$ that is the identity matrix multiplied by a function $ \{1,\ldots,\Lambda\}^4\to \C$ denoted by the same letter, $P= (e^2 Z+R_N ) \mathbf 1=(Z+R_N ) \mathbf 1$ \review{since $e^2=1$. Notice that both $Z$ and $R_N$ being always positive $P$ is invertible.}
In particular, powers $P^\ell$ of $P$ are meant pointwise (not as a matrix or tensor). One therefore has the commutation 
of $P$ with the field part $F[X]$, 
\begin{subequations}\label{CommuteFP}
\begin{align}\label{CommuteFPa}
 P \times F[X] = F[X] \times  P\,,\hspace{25pt} \text{for all }X\in \M_N^{p,q}\,.
 \end{align}
It is important to realize in which sense the regulated Hessian 
of the interpolating action is an inverse of the Hessian of $\mathcal W_N$ in source space, as this defines the way we have to take the Neumann series to invert $\Hess_\sigma^\tau \Gamma + R_N^\tau$.
Although in the $\Mn$ factor of superspace this is an ordinary matrix 
product---see the groupoid property in the indices $i,j,k$ inside the proof 
of the FRGE, 
$\{\Hess_\sigma \Gamma_N [X] + R_N \}_{ij\mid xb;ay} (\Hess^J_{\sigma}\mathcal W[J])_{jk\mid cx;yd } =(1_n \otimes 1  \totimes 1 )_{ik \mid ab;cd}$---each entry of that 
matrix is multiplied according the product $\star$; this product is easier to recognize in \eeqref{withouttau}. 
That is to say, the way to invert in \ref{Wetterich} 
the regulated Hessian as dictated by the proof of the FRGE,
is the algebra $M_n(\A_n,\star )$ and not $M_n(\A_n,\times)$. 
The commutation \eeqref{CommuteFP} can be replaced 
by
\begin{align}\label{CommuteFPb}
 P_\tau \star  F_\tau[X] = F_\tau[X] \star   P_\tau\,,\hspace{25pt} \text{for all }X\in \M_N^{p,q}\,.
 \end{align}
 \end{subequations}
 since for $(\A_{n},\star) $ the unit is $1\totimes 1$ and $P_\tau$ can 
 be treated as a scalar function. We take the Neumann series of the twisted 
 version $(\Hess_\sigma^\tau \Gamma_N[X]+R_N^\tau)\inv$. Namely by \eeqref{CommuteFPb},
 \begin{align}\label{Neumann}
   \sum_{k=0}^\infty (-1)^k P_\tau\inv \star \big\{P^{-1}_\tau F_\tau[X] \big\}^{\star k} = \sum_{k=0}^\infty (-1)^k \big\{P_\tau^{-(k+1)} F_\tau[X] ^{\star k}\big\}\,.
   \end{align}
Underlying this structure is the independence of $P$ from the matrices $X=\{X_j\}$. Thus, when evaluated, $(P_{\tau})^\ell$ sits in the constant part of $\A_{n,\Lambda}$, so powers of $P_\tau$ act on the field part by scalar multiplication. On the other hand,
$(F_\tau[X]) ^{\star k}$ does mean the matrix product in 
the field part \eqref{fieldpart} of $\A_{n,\Lambda}$.
Then, using the associativity of $\star$ (Prop. \ref{thm:associativity}), it is routine to check 
that the series \eqref{Neumann} serves as
inverse of $P_\tau \oplus F_\tau [X]$ in the 
sense that their product in either order yields $\mathbf 1_\tau=1_n\otimes 1_\Lambda \totimes 1_\Lambda$. 
 Therefore,
\begin{align} 
\frac1{\Hess_\sigma^\tau \Gamma_N[X]+R_N^\tau}&=\sum_{k=0}^\infty (-1)^k \big[ P_\tau^{-(k+1)} ( F_\tau  ^{\star k}) \big]\,.
\label{aseio}
\end{align}
Assuming a truncation necessitates a 
compatible supertrace, $\STrN$.
Since functions $G:\{1,\ldots,\Lambda \}^4 \to \C$ 
act multiplicatively on the fields, we let
 \begin{align}
  \label{STr2}
\STrN \big(G \cdot  W[X]\big) = \bigg( \sum_{a,b,c,d=1 }^\Lambda G_{ab;cd} \bigg)\cdot (\Tr_n\otimes \TrN \otimes \TrN ) \big(W_N[X]\big)\,
\end{align}
for $W$ a field ($\deg W\neq 0$) in $\Mn\otimes \A_{n,\Lambda}$.
Here, $W_N$ is the same matrix of words $W$ projected to $\Mn\otimes  \A_{n,N}$.
 Also, $\STr$ is defined to be identically zero 
 on the `constants' of the free algebra (in the terminology of Sec. \ref{sec:toolkit}), or 
 \begin{align}\label{STrdeg0}
\STrN (L)=0 \text{ if } L\in \C\cdot(1_n\otimes 1_N\otimes 1_N) \mbox { or } L\in \C\cdot(1_n\otimes 1_N\totimes 1_N)
\,. 
\end{align} 
This follows from any of the previous Ans\"atze for $\Gamma_N$, 
but it holds in general on physical grounds, 
 since that constant part in the action 
 corresponds to the vacuum energy \cite{Morris}. 
 However, the constant part of the algebra cannot be
fully ignored since is the one that regulates the RG-flow 
and that part appears multiplying the fields. \\
\begin{remark}
 It would be interesting to answer whether the vanishing of $\STrN (L)$ 
 (here and in the physics literature, as part of the definition) 
 yields constraints on the IR-regulator. Namely, to explore the conditions
 that the equation 
 $\STrN (P\inv 1_N\otimes 1_N)=0$ imposes on $R_N$, if 
 one does not automatically include in the definition 
 the condition \eqref{STrdeg0}.  
\end{remark}

\begin{proposition}\label{thm:tadpole}
The RG-flow is generated by the noncommutative Laplacian scaled by 
$\varrho:= \sum_{a,b,c,d} (\partial_t R_N \cdot  P^{-2})_{ab;cd}$. That is, in the `tadpole approximation', the FRGE is given 
by  \vspace{-2pt}\begin{align}\label{tadpole}
\partial_t \Gamma_N[X] = \displaystyle
-\frac12  \varrho \TrN\otimes \TrN \big(\nabla^2 \Gamma_N \big)\,.
\end{align}
\end{proposition}

\begin{proof}
The tadpole approximation means to cut
Eq. \eqref{aseio} to $k=1$. It is immediate 
that one can undo the twists from the Hessian and $R_N^\tau$ altogether, with that of $\partial_t R_N^\tau$ since in this simple case $\star$ is not implied. By Eq. \eqref{STr2} this means that 
\begin{align*}\partial_t \Gamma_N[X]&=+
\frac12
\STrN \bigg\{\sum_i  \nonumber 
\frac{ \partial_t R_N}{\Hess_\sigma  \Gamma_N  +   R_N } \bigg\} \\
&=- \frac12 \bigg\{\sum_{a,b,c,d} (\partial_t R_N \cdot  P^{-2})_{ab;cd} \bigg\} \Tr_n\otimes \TrN\otimes \TrN  
\big(  F[X] \big)  \\
 &=-
 \frac12 \varrho \Tr_n\otimes \TrN\otimes \TrN\big\{   F[X] +  F[0]  \big\} 
\end{align*}
were Eq. \eqref{STrdeg0} has been used from the first to the second line,
and from there to the third too.  Now, $ F[X] +  F[0] = \Hess_\sigma \Gamma_N[X] $,
which traced over the first $\Mn$ factor, is by definition the NC-Laplacian. 
\end{proof}
 We next justify the approximation given in eqs. \eqref{STr2}--\eqref{STrdeg0} and relate it with  
 the definition of $\STr$. Notice that the support of the function $G\hp{N}_k:\{1,\ldots ,\Lambda\} ^ 4 \to \re$ given by
 $G\hp{N}_k=(\partial_t R_N) \cdot P^{-(k+1)}$ becomes an $N$-dependent region of $\{1,\ldots ,\Lambda\} ^ 4$. Generally, one cannot find a function $f_n(N)$ such that $\STr(G_k\hp N\cdot W[X])=f_k(N) \cdot \Tr_n\otimes\Tr_{\A_{n,N}}(W_N[X])$, or explicitly such that  
 \[
\sum_{a,b,c,d=1}^\Lambda [G\hp{N}_k]_{ba;dc}  (W[X])_{ab;cd} \stackrel{}{=}
f_k(N)  \Tr_n\otimes \TrN \otimes \TrN( W_{N}[X]^k)
\] 
holds 
for a $W[X] \in   M_n(\A_{n,\Lambda}) $ in the field part of the free algebra, 
with $W_N[X] \in  M_n(\A_{n,N})$. 
What is done in practice is to assume this replacement,
but in return to let the function $f_k(N)$ be governed by the FRGE.  
We moreover use a regulator $R_N$ whose support 
is inside $\{1,\ldots ,N\} ^ 4$. 

In order to exploit the FRGE, one needs to 
compute the first powers of the 
expansion \eqref{Neumann}. Defining $\tilde h_k(N)=\sum_{a,b,c,d}^\Lambda (G\hp{N}_k)_{ab;cd}$,
which, since neither  $\partial_t R_N$ nor $P^{-(k+1)}$ 
have field dependence, equals 
\begin{align} \label{barh}
\tilde h_k(N)=\sum_{a,b,c,d=1}^\Lambda (\partial_t R_N)_{ab;cd}  P^{-(k+1)}_{ab;cd}\,, 
\end{align}
one obtains after projecting
\begin{align} \label{FPexpansion} 
\partial_t \Gamma_N [X] 
&\stackrel{\eqref{Wetterich}}{=}\frac12
  \STrN^\tau \bigg( \sum_{k=0}^\infty  (-1)^k G\hp{N}_k  
  \cdot 
  \{ F_\tau[X]\}^{\star k} \bigg) \\
 &\!\!\!\stackrel{\eqref{STr2} \& \eqref{STrdeg0}}{=}\,\,\, \frac12\sum_{n=1}^\infty  (-1)^k \tilde h_k(N) 
   (\Tr_n\otimes\TrN^{\totimes 2} )\big\{  F_\tau[X]\big\}^{\star k} \nonumber\\
  & \nonumber \,\, \,\,\,\,\,= \frac12 (\Tr_n\otimes\TrN^{\totimes 2}) \big\{ - \tilde h_1(N)  F_\tau [X] \\ & \hspace{3.4cm} +\tilde h_2(N)   \big(F_\tau [X]\big)^{\star 2} +\ldots \big\} \,.\nonumber
 \end{align}
where $\TrN \totimes \TrN (\mathcal Q)= \Tr_{\A_n}((1_N\totimes 1_N) \times \mathcal Q)$ 
in terms of which we $\STrN^\tau $. That twist comes from $R_N^\tau$,
whose untwisted part was absorbed in the functions $G\hp{N}_k$. 
 We remark that Eq. \eqref{STr2} does not take into account 
 the symmetry breaking caused by the regulator $R_N$, 
 which  is related to ignoring the modified Ward-Takahashi\footnote{Regarding the Ward-Takahashi identity \cite{fullward,fullward_add} of tensor 
 models, a sister theory of matrix models, the progress of the WTI-constrained RG-flow is reviewed in \cite{Baloitcha:2020idd}. See also 
\cite{Lahoche:2019ocf}.} identity \cite{LitimPawlowski} caused by $R_N$. \par
 \vspace{.2cm}
 
 \hspace{0.32cm} 
\fbox{\begin{minipage}{.75\textwidth}
From this point on, we focus on large-$N$ results 
and consider the fields as projected matrices of size $N\times N$. 
Terms of order $O(N\inv)$ will be often ignored in our computations.  
Also, since $F$ is not needed again,
we rename $F_\tau$ to $F$. 
\end{minipage}}
  
  \section{``Coordinate-free'' matrix models} \label{sec:CrossCheck}
 
  We cross-check that, notwithstanding the somewhat different statements,  our purely-algebraic 
  approach yields, for the Hermitian case with $n=1$, the results that \cite{EichhornKoslowskiFRG} presented in ``coordinates'' 
  (that is, written with matrix entries). Here, we also calibrate
the IR-regulator for later use in Section \ref{sec:2MM}.

  The interpolating action $\Gamma_N[X]$ is given by 
 (applying $\TrN^{\otimes 2}$ to) the next operators 
 that define our truncation: 
 \allowdisplaybreaks[1]
  \begin{align*}
&
  \frac{Z }{2 N}1_N\otimes X^2 +\frac{\bar{\mathsf{g}}_{4} }{4
   N}1_N\otimes X^4 +\frac{\bar{\mathsf{g}}_{6}}{6 N}1_N\otimes X^6\\ 
   &+\frac{\bar{\mathsf{g}}_{2|2}}{8}  X^2\otimes X^2+\frac{\bar{\mathsf{g}}_{2|4}}{8}  X^2\otimes X^4\,.
 \end{align*}
 Since $n=1$,  the NC-Laplacian equals the NC-Hessian $\partial^2$, which
on $\Tr \mathcal O$ for an operator $\mathcal O\in \Cfree 1 $ equals $(\partial \circ \Day)  \mathcal O(X)$ by Claim \ref{thm:indexfrei}. So, by Claim \ref{thm:NCLap} and Eq. \eqref{NCLap_prod} one gets
\begin{align*}
 \frac{1}{2N} \partial^2  \Tr_{\A_{1,N}}\big(  {1_N\otimes X^2} \big) &=  1_N\otimes 1_N \\
 \frac1{4N} \partial^2   \Tr_{\A_{1,N}}\big( {1_N\otimes X^4}\big)  &=  X\otimes X+ 1_N\otimes X^2+ X^2\otimes 1_N \\
 \frac{1}{8}\partial^2   \Tr_{\A_{1,N}}\big( X^2\otimes X^2 \big)&=   X\totimes X+ 1_N\otimes 1_N \TrN \Big(\frac{X^2}{2}\Big) \\
 \frac{1}{6N}\partial^2   \Tr_{\A_{1,N}}\big({1_N\otimes X^6} \big)&=  X\otimes X^3 + 1_N\otimes
   X^4+ X^2\otimes X^2 \\ & + X^3\otimes X+ X^4\otimes 1_N \\
 \frac1{8} \partial^2  \Tr_{\A_{1,N}}\big( X^2\otimes X^4\big) &=  X\totimes X^3+X^3\totimes X+  1_N\otimes 1_N
   \TrN\Big(\frac{X^4}{4}\Big) \\ &  +\big\{ {X^2}\otimes 1_N+ X\otimes
   X+  1_N\otimes {X^2} \big\} \TrN 
   \Big(\frac{X^2}{2} \Big)\,.
\end{align*} %
\allowdisplaybreaks[0]%
One now ``twists'' these equations. The expression for $F[X]=\Hess^\tau \Gamma [X] - Z (1_N\totimes 1_N )$ follows from the first equation in this list (after exchange of the tensor product with the twisted version). We keep odd-degree operators in $F$, even if 
we first included even-degree ones, since 
we need powers of $F$ and even-degree operators are generated from odd-degree ones. 

 By neglecting odd-degree after taking the $\star$-powers of $F[X]$,
 as well as truncating them to degree-six operators, 
 the $FP\inv$ expansion \eqref{FPexpansion} in 
 this setting reads:
\begin{salign}
\partial_t \Gamma_N[X]=&-\frac{1}{2} \frac{\tilde{h}_1}{N^2} 
\Big\{ 
(N^2+2) \bar{\mathsf{g}}_{2|2} + 4 N \bar{\mathsf{g}}_{4}
\Big\}
\TrNX{2} \\ 
& +  \Big\{-\frac{\tilde{h}_1}{N^2} \Big( 
(4+\frac{N^2}{2})\bar{\mathsf{g}}_{2|4}+4 N \bar{\mathsf{g}}_{6} 
\Big)\\ & \quad +\frac{\tilde{h}_2}{N^2} \big(
12\bar{\mathsf{g}}_{2|2} \bar{\mathsf{g}}_4 + 4 N \bar{\mathsf{g}}_4^2\big)
\Big\}\TrNX{4}  \\  &+  \numerada \label{1MMGammadot} \Big\{
\frac{\tilde{h}_2}{N^2} \big((8 + N^2) \bar{\mathsf{g}}_{2|2}^2 + 8 N \bar{\mathsf{g}}_{2|2} \bar{\mathsf{g}}_{4}  + 12 \bar{\mathsf{g}}_{4}^2\big) \\ & \quad - \frac{\tilde{h}_1}{N^2} (4 N \bar{\mathsf{g}}_{2|4} + 4 \bar{\mathsf{g}}_6)
\Big\} \frac{1}{8}\TrN^2(X^2)\\  
& + \Big\{
 \frac{\tilde{h}_2}{N^2} (36 \bar{\mathsf{g}}_{2|4} \bar{\mathsf{g}}_4 + 30 \bar{\mathsf{g}}_{2|2} \bar{\mathsf{g}}_6 + 12 N \bar{\mathsf{g}}_4 \bar{\mathsf{g}}_6) \\ &
 \quad - \frac{\tilde{h}_3}{N^2} (81 \bar{\mathsf{g}}_{2|2} \bar{\mathsf{g}}_4^2 + 6 N \bar{\mathsf{g}}_4^3)
\Big\}\TrNX{6}
\\  &+\Big\{
  \frac{\tilde{h}_2}{N^2} (\bar{\mathsf{g}}_{2|4} ((38 + N^2) \bar{\mathsf{g}}_{2|2} + 12 N \bar{\mathsf{g}}_4) + 8 N \bar{\mathsf{g}}_{2|2} \bar{\mathsf{g}}_6 + 48 \bar{\mathsf{g}}_4 \bar{\mathsf{g}}_6) \\ &
  \quad-\frac{\tilde{h}_3}{N^2} (72 \bar{\mathsf{g}}_{2|2}^2 \bar{\mathsf{g}}_4 + 12 N \bar{\mathsf{g}}_{2|2} \bar{\mathsf{g}}_4^2 + 48 \bar{\mathsf{g}}_4^3) 
\Big\}\TrNX{2} \TrNX{4}\,,
\end{salign}%
up to the third non-trivial term ($h_r=0$ for $r\geq 4$) in the $FP\inv$-expansion. This equation was 
obtained using the product rules of Proposition \ref{productrules}: 
For instance, the cubic term in $\bar{\mathsf{g}}_4$ in the fifth line of  \eqref{1MMGammadot}
comes from $P^{-4}F^{\star 3}$, more concretely from   \[-(\tilde h_3/2N^2)\bar{\mathsf{g}}_4^3 \TrN^{\totimes 2} \big[  
(X^2 \totimes 1_N)^{\star 3} +( 1_N \totimes X^2 )^{\star 3}+ \ldots \big] \,,\] 
where the dots omit other terms in the cube of $F$. Graphically, the $\bar{\mathsf{g}}_4^3$-contribution to $\bar{\mathsf{g}}_6$ is (cf. Eq. \eqref{listones} too)
\begin{align}\label{g43}
\bar{\mathsf{g}}_4^3 \sim  \includegraphicsd{.5}{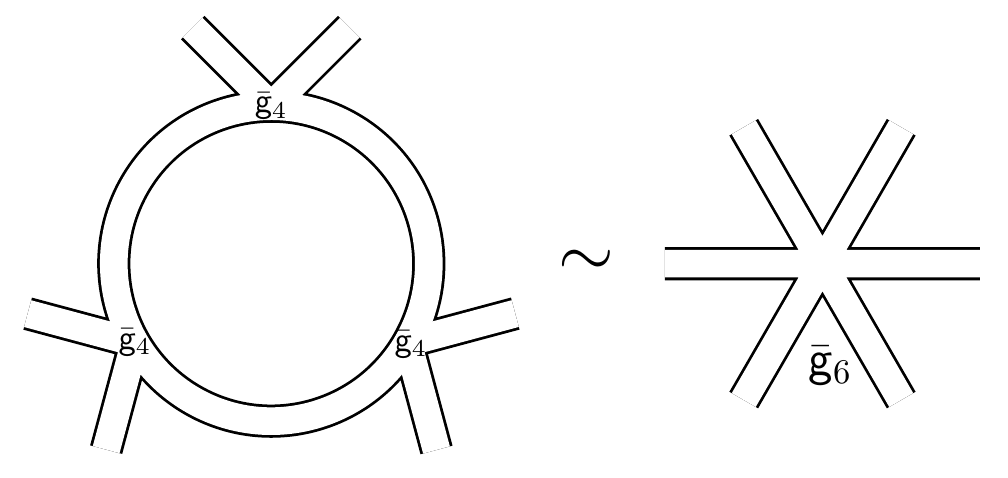}
\end{align}
We let $h_k=\lim_{N\to \infty }Z^k\tilde h_k(N)/N^2 $, which due to 
Eq. \eqref{barh} is independent of $Z$, and choose later an explicit regulator $R_N$  that 
makes $h_k$ only dependent on $k$ in the large-$N$ limit. Thereafter, the  contributions to the $\beta$-functions coming from 
\textit{quantum fluctuations}\footnote{These are the coefficients
of a $\partial_t \Gamma_N[X]$ in the operator in question.} can be read off from Eq. \eqref{1MMGammadot}. 
To state the quantum fluctuations in terms of the renormalized
quantities (without bar), one needs to find the way these scale with $Z$ and $N$.  
We let $\bar{\mathsf{g}}_{2k}=Z^{a_k} N^{-b_k} \mathsf{g}_{2k}$ 
and $\bar{\mathsf{g}}_{u|2k-u}=Z^{j_k} N^{-i_k} \mathsf{g}_{u|2k-u}$ (for even $u$, 
with $0<u<2k $).\par 

To solve for $a_k,b_k,i_k,j_k$, one asks 
the equation $\beta_{I} = \partial _t \mathsf{g}_{I} $ to remain finite
for each
operator $\mathcal O_I$
as $N\to \infty$. This leads to
\begin{align*}
\bar{\mathsf{g}}_{4}&=Z^{2} N^{-1} \mathsf{g}_{4}, &&& \bar{\mathsf{g}}_6&=Z^{3} N^{-2} \mathsf{g}_6 , \\
 \bar{\mathsf{g}}_{2|2}&=Z^2 N^{-2}\mathsf{g}_{2|2}\,,&& \mtr{and}&  \bar{\mathsf{g}}_{2|4}&=Z^3 N^{-3}\mathsf{g}_{2|4} \,.\end{align*}
These scalings, together with the 
quantum fluctuations from Eq. \eqref{1MMGammadot}, yield for the anomalous dimension $\eta=-\partial_t \log Z$
and the $\beta$-functions in the large-$N$ limit:
\begin{subequations}
\begin{align} 
 \eta & =   h_{1}\Big(\frac{1}{2} \mathsf{g}_{2|2}+2 \mathsf{g}_{4}\Big) \,,  \label{1MMeta}   \\ 
 \beta_{{4}} &= (1+2\eta )\mathsf{g}_4   + 4 h_{2} \mathsf{g}_{4}^2-h_1 \Big(4 \mathsf{g}_{6} +\frac{\mathsf{g}_{2|4} }{2}\Big)\,, \label{1MMbeta4} \\ 
 \beta_{{2|2}}&=  (2+ 2\eta  )\mathsf{g}_{2|2} -4h_1(\mathsf g_{2|4} +\mathsf g_6)  \\ &\quad \nonumber+ h_2 (\mathsf g_{2|2}^2 + 8 \mathsf g_{2|2} \mathsf g_4 + 12 \mathsf g_4^2) \,,\\
 \beta_{{6}}&= (2+3\eta )\mathsf{g}_{6} + 12 \mathsf{g}_{4} \mathsf{g}_{6} h_{2}-6 \mathsf{g}_{4}^3 h_{3}\,, \phantom{\frac12}\\ 
 \beta_{{2|4}}&=(3+3\eta )\mathsf{g}_{2|4}+h_{2} (\mathsf{g}_{2|2} \mathsf{g}_{2|4}+8 \mathsf{g}_{2|2} \mathsf{g}_{6}\\ & \quad +12 \mathsf{g}_{2|4} \mathsf{g}_{4}+48 \mathsf{g}_{4} \mathsf{g}_{6}) -h_{3} \big(12 \mathsf{g}_{2|2}\mathsf{g}_{4}^2+48 \mathsf{g}_{4}^3\big)\,.\phantom{\frac12} \nonumber
 \end{align}\label{1MMbeta24}%
\end{subequations}%
We only are in debt with the explicit regulator $(R_N^\tau)_{ab;cd}=r_N(a,c)\delta_{a}^d\delta_{c}^b$ 
for $r_N$ defined on $\{1,\ldots, \Lambda\}^2$ and given by 
\begin{align}r_N(a,b)= Z\cdot    \bigg[\frac{N^2}{ a^2+b^2} -1 \bigg] \cdot \Theta_{\mathbb D_N}(a,b) \,,
\label{regulator}
\end{align} 
being $\Theta_{\mathbb D_N}(a,b)$  the indicator function in the 
disc $a^{2}+b^{2}\leq N^{2}$. 

\begin{figure}[h!]
 \includegraphics[width=.52\textwidth]{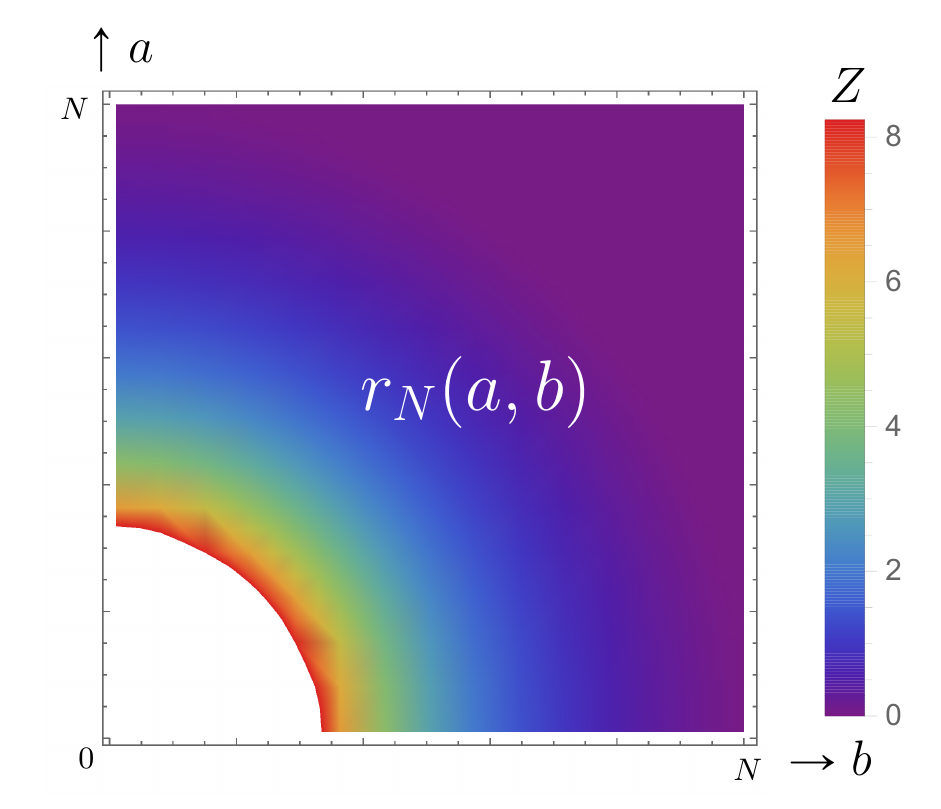}
\caption{ 
The plot shows the support of the chosen IR-regulator $r_N(a,b)$,
contained in the square $\re^+_{\leq N} \times  \re^+_{\leq N} $. 
(The white quarter of disk means a truncation of the graph
around the origin.)
\label{fig:RN}}
\end{figure} 

It turns out that for this regulator, $Z^k\tilde h_k/N^2$ indeed converges to a number $h_k$ independent of $N$, when this parameter is large. The first values are
in fact \begin{align}\label{hvalues}
h_1=\frac{\pi}{24}  (6-5 \eta ),\quad  
h_2=\frac{\pi}{48}   (8-7 \eta ),\quad 
h_3=\frac{\pi}{80}  (10-9 \eta )\,.
\end{align}
Inserting the four fixed point equations, i.e.  $\beta_{g_I^\fxpt}|_{\eta^\fxpt=\eta(\mathsf g^\fxpt)}=0$ for $I=2,4,2|2 $ and $2|4$,
one finds, on top of the Gau\ss ian trivial fixed point ($\mathsf{g}_I=0$ for each $I$),
several fixed points, tagged here with a little black diamond. The interesting one to be reproduced is expected be $-1/12$, the critical value of $\mathsf{g}_4$ 
for gravity coupled to conformal matter \cite{dFGZ}.
The latter has been identified in \cite{EichhornKoslowskiFRG}, who report $\mathsf{g}^{\fxpt}_{4}|_{\text{\tiny [EK13]}}=-0.056$ using the very same truncation\footnote{The same
authors report the possibility to obtain the exact solution in \cite{TowardsEK} by imposing it
and then solving for the regulator (in the tadpole approximation); but our aim here is to compare regulators in the same truncation.}. In contrast, we get 
\begin{align}\nonumber
\eta^\fxpt & = -0.2494, &&& \mathsf{g}_{4}^\fxpt&= -0.08791, &&&  \mathsf{g}_{2|2}^\fxpt &= -0.17415 \,,& \\ 
\label{1MMfixpt}
\mathsf{g}_{6} ^\fxpt &=
-0.003386, &&&  \mathsf{g}_{2|4}^\fxpt & =-0.02423\,. &&&  &
\end{align}%
This fixed point, obtained with the IR-regulator $r_N$ of Eq. \eqref{regulator}
gets far closer ($\mathsf{g}_{4}^\fxpt= -0.08791$) to the exact value $\mathsf{g}_{\mtr c}=-1/12=-0.083\bar3 $, which suggests that we should stick to our $r_N$ for the two-matrix models treated next.
%

\section{Two-matrix models from noncommutative geometries}\label{sec:2MM}

\subsection{Theory space} \label{sec:2MM_theoryspace}
The conventions for the coupling constants are the following, with numerical factors incorporated later. 
For \[  n_1,\ldots, n_{2t}, l_2,\ldots,k_{2t-1}, l_{2t-1} \in \Z_{> 0} \and  l_1, l_{2t},k_1, k_{2t} \in \Z_{\geq 0}\,,\]
we associate with each operator the following coupling constants:
\begin{subequations}
\begin{align}
 \ac_{2k}  &\leftrightarrow  A^{2k}   \qquad (k\geq 2 ) \\  
 \bc_{2k} &\leftrightarrow  B^{2k} \qquad (k \geq 2 )  \\ 
\cc_{n_1n_2\cdots n_{2t}}  & \leftrightarrow \underbrace{A^{n_1}B^{n_2}\cdots A^{n_{2t-1}}  B^{n_{2t}}}  \\
\dc_{l_1l_2\cdots l_{2s}| k_1k_2\cdots k_{2t}}  & \leftrightarrow \overbrace{A^{l_1}B^{l_2}\cdots B^{l_{2s}}}^{\quad AB-\text{alternating}} \otimes \overbrace{A^{k_1}B^{k_2}\cdots B^{k_{2t}}}^{AB-\text{alternating}} 
 \, 
\end{align}
\end{subequations}
Notice the alternating convention in the letters.
For coupling constants of type $\cc$ and $\dc$  (mnemonics: `combined' and `disconnected')
some care is needed. Operators can always begin with the highest power of $A$, which 
for $\cc$ is never zero---otherwise the respective operator is a pure power
of either $A$ or of $B$---in order 
to reduce the number of constants. This is due to 
the possibility to cyclicly reorder ($\sim$) the words, as these appear inside a trace.
Only the first and last parameters can be zero for $\dc$-constants. 
In order to include an odd number of powers of the letters,
the last integer is allowed to be zero. If this is so, we agree to omit the 
rightmost zero. \par 
Both conventions are illustrated with $ABA\otimes BAB \sim ABAB^0\otimes AB^2 $, whose coupling constant is 
$\dc_{1110|12}=\dc_{111|12}$. On the other hand a leftmost zero 
is important: from the definition $\dc_{l_1l_2\cdots l_t|I } \neq 
\dc_{0l_1 l_2\cdots l_t| I}$, since $A^{l_1}B^{l_2} \cdots B^{l_{2t}} \neq A^{0} B^{l_1} A^{l_2} \cdots A^{l_{2t}}  $.
Notice that $\dc$ has to satisfy a symmetry condition: $\dc_{I|I'}=\dc_{I'|I}$ for any integer multi-indices $I,I'$ 
(since the respective operators do), so we only keep one of the two.  \par 
As before, a bar on a coupling constant,  $\acb,\ldots,\dcb$, denotes its unrenormalized value,
whose $N$-dependence we do not show, for the sake of keeping the notation compact.  

\scriptsize
\begin{table}  \normalsize 
 \[ 
\begin{array}{cc|cc} \text{\small\textsc{Degree}} &
\text{\small\textsc{Operators}} & \text{ \textsc{\small Coupling constant}} &\text{ \textsc{\small Scalings}} \normalsize\\[6pt] 
\text{\small\textsc{Quadratic}} \normalsize  & \normalsize   1_N\otimes (A\cdot A) &  \frac12 \Za \ea  &  *   \\[3pt]
& 1_N\otimes (B\cdot B) & \frac12 \Zb \eb  & * \\[3pt]  \normalsize 
& A\otimes A & \frac12 \dcb_{1|1} & 1/N \\[3pt]  
& B\otimes B & \frac12\dcb_{01|01} & 1/N \\[8pt]
\text{\small\textsc{Quartic}}   \normalsize 
&  \normalsize  1_N\otimes (A\cdot A\cdot A\cdot A) & \frac{1}{4}\acb_4 & 1/N \\[3pt]
& 1_N\otimes (B\cdot B\cdot B\cdot B) & \frac{1}{4}\bcb_4  & 1/N\\[3pt]
& 1_N\otimes (A\cdot A\cdot B\cdot B) & \ccb_{22} \ea \eb  & 1/N\\[3pt]
&1_N\otimes (A\cdot B\cdot A\cdot B) & -\frac{1}{2} \ccb_{1111} \ea \eb & 1/N \\[3pt]
  &(A\cdot B)\otimes (A\cdot B) & \dcb_{11|11}  & 1/N^2\\[3pt]
  &(A\cdot A)\otimes (B\cdot B) & 2 \dcb_{2|02} \ea \eb& 1/N^2\\[3pt]
  &A\otimes (A\cdot A\cdot A) & \dcb_{1|3} \ea& 1/N^2\\[3pt]
  &A\otimes (A\cdot B\cdot B) & \dcb_{1|12} \eb& 1/N^2 \\[3pt]
  &B\otimes (A\cdot A\cdot B) & \dcb_{01|21} \ea& 1/N^2\\[3pt]
  &B\otimes (B\cdot B\cdot B) & \dcb_{01|03} \eb& 1/N^2 \\[3pt]
  &(A\cdot A)\otimes (A\cdot A) & 3 \dcb_{2|2}& 1/N^2\\[3pt]
  &(B\cdot B)\otimes (B\cdot B) & 3 \dcb_{02|02}& 1/N^2\\[3pt]
  \end{array}\]
  \caption{Quadratic and quartic operators and their coupling constants. Notice that $A\otimes B$ and $1_N\otimes A\cdot B$ (the latter appearing in the Ising-2-matrix model) are forbidden. The scalings corresponding to the quadratic connected operators are 
  in the wave function renormalization $\protect\Za{},\protect\Zb{}$, which are in each case determined by the RG-flow. \label{tab:QQops}} 
\end{table}
\begin{table} \normalsize
\[ 
 \begin{array}{cc|cll}
 &
 \text{\small\textsc{Sextic Operators}} & \text{\small \textsc{NCG coefficient}} &\text{\small \textsc{Coupling}}    & \text{\small\textsc{Scalings}}   
  \\[-1pt] & & \small \textsc{value} &  \small\textsc{constant}  & 
 \fontsize{11.1}{14.9}\selectfont  \\[6pt] \fontsize{11.1}{14.9}\selectfont 
 &1_N\otimes (A\cdot A\cdot A\cdot A\cdot A\cdot A) & \ea   &    \acb_6   & 1/N^2 \\[3pt] &
1_N\otimes (A\cdot A\cdot A\cdot A\cdot B\cdot B) & 6  \eb  &    \ccb_{42} & 1/N^2    \\[3pt] &
1_N\otimes (A\cdot A\cdot A\cdot B\cdot A\cdot B) & -6\eb    &   \ccb_{3111} & 1/N^2  \\[3pt] &
1_N\otimes (A\cdot A\cdot B\cdot A\cdot A\cdot B) & 3\eb   &   \ccb_{2121} & 1/N^2   \\[3pt] &
1_N\otimes (B\cdot B\cdot B\cdot B\cdot B\cdot B) & \eb    &     \bcb_6  & 1/N^2 \\[3pt] &
1_N\otimes (A\cdot A\cdot B\cdot B\cdot B\cdot B) & 6  \ea &   \ccb_{24} & 1/N^2  \\[3pt] &
1_N\otimes (A\cdot B\cdot B\cdot B\cdot A\cdot B) & -6 \ea &   \ccb_{1311}  & 1/N^2 \\[3pt] &
1_N\otimes (A\cdot B\cdot B\cdot A\cdot B\cdot B) & 3 \ea  &    \ccb_{1212} & 1/N^2  \\[3pt] &
 A\otimes (A\cdot A\cdot A\cdot A\cdot A) & 2 &     \dcb_{1|5} & 1/N^3 \\[3pt] &
 A\otimes (A\cdot B\cdot B\cdot B\cdot B) & 2  &  \dcb_{1|14} & 1/N^3  \\[3pt] &
 A\otimes (A\cdot A\cdot A\cdot B\cdot B) & 6 \ea \eb  &  \dcb_{1|32}  & 1/N^3  \\[3pt] &
 A\otimes (A\cdot A\cdot B\cdot A\cdot B) & -2\ea \eb  &   \dcb_{1|2111} & 1/N^3   \\[3pt] &
 B\otimes (A\cdot A\cdot A\cdot A\cdot B) & 2&  \dcb_{01|41}  & 1/N^3 \\[3pt] &
 B\otimes (A\cdot A\cdot B\cdot B\cdot B) & 6\ea \eb   &    \dcb_{1|23} & 1/N^3  \\[3pt] &
 B\otimes (A\cdot B\cdot B\cdot A\cdot B) & -2 \ea \eb    &  \dcb_{01|1211}   & 1/N^3  \\[3pt] &
 B\otimes (B\cdot B\cdot B\cdot B\cdot B) & 2  &  \dcb_{01|05}   & 1/N^3  \\[3pt] &
 (A\cdot B)\otimes (A\cdot A\cdot A\cdot B) & 8 \ea  &  \dcb_{11|31}   & 1/N^3  \\[3pt] &
 (A\cdot B)\otimes (A\cdot B\cdot B\cdot B) & 8 \eb   &    \dcb_{11|13} & 1/N^3 \\[3pt] &
 (A\cdot A)\otimes (A\cdot A\cdot B\cdot B) & 8  \eb   &     \dcb_{2|22} & 1/N^3  \\[3pt] &
 (A\cdot A)\otimes (A\cdot B\cdot A \cdot B) & -2 \eb  &  \dcb_{2|1111}  & 1/N^3  \\[3pt] &
 (A\cdot A)\otimes (A\cdot A\cdot A\cdot A) & 5 \ea &  \dcb_{2|4}   & 1/N^3 \\[3pt] &
 (A\cdot A)\otimes (B\cdot B\cdot B\cdot B) &\ea  &    \dcb_{2|04} & 1/N^3\\[3pt] &
 (B\cdot B)\otimes (A\cdot A\cdot B\cdot B) & 8  \ea &  \dcb_{02|22}   & 1/N^3 \\[3pt] &
 (B\cdot B)\otimes (A\cdot B\cdot A\cdot B) & -2\ea   &   \dcb_{02|1111}  & 1/N^3   \\[3pt] &
 (B\cdot B)\otimes (B\cdot B\cdot B\cdot B) & 5  \eb  &    \dcb_{02|04} & 1/N^3   \\[3pt] &
 (B\cdot B)\otimes (A\cdot A\cdot A\cdot A) & \eb   &    \dcb_{02|4} & 1/N^3\\[3pt] &
 (A\cdot A\cdot A)\otimes (A\cdot A\cdot A) & \frac{10}{3}  &   \dcb_{3|3} & 1/N^3 \\[3pt] &
 (A\cdot B\cdot B)\otimes (A\cdot A\cdot A) & 4 \ea \eb   &  \dcb_{12|3} & 1/N^3  \\[3pt] &
 (A\cdot A\cdot B)\otimes (A\cdot A\cdot B) & 6  &  \dcb_{21|21}  & 1/N^3 \\[3pt] &
 (B\cdot B\cdot B)\otimes (B\cdot B\cdot B) & \frac{10}{3}  &  \dcb_{03|03} & 1/N^3  \\[3pt] &
 (A\cdot A\cdot B)\otimes (B\cdot B\cdot B) & 4\ea \eb  & \dcb_{21|03}  & 1/N^3  \\[3pt] &
 (A\cdot B\cdot B)\otimes (A\cdot B\cdot B) & 6    & \dcb_{12|12} & 1/N^3
\end{array}\] 
\caption{Sextic operators, with 
their running coupling constants and scalings\label{tab:SexticOps}. 
} 
\end{table}

\fontsize{11.49}{14.9}\selectfont  

\subsection{Compatibility of the RG-flow with the Spectral Action}
\label{sec:Compatibility}
We now prove that in the double-trace truncation 
the RG-flow does not generate more operators 
than those allowed by the NCG-structure. 

\begin{proposition}\label{thm:Compatibility}
Pick a two-matrix model 
that includes finitely many single-trace operators $\TrN Q$, $ Q\in \CnN $,  
and assume that each of them appears (probably with other coefficient) in the Spectral Action for certain fuzzy $2$-dimensional geometry. Then the RG-flow generates at any order in the $FP\inv$-expansion exclusively operators that appear again, generally
with a different non-zero coefficient, in the Spectral Action $\Tr f(D)$, 
which in the worst case would require a suitable (generally higher-degree) polynomial $f$. 
\end{proposition}

\begin{proof}
Suppose that $\TrN Q$, with $Q\in \CnN$, features in the 
Spectral Action for a fuzzy geometry.
First, we show that the NC-polynomial 
$(\partial ^A \circ \partial^A \TrN Q )^{\star k} \in \A_2$ appears for each $k\in \Z_{\geq 1 }$ in the Spectral Action for the same fuzzy geometry---we argue later for the most general case containing mixed derivatives.   
From \eqref{undirected}, $\partial ^A \circ \partial^A \TrN Q $ contains two powers of $A$ less than 
the original NC-polynomial $Q$, which,
since it appears in the Spectral Action,
has an even degree $\deg_A(Q),\deg_B(Q) \in 2 \Z_{\geq0}$.
Therefore, so does the double derivative,
and by Lemmas \ref{thm:coeffssingle} and \ref{thm:coeffsbitrace}, 
$\partial ^A \circ \partial^A \TrN Q $ 
appears in the Spectral Action. 
The condition holds for any power 
$(\partial ^A \circ \partial^A \TrN Q )^{\star k}$
since the even-degree conditions are still
satisfied and therefore each monomial $w_1\otimes w_2$ or $w_1\totimes w_2$
in $(\partial ^A \circ \partial^A \TrN Q )^{\star k}$
appears in the Spectral Action $\Tr f(D)$ for a polynomial $f$ 
with non-zero coefficient in degree $m$, being $m=m(w_1,w_2)$ given by Lemma \ref{thm:coeffsbitrace}.\par 
The argument is still true for 
different NC-polynomials $Q_i$ appearing in the original Spectral Action
and the even-degree argument holds not only for powers of double derivatives of these, 
but can be clearly extended to 
\begin{align*}
\sum_{j_1,j_2,\ldots,j_r} 
(\partial ^{X_i} \circ \partial^{X_{j_1}} \TrN Q_1 )\star
(\partial ^{X_{j_1}} \circ \partial^{X_{j_2}} \TrN Q_2 ) \star\\ \qquad\qquad\cdots 
\star (\partial ^{X_{j_{r}}} \circ \partial^{X_i} \TrN Q_{r+1})
\end{align*}
since in the product the same derivative $\partial^{X_k}$ ($X_k \in \{A,B\}$) appears 
an even number of times. All the NC-polynomials generated by the supertrace
in the FRGE are of this form, and having even degree
in both matrices, the argument above leads in this case to the result.
\end{proof}
Proposition \ref{thm:Compatibility} says that 
if the bare action would contain only single-trace
operators, then all the operators that 
the RG-flow generates, including double-trace operators,
are compatible with the structure of fuzzy geometries.
This implies that for a realistic bare action,
which includes only double-trace operators as dictated by the Spectral Action for a fuzzy $2$-dimensional geometry, the RG-flow generates (up to triple traces excluded in the truncation) exclusively NCG-compatible operators. 
Both structures can therefore be seen as highly compatible.

\subsection{The truncated effective action}\label{sec:2MM_themodel}

The model we adopt includes all the operators 
appearing in the Spectral Action for fuzzy geometries 
computed in \cite{SAfuzzy} up to sixth degree.
For 2-dimensional fuzzy geometries, \[
\Gamma_N[A,B]=
\Tr_{\A_{2}} \Big \{ 1_N\otimes P(A,B)  + \sum_\alpha \Psi_\alpha(A,B) \otimes \Upsilon_\alpha (A,B) \Big\}, \qquad 
\]
where $ P, \Psi_\alpha,\Upsilon_\alpha \in \C_{\langleb 2\rangleb} =\C \langleb A,B\rangleb $ are given,
degree by degree by Tables \ref{tab:QQops} and \ref{tab:SexticOps}. There, 
a dot means the usual matrix product. 
The number of running coupling constants 
turns out to depend not only on the dimension, but 
also on the signature of the fuzzy geometry, see Table \ref{tab:numofOps}.
We stress that for the quartic and quadratic 
operators we do take the coupling constants with the symmetry factors 
and signs present in the NCG-action. For 
the sextic operators we drop the numerical normalization factors,
in order to avoid rational coefficients.
 \begin{table}[H]\small
\begin{tabular}{ccccc}
 \textsc{Geometry} & \textsc{Signature} & \textsc{KO-dim.} & $\#$ \textsc{Operators} & $\#$  \textsc{Operators}  \\
  &  &  &  \textsc{in the RG-flow} & \textsc{with duality}\\
`Double time' & $(+,+)$& 6& 48 & 26 \\
  Lorentzian &  $(+,-)$ &  0 &  41 & ---\\
 Riemannian & $(-,-)$ & 2 & 34 & 19
\end{tabular}
\vspace{8pt}
\caption{Number of operators for each signature. There is no
duality for the (1,1) geometry. \label{tab:numofOps}}
\end{table}%
\subsection{The $\beta$-functions}
We present now the set of equations satisfied 
by the fixed points $\mathsf{g}_{\!\balita\!\!\!}^\fxpt=\{\ac^\fxpt_{\balita} ,\bc^\fxpt_{\balita},\cc^\fxpt_{\balita},\dc_{\!\!\balita \!\!|\!\!\balita\!\!}^\fxpt\}$,
 determined
by the vanishing of all $\beta$-functions $\beta_{\!\balita\!\!\!}=\partial_t \mathsf{g}_{\!\balita\!\!}$. 
We recall that \[h_k=\lim _{N\to \infty} \frac{1}{N^2} \sum_{a,b,c,d=1}^\Lambda \frac{(\partial_t R_N)_{ab;cd} }{ P^{(k+1)}_{ab;cd}}\,,\]
which are real numbers in the case of the quadratic regulator of Section \ref{sec:CrossCheck} and whose whose values are 
given by eqs. \eqref{hvalues}.  Next result is more transparent 
if one does not specify these coefficients yet (and holds for any $R_N$ verifying that these $h_k$ are all independent of $N$). 
\begin{theorem}\label{thm:FixedPoints}
Assuming $\Za=\Zb=:Z$, to second order in the $FP\inv$ expansion $(h_r=0, r\geq 3)$,
in the double-trace and sixth-degree truncation, the $\beta$-functions of the 2-matrix model corresponding to a 2-dimensional fuzzy geometry with signature $\diag (\ea,\eb)$ 
are given in the large-$N$ limit by the following blocks of equations:\par 
\noindent

First, the degree-2 operators yield the anomalous dimension and following relations:
\allowdisplaybreaks[2]
 \begin{align*}
 2 h_{1} (\ac_{4}+\cc_{22}+2 \dc_{2|02}+6 \dc_{2|2})&=\etaa \\
 2 h_{1} (\bc_{4}+\cc_{22}+6 \dc_{02|02}+2 \dc_{2|02})&=\etab \\
 -h_{1} [\ea (\ac_{4}-\cc_{1111})+2 \dc_{1|12}+6 \dc_{1|3}]+\dc_{1|1} (\eta +1)&=\beta( \dc_{1|1} ) \\
 -h_{1} [\eb (\bc_{4}-\cc_{1111})+6 \dc_{01|03}+2 \dc_{01|21}]+\dc_{01|01} (\eta +1)&=\beta( \dc_{01|01} )   \end{align*} 
The next block encompasses the connected quartic 
 couplings: 
 \begin{align*}
 h_{2} \big(4 \ac_{4}^2+4 \cc_{22}^2 \big)+\ac_{4} (2 \eta +1) & \\   -h_{1} (24 \ac_{6} \ea+4 \cc_{42} \eb+4 \dc_{02|4} \eb+4 \dc_{2|4} \ea)&= \beta(\ac_{4} ) \\[1.5ex]
 h_{2} \big(4 \bc_{4}^2+4 \cc_{22}^2 \big)+\bc_{4} (2 \eta +1)  & \\ -h_{1} (24 \bc_{6} \eb+4 \cc_{24} \ea+4 \dc_{02|04} \eb+4 \dc_{2|04} \ea)&= \beta(\bc_{4} ) \\[1.5ex]
 -h_{1} \big(2 \ea\cc_{1212}  +\eb 2 \cc_{2121}  +3\ea   \cc_{24} +3 \eb \cc_{42}+\ea \dc_{02|22} +\eb\dc_{2|22}  \big) & \\ 
   +
 h_{2} \big(2 \ac_{4} \cc_{22}  +2 \bc_{4} \cc_{22}  +2\ea \eb \cc_{1111}^2 +2\ea \eb \cc_{22}^2 \big)+ \cc_{22} (2 \eta +1) &=\beta (\cc_{22}  ) \\[1.5ex]
 8 \ea \eb\cc_{1111} \cc_{22}  h_{2}+\cc_{1111} (2 \eta +1)  & \\ +h_{1} \big(4  \ea\cc_{1311} +4 \eb \cc_{3111} +2\ea \dc_{02|1111}  +2 \eb \dc_{2|1111} \big)&=\beta (\cc_{1111}
    )
    \end{align*} 
 The $\beta$-functions for the connected sextic 
 couplings are 
       \begin{align*}
       2 h_{2} (6 \ac_{4} \ac_{6}+\ea \eb\cc_{22} \cc_{42} )+\ac_{6} (3 \eta +2)&= \beta(\ac_{6} ) \\[1.5ex]
 2 h_{2} (6 \bc_{4} \bc_{6}+\ea \eb\cc_{22} \cc_{24} )+\bc_{6} (3 \eta +2)&= \beta(\bc_{6} ) \\[1.5ex]
 4 h_{2} \{\ac_{4} \cc_{3111}+\ea \eb [\cc_{22} (\cc_{1311}+2 \cc_{3111})& \\-\cc_{1111} (2 \cc_{2121}  +\cc_{42})]\}+\cc_{3111} (3 \eta +2)&=\beta (\cc_{3111}  ) \\[1.5ex]
 2 h_{2} [2 \ac_{4} \cc_{2121}+\ea \eb (-2 \cc_{1111} \cc_{3111}& \\+4 \cc_{2121} \cc_{22}+\cc_{22} \cc_{24})]+\cc_{2121} (3 \eta +2)&=\beta (\cc_{2121}  ) \\[1.5ex]
 2 h_{2} [\ac_{4} \cc_{24}+3 \bc_{4} \cc_{24}+2 \ea \eb (\cc_{22} (3 \bc_{6}+\cc_{2121}+\cc_{24}+\cc_{42})& \\-\cc_{1111} \cc_{1311})]+\cc_{24} (3 \eta +2)&=\beta (\cc_{24}  ) \\[1.5ex]
 4 h_{2} \{\bc_{4} \cc_{1311}+\ea \eb [\cc_{22} (2 \cc_{1311}+\cc_{3111})& \\-\cc_{1111} (2 \cc_{1212}+\cc_{24})]\}+\cc_{1311} (3 \eta +2)&=\beta (\cc_{1311}  ) \\[1.5ex]
 2 h_{2} [ 2 \bc_{4} \cc_{1212}+\ea \eb (\cc_{22} (4 \cc_{1212}+\cc_{42})& \\-2 \cc_{1111} \cc_{1311})]+\cc_{1212} (3 \eta +2)&=\beta (\cc_{1212}  ) \\[1.5ex]
 2 h_{2} [3 \ac_{4} \cc_{42}+2 \ea \eb (3 \ac_{6} \cc_{22}-\cc_{1111} \cc_{3111}+\cc_{1212} \cc_{22}& \\+\cc_{22} \cc_{24}+\cc_{22} \cc_{42})+\bc_{4} \cc_{42}]+\cc_{42} (3 \eta +2)&=\beta (\cc_{42}  ) 
   \end{align*}
   %
And a last block of $\beta$-functions
for the disconnected couplings is located in Appendix \ref{sec:DiscoFixedPt}.
\end{theorem}

\begin{proof}
We address first the first order in the $FP\inv$-expansion. 
This part of the proof consists on the following steps: \par
 \textit{$\balita$ Step 1.} 
 The computation of the second-order derivatives
 of all the operators, determine  
 the NC-Hessians to insert in the 
 $FP\inv$-expansion.  We now give the NC-Hessians computed using Claim \ref{thm:NCLap}
as well as their trace, the NC-Laplacian using Eq. \eqref{NCLap_prod}. We write 
some of them down up in Table \ref{tab:Hess24degr} to quartic operators;
those omitted might be obtained by the exchange $A\leftrightarrow B, \ea \leftrightarrow \eb$ (and adjusting the matrix structure).
\allowdisplaybreaks[2]
 
\begin{table}[H]
\begin{align*}
 \begin{array}{cr} \small \textsc{Operator} & \small\phantom{}\textsc{Its } \Hess_\sigma  \\[6pt]   \hline  
  & \\[-1ex]  
 \Tr(A^4) & \left(
\begin{array}{cc}
 4 \ea ({1}\otimes A^2+A^2\otimes {1}+A\otimes A) & 0 \\
 0 & 0 \\
\end{array}
\right) \\ [4ex]
 \Tr^2B & \left(
\begin{array}{cc}
 0 & 0 \\
 0 & 2 \eb {1}\totimes {1} \\
\end{array}
\right) \\ [4ex]
 \Tr(A B A B) & \left(
\begin{array}{cc}
 2 \ea B\otimes B & 2 ({1}\otimes BA+AB\otimes {1}) \\
 2 ({1}\otimes AB+BA\otimes {1}) & 2 \eb A\otimes A \\
\end{array}
\right) \\ [4ex]
 \Tr(A) \Tr(A^3) & \left(
\begin{array}{cc}
 3 \ea [\Tr(A) (A\otimes {1}+{1}\otimes A) & 0 \\ +{1}\totimes A^2+A^2\totimes {1}]  & \\[1ex]
 0 & 0 \\
\end{array}
\right) \\ [4ex]
 \Tr A^2 \Tr B^2 & \left(
\begin{array}{cc}
 2 \ea {1}\otimes {1} \Tr B^2 & 4 A\totimes B \\
 4 B\totimes A & 2 \eb {1}\otimes {1} \Tr A^2 \\
\end{array}
\right) \\ [4ex]
 \Tr^2 A^2 & \left(
\begin{array}{cc}
 4 \ea ({1}\otimes {1} \Tr A^2+2 A\totimes A) & 0 \\
 0 & 0 \\
\end{array}
\right) \\ [4ex]
\end{array}
\end{align*} 
\caption{Some Hessians of second and fourth order operators \label{tab:Hess24degr}}
\end{table}
\noindent
The expressions for $\Hess_\sigma \Tr (A A B B)$ 
and $ \Hess_\sigma \{\Tr(A) \Tr(A B B)\}$, the quartic operators missing
in Table \ref{tab:Hess24degr} were already given in Example \ref{ex:MultHessians} 
and show that 
the complexity rapidly grows. For sake of readability, the bulkier sixth-degree operators completing 
the running 34 or 48 involved in the flow, 
are located in Appendix \ref{sec:NCLap}.

\textit{$\balita$ Step 2.}  
To first order, one computes 
the regularized NC-Laplacian $F = \ea\Fa + \eb\Fb$ 
of the effective action 
in terms of \[\Fa=(\partial^A \circ \partial^A) \Gamma_N [A,B]-(\partial^A \circ \partial^A) \Gamma_N [A,B]\big|_{A=B=0}\]
and of \[\Fb=(\partial^B \circ \partial^B) \Gamma_N [A,B]-(\partial^B \circ \partial^B) \Gamma_N [A,B]\big|_{A=B=0}\,.\] 
\par 
\textit{$\balita$ Step 3.} One takes the double traces of the resulting
expression. The $h_1$-terms, i.e. the trace $\TrN^{\otimes 2}$ 
of the NC-Laplacian, read as in Appendix \ref{sec:IntermediateSteps}.

From that expression one can deduce some of the quantum fluctuations. In the large-$N$ limit, 
according to the scalings given in Tables \ref{tab:QQops} and \ref{tab:SexticOps}, the matching of the $h_1$ coefficients in the fixed point equations given in the statement can be verified.  
The scaling $N^{-m(Q)}$ of the coupling constant 
 $\mathsf{g}_{Q}$ that corresponds with the operator $Q\in \Cfree{2}\otimes\Cfree{2}$
determines the coefficient
of the form  \[ \Big( \frac{\deg_A (Q)+\deg_B(Q)}{2}\eta + m(Q)\Big)\times  \mathsf{g}_{Q}\,, \]
appearing in the $\beta_Q$-function.

We now sketch the second order: Having computed in Step 1 
the 48 NC-Hessians, one $\star$-squares the $(48-2)$ 
Hessians appearing in $F$ (the two subtracted operators
are $A^2$ and $B^2$ whose Hessian is absorbed in $P$).  
The $\sim 10^3$ matrices of size $2\times 2$ with NC-polynomial entries are omitted, but each of these 
was computed as in Example \ref{ex:MultHessians}. 
Then, one traces $F^{\star 2}$ in superspace to collect quantum fluctuations 
for each operator. Taking the large-$N$ limit of these leads to the results.  
\end{proof}

\subsection{Dualities}\label{sec:dualities}
It is convenient to look for dual solutions 
while aiming at determining the fixed points from the 
vanishing $\beta$-functions. \review{The duality is meant in the following sense.
To reduce the number of fixed-point equations, 
one makes some of them redundant by imposing the
$A \leftrightarrow B$ duality for couples of operators that allow it.
Thus, e.g., }
\[
\TrN (A^3BAB) \leftrightarrow \TrN (B^3ABA).
\]
is reflected in the duality $ \cc_{1311}  \leftrightarrow  \cc_{3111}$.
Imposing dualities does not halve the number of 
running constants, since some operators, e.g., $\TrN(ABAB)$, are invariant
under the $A\leftrightarrow B$ exchange (self-dual). With this in mind, we have
the following list:
\begin{remark} \label{thm:duality}
For the geometries $(2,0)$ and $(0,2)$ a duality in the 
effective action is manifest. Therefore, the $\beta$-functions
together with the equations for the anomalous dimensions $\etab$ and $\etaa$,
are invariant under the following exchange for the (2,0)-geometry: 
\begin{salign}
 \etab & \leftrightarrow    \etaa,&&&   \dc_{1|1} & \leftrightarrow  \dc_{01|01},&&&   
  \dc_{01|03} & \leftrightarrow  \dc_{1|3}, \\ \dc_{01|21} & \leftrightarrow  \dc_{1|12},&&&   \dc_{02|02} & \leftrightarrow  \dc_{2|2},&&&   \bc_{6} & \leftrightarrow  \ac_{6},&&&   
    \\ \cc_{1311} & \leftrightarrow  \cc_{3111},&&&   \cc_{1212} & \leftrightarrow  \cc_{2121},&&&   \dc_{01|05} & \leftrightarrow  \dc_{1|5},  \\ \dc_{01|41} & \leftrightarrow  \dc_{1|14},&&&   \dc_{01|1211} & \leftrightarrow  \dc_{1|2111},&&&   
  \dc_{11|13} & \leftrightarrow  \dc_{11|31},  \\ 
  \dc_{02|1111} & \leftrightarrow  \dc_{2|1111},&&&   \dc_{02|04} & \leftrightarrow  \dc_{2|4},&&&   \dc_{03|03} & \leftrightarrow  \dc_{3|3},   \\
  \dc_{21|03} & \leftrightarrow  \dc_{12|3},&&&   \dc_{12|12} & \leftrightarrow  \dc_{21|21}\,.    
  &&&   \dc_{01|23} & \leftrightarrow  \dc_{1|32},   \\
  \bc_{4} & \leftrightarrow  \ac_{4} ,&&&   \cc_{24} & \leftrightarrow  \cc_{42},&&&
   \dc_{02|4} &  \leftrightarrow  \dc_{2|04},
   \\
 \dc_{02|22} & \leftrightarrow  \dc_{2|22},   &&&   
  \end{salign}
  For the (0,2)-geometry, one excludes from this list
  the exchanges implying 
 \begin{align}\label{leftout}\dc_{01|01},\,\dc_{01|03},\,\dc_{01|21},\,\dc_{01|05},\,\dc_{01|23},\,\dc_{01|41} \and \dc_{01|1211}\,.\end{align}
\end{remark}
\begin{proof}
For both geometries $\ea=\eb$ holds. 
The duality is straightforwardly verified by inspecting the 48 equations.
For the $(0,2)$-geometry, $\TrN B=0$ which means that we make $\dc_{01|I}\equiv 0$,
where $I$ stands for any index combination, which is the list \eqref{leftout}.
(Also $\TrN A=0$ but excluding all $\dc_{01|\!\balita\!\!}$'s automatically excludes all
$\dc_{1|\!\balita\!\!} $'s.)
\end{proof}

\subsection{Methods and Results for the geometry $(0,2)$, or $(-,-)$}\label{sec:Geom02}

The fixed-point equations 
are the simultaneous zeros of all the $\beta$-functions listed in Theorem 
\ref{thm:FixedPoints} (and the two first equations there 
for the anomalous dimension). These are the eigenvalues of the stability matrix
\begin{align}
 -\text{Eig}\bigg\{ \bigg(\dervpar{ \beta_I{(\eta^\fxpt, \mathsf g_{\!\!\balita\!\!\!}})}{\mathsf g_{I'}}\bigg)\bigg|_{\mathsf g^\fxpt}\bigg\}_{I,I'}\,,
\label{stabilitymatrix}
\end{align} where $I,I'$ run over the flowing coupling constants. 
While analyzing the solutions:  \review{
\begin{itemize}\setlength\itemsep{.44em}
 \itemb We exclude the Gau\ss ian fixed point $\mathsf g_{\!\!\balita\!\!\!}^\fxpt=0$
with critical exponents determined by the scalings.
\itemb We report fixed points with at least one non-vanishing
connected coupling ($\ac,\bc,\cc$ types; but solutions with only nonvanishing 
$\dc$-type do exist). 
\itemb We do not report solutions that lead to imaginary 
critical exponents. That is, the reported solutions
correspond all to solely real eigenvalues of 
the stability matrix \eeqref{stabilitymatrix}. 
\itemb We only report solutions with coupling constants inside 
the $|\mathsf{g}_{\!\!\balita\!\!}|\leq 1$ hypercube. This 
restriction is due to our approach, which uses the $FP\inv$-expansion.
Without this restriction, the operators kept in the truncation 
would be less important than those we dropped. 
\end{itemize} }

Under these criteria, from the \review{$\sim 600$}
fixed point solutions for the $(0,2)$ geometry, we obtain a unique solution 
with a single positive eigenvalue, or in other words, a single relevant direction:
\begin{align}
\theta = +0.2749  \label{CritExp02} \,.
\end{align}
\review{The values of the coupling constants corresponding to it read}\reviewinequation{
\begin{salign} 
 \eta ^{\fxpt}&=  -0.3625  &&& \ac_{4}^{\fxpt}&=  -0.07972  &&& \ac_{6}^{\fxpt}&=  0  &&& \cc_{1111}^{\fxpt}&= 0 \\
 \cc_{2121}^{\fxpt}&=  0  &&& \cc_{22}^{\fxpt}&=  -0.03986  &&& \cc_{3111}^{\fxpt}&=  0  &&& \cc_{42}^{\fxpt}&=  0 \\ \dc_{2|02}^{\fxpt}&=  -0.01337  &&& \dc_{2|04}^{\fxpt}&=  0 &&&
   \dc_{2|1111}^{\fxpt}&=  0  &&& \dc_{12|3}^{\fxpt}&=  0  \\ \dc_{11|11}^{\fxpt}&=  -0.004201   &&& \dc_{2|4}^{\fxpt}&=  0  &&&  \dc_{2|22}^{\fxpt}&=  0 &&& \dc_{11|31}^{\fxpt}&=  0 \\ \dc_{2|2}^{\fxpt}&=  -0.005156 &&& \dc_{21|21}^{\fxpt}&=  0 &&& \dc_{3|3}^{\fxpt}&=  0\,.
\end{salign}
} 

\subsection{Results for the geometry $(2,0)$, or $(+,+)$} \label{sec:Geom20}
We report now the fixed points under the 
same criteria listed for the $(0,2)$ geometry (Sec. \ref{sec:Geom02}),
which restricts the \review{$\sim 600$ real solutions} to a few we now describe. 
If we further impose that the solution has precisely one relevant direction, then 
that critical exponent is unique and given by
\begin{align}
 \theta=+ 0.2749  \label{CritExp20}
\end{align}
\review{and the corresponding fixed point has the coupling constants:}
\reviewinequation{
\begin{salign}
 \eta ^{\fxpt} &= -0.3625&&&\ac_{4}^{\fxpt}&= -0.07972&&&\ac_{6}^{\fxpt}&= 0&&&\cc_{1111}^{\fxpt}&= 0  \\
   \cc_{22}^{\fxpt}&= -0.03986 &&&  \cc_{2121}^{\fxpt}&= 0 &&&\cc_{3111}^{\fxpt}&= 0 &&&
    \cc_{42}^{\fxpt}&= 0 
     \\   \dc_{2|02}^{\fxpt}&= -0.01337  &&& \dc_{2|04}^{\fxpt}&= 0 &&&
     \dc_{2|1111}^{\fxpt}&= 0 &&& \dc_{1|5}^{\fxpt}&= 0 \\ 
     \dc_{2|2}^{\fxpt}&= -0.005156&&&\dc_{2|22}^{\fxpt}&= 0 &&&    \dc_{2|4}^{\fxpt}&= 0&&&\dc_{12|3}^{\fxpt}&= 0 
   \\ \dc_{1|12}^{\fxpt}&= -0.00985&&& \dc_{3|3}^{\fxpt}&= 0 &&&  \dc_{21|21}^{\fxpt}&= 0 &&&\dc_{1|14}^{\fxpt}&= 0\\  
    \dc_{1|3}^{\fxpt}&= -0.00985&&& \dc_{1|2111}^{\fxpt}&= 0 &&&\dc_{1|32}^{\fxpt}&= 0&&& & \\ 
    \dc_{01|01}^{\fxpt}&= -0.2543&&&\dc_{11|11}^{\fxpt}&= -0.004201 &&& 
    \dc_{11|31}^{\fxpt}&= 0\,.  & &&& & \numerada 
\label{uniqueCC20} 
  \end{salign}}
Solutions with more connected non-vanishing
coupling constants exist (e.g., $\cc_{1111}\neq 0$ relevant for the $ABAB$-model\footnote{This model has
recently been addressed in \cite{ABAB}. See also the issue risen in \cite{CommentABAB} about one of the $\beta$-functions obtained
in \cite{ABAB}.} \cite{KazakovABAB,Ambjorn:ABAB,Ambjorn:ABAB2},
but they require two relevant directions (in this truncation). These are located in 
\ref{tab:CritCoupl20} in Appendix \ref{sec:Criteriondiffers}.  
\review{In particular, the agreement with the 
result of \cite{KazakovABAB} for the $A^4$-coupling is remarkable: 
\begin{align}\label{almostexact}
\ac_{4}^{\fxpt} = -0.07972  \approx   -\frac{1}{4\pi}  = (\ac_{4}^{\fxpt})|_{\text{[KZJ99]}} (=-0.079577...)
\end{align}
if one takes into account the flipped sign convention for $(\ac_{4}^{\fxpt})|_{\text{[KZJ99]}}$ (called $-\alpha$ there).
}Also notice that
\[
   \cc_{22}^{\fxpt} = -0.03986 \approx  -\frac1{8\pi} (= -0.039788...)
\]

%
%
%
%
%
%

\section{Conclusion and discussion}\label{sec:Conclusion}

Fuzzy geometry has elsewhere \cite{Glaser:2016epw,BarrettDruceGlaser} motivated
intrinsically random noncommutative geometric, numerical methods and statistical tools. Here, we use the fact that random NCG is in line with 
(Euclidean) QFT in 
order to explore fuzzy geometries via
 the Functional Renormalization Group for 
the multimatrix models these boil down to.

Using differential operators based on abstract algebra,
noncommutative calculus was useful to describe
the Functional Renormalization Group for general multimatrix models.
This paper focused on those derived from fuzzy spectral triples, which therefore allow both Hermitian and anti-Hermitian random matrices. We introduced a NC-Hessian---a non-symmetric(!) matrix of 
noncommutative derivatives---and a NC-Laplacian\footnote{These differential operators are treated more in detail in a future
paper.} on the free algebra. The latter is given by
\[
\nabla^2= \nabla 
\circ \Day = \text{noncommutative divergence $\circ$ cyclic gradient},
\]  %
wherein the noncommutative divergence is the operator \
\[\nabla \, Q = \sum_{i=1}^n e_i\partial^{X_i} Q_i  \text{ for } Q=(Q_1,\ldots, Q_n)\in \Cn^n\] and the cyclic gradient $\Day\, \Phi  = (\Day^{X_1} \Phi,\ldots, \Day^{X_n} \Phi)$
for $\Phi\in \Cn=\C \langleb X_1,\ldots,X_n \rangleb $. 
The NC-Hessian  governs the exact
Wetterich-Morris FRGE 
and $\nabla^2$ does so in the tadpole approximation, where
it has the form of a noncommutative heat equation (Prop. \ref{thm:tadpole}).
\review{One advantage of the present analysis is the ability to 
drop the assumption made by \cite{EichhornKoslowskiFRG} 
that $P$ commutes with $F[X]$---supposed there to hold
in a certain approximation scheme.} This turns out 
to be a consequence of the structure of the free algebra. \par 

The coordinate-free setting common to algebrists speeds up computations and 
facilitates writing proofs, which can be taken as a tool for more mathematical works implying the functional RG. Introducing that elegant language was ``priced'' at introducing $\totimes$, a new (twisted) product additional to Kronecker's. 
In fact, 
\begin{quote}
{the RG-flow for $n$-matrix models takes place in the algebra $M_n( \A_{n,N} )$ of matrices over $\A_{n,N} = (\CnN)^{\otimes 2} \oplus (\CnN)^{\totimes 2}$ with $\star$-product\footnote{This similar \textit{notation} is otherwise, also in noncommutative field theory, a known product. But ours here  does \textbf{not} refer to Moyal product.} given by Proposition \ref{thm:rightproductformula},} 
 \end{quote}
where $\CnN$ is the free algebra generated by $n$ matrices of size 
$N\times N$, and the RG-time\footnote{The right 
RG-parameter was not discovered here, but it was long known since \cite{Brezin:1992yc,Carles}} is $\log N$. Importantly, this $\star$-multiplication 
is not chosen by us here just because it satisfies nice mathematical properties, rather  
the FRGE \textit{dictates} it. In that sense, to present the proof of a 
``standard result'' sometimes pays off. 
Since many of the operators that appear in free algebra 
were originated in matrix theory (see \cite{GuionnetFreeAn,VoiculescuFreeQ,Rota}),
we remark that the $\star$-product given in Proposition \ref{thm:rightproductformula} 
(which, concretely for matrices had to be proven here) 
can be taken as a definition in abstract algebra, as no reference 
to matrix size or entries is made, replacing the trace $\TrN$ by a state 
$\varphi:\A_n\to \C$, whose cyclicity renders $\star$ associative (by Prop. \ref{thm:associativity}):
\begin{subequations}
\begin{align}
(U \totimes W) \star ( P\totimes Q)  &=  PU \totimes WQ \,, \\
(U \otimes W) \star ( P\totimes Q)  &=U  \otimes PWQ  \,, \\
(U \totimes W) \star ( P\otimes Q)  &= WPU  \otimes Q \,, \\
(U \otimes W) \star ( P\otimes Q) &= \varphi (WP) U\otimes Q\,.
\end{align}\end{subequations}
Similarly, the ``obvious'' product in Proposition \ref{productrules},
which resembles (only for monomials though) matrix multiplication 
on $M_2(\Cn)$, suggests
that the algebra $M_n ( M_2( \Cn  ))$ could be relevant\footnote{If the 
FRGE were not a second-order NC-differential equation, the number 2 would not appear in $M_n ( M_2( \Cn  ))$. 
The number 2 should not be confused with the two of 2-matrix models, 
or the number of products of traces allowed here, which is also two.} for an additional description of the FRGE, if one  
trades the product $\star$ by $\times$ using relations like \eeqref{asdfasdf}.
\par

Most of our findings rely on the algebraic structure of the RG-flow
but important part of the conclusion are the critical exponents 
for each geometry. For matrix models corresponding to 2-dimensional fuzzy geometries, the $\beta$-functions were extracted (Thm. \ref{thm:FixedPoints}) and 
the fixed point equations were numerically solved. The critical exponents found here---for the $(0,2)$ and $(2,0)$ geometries $\theta=+0.27491 $---were obtained from all the fixed point solutions as the 
unique solution that featured a single relevant direction. The fixed-point coupling constants do require a matrix mix, e.g., the coupling $\cc_{22}$ corresponding to the operator $ABBA$ is non-vanishing (see App. \ref{sec:Criteriondiffers}, where we report fixed points with two relevant directions for where more non-vanishing mixed operators in the flow, e.g., $ABAB$). \par 

\review{It is also remarkable that the operators that appear here in the $(2,0)$ 
geometry (of $(+,+)$ signature) are all generated by the RG-flow of the Hermitian two-matrix $ABAB$-model, whose exact solution by Kazakov--Zinn-Justin \cite{KazakovABAB} predicts 
a critical value $1/4\pi$ for the common coupling constant of the operators\footnote{Mind the flipped sign convention. Also that the couplings of the operators $A^4$ and $B^4$ 
have to coincide.} $-\frac14\Tr(A^4+B^4)$ 
and $-\frac12\Tr(ABAB)$. In view of Eq. \eqref{almostexact},
we obtained for the coupling of $A^4$ and $B^4$ a strikingly close value}%
\[ \text{our prediction} = 1.00179 \times \text{exact solution}.\]
However, the prediction of the other coupling 
does not enjoy the same success.
\par 

Concerning the NCG-structure, we showed in Section \ref{sec:Compatibility}
that a truncation by operator-degree and by number of 
traces was consistent with the structure 
of the Spectral Action for fuzzy 2-dimensional geometries.
Due the complexity\footnote{Thinking of words 
in $\Cfree{2}$ as sequences of 0's and 1's, 
this algebra has enough ``memory space'' for any digitizable 
data.} of the free algebra $\Cfree{2}$, it is not obvious
that the RG-flow should respect this structure. 
For example, recall that in the Hermitian random matrix model
the operators $X^{m}\otimes   X^{l}$, with $m$ and $l$ odd, are generated by the RG-flow;
these are removed by hand (in the truncations used in Sec. \ref{sec:CrossCheck} and \cite{EichhornKoslowskiFRG}). 
In contrast, truncations for fuzzy geometries do not require to drop other operators
than triple traces and operators that exceed a maximum degree.   
Notwithstanding this high compatibility, as perspective, it remains to improve the precision of the 
present results. We identify possible error sources 
in the computation of the 
fixed points as well as improvements to our approach: 
\begin{itemize}\setlength\itemsep{.44em}
 \itemb Extending the exploration from the examined unit-hypercubes to a larger 
region and estimation of residues in order to look for fixed points that correspond
to Dirac operators (i.e. obey a relation between the 
coefficients similar to that of Table \ref{tab:QQops}-\ref{tab:SexticOps}). This would allow to compare with Monte-Carlo 
simulations for the true Dirac operator
of fuzzy geometries \cite{Glaser:2016epw}. 

\itemb The exact RG-flow should consider operators that are not 
pure traces of elements in the free algebra, but that are smeared 
with functions resulting from the IR-regulator. 

\itemb Addressing the solution removing the duality 
imposed here; otherwise we might miss important fixed points
for which the $A\leftrightarrow B$ symmetry is broken. 

\itemb Another improvement that might lead to accuracy
is to consider more terms ($h_3\neq 0$) in the $FP\inv$-expansion.
With 48 running operators, this analysis requires time.

\itemb The arbitrariness in the definition of the 
IR-regulator $R_N$ might affect the numerical results. 
For this paper, this regulator has been calibrated by
imposing on it to lead to a good approximation 
to the expected solution for Hermitian matrix models, but
the lack of ``uniqueness'' of $R_N$ is unsatisfying. 
\review{More constrictions on $R_N$ should be thoroughly investigated. 
An important guide\footnote{I thank an anonymous referee for this reference} 
in order to achieve the optimization of the matrix IR-regulator
is \cite{MindTheGap}. The adaptation of that idea from the bosonic QFT-case
to the matrix case might be might sound straightforward , but 
the different sort of propagators should be taken into account---this actually requires some care.)}
\end{itemize}
\noindent
Further related directions are:
\begin{itemize}\setlength\itemsep{.44em}
                                                                                                   
\itemb \review{The NC-differential operators that we employed here govern also the 
Schwinger-Dyson of entirely general multimatrix models \cite{Guionnet:2010qg,MingoSpeicher}}.
Based on it, one can continue the investigations 
using Topological Recursion \cite{EynardOrantin,Chekhov:2006vd}
to address a solution of the models treated here.  
For one-dimensional geometries \cite{KhalkhaliTR}
report progress in this topic, using different analytic methods. Also, multimatrix models are known to be related to free probability 
whose tools might be helpful for this task. This paper 
 puts a common language forward, at least. 

\itemb In order to obtain the present results, we studied geometries whose 
effective action was manifestly symmetric in both random matrices
and for which the theory space was reducible to non-redundant couplings. 
The search for fixed points in the absence of the dualities, which for instance for the $(1,1)$-geometry means 41 flowing operators in the present
 truncation, was postponed. However,
 the formalism is appropriate for these and higher-dimensional ones.

\itemb Adding matter fields to these models can be 
accomplished by random almost-commutative geometries. With the FRGE developed here, one has 
a tool to delve into fuzzy geometries coupled to simplified matter sectors,
e.g., Maxwell or Yang-Mills(--Higgs) theories, addressed in a companion paper \cite{MultimatrixYMH}.  This brings us even closer 
to the original motivation (Sec. \ref{sec:intro}).
\end{itemize}
\vspace{2ex}

   \subsection*{Acknowledgements}
I thank Astrid Eichhorn for providing further details on \cite{EichhornKoslowskiFRG}, which positively impacted the present work. I was benefited by a discussion with Ant\^onio Duarte Pereira, without which 
I would not have discovered an error in previous versions. 
For comments that were helpful to improve the exposition and 
be more specific, I thank an anonymous referee.
The author was supported by the TEAM programme of the Foundation    
for Polish Science co-financed by the European Union under the
European Regional Development Fund (POIR.04.04.00-00-5C55/17-00).
 \vspace{1cm}
 \appendix %
\section{Glossary, conventions, other notations}\label{sec:Glossary}
     \thispagestyle{empty}
 \fontsize{11.4}{14.3}\selectfont  
 \begin{tabularx}{.85\textwidth}{lX}
     $1_N$&  identity matrix in $\MN$, corresponding to the empty word of $\CnN$ \\
      $\cdot$&  sum of matrix products $X\cdot J=\sum_{i=1}^n X_i J^i$; also sometimes an ordinary matrix product \\
     $\fxpt $& tags a fixed point (RG-context) \\
     $\fxpt_k$&  $k$-th fixed point \\
      $\star$&  product on $\A_n$ \\
      $\times$& (when not applied to scalars) is obvious product on $\A_n$ $ (A\otimes B)\times (C\otimes D)=AC\times BD $ (see Prop. \ref{productrules}) \\
           $\otimes_\tau$& twisted tensor product \\
                $* $& adjoint of a matrix, i.e. mainly dagger in physics \\
     $\nabla^2$   &   Noncommutative Laplacian  \\
     $\nabla^2_\tau$   &  twisted NC-Laplacian (trace of the twisted NC-Hessian) \\ 
     $\nabla^2_i$   & abbreviates $\partial^{X_i}\circ \partial^{X_i}$    \\[4pt]
    &  \hspace{4.7cm}$\diamondsuit$
\\[2ex]
$A,B$&  random matrices in the 2-dimensional fuzzy geometries
 \\   $a,b,c,d$&  indices corresponding to matrix entries
 \\
   $\A_n$& $\Cn^{\otimes 2}\oplus \Cn^{\totimes 2}$
\\   $\ac_I$&  renormalized coupling constant associated with an operator $\mathcal O_I(A)$
\\   $\bc_I$&  coupling constant associated with $\mathcal O_I(B)$
 \\
$\cc_I$&  coupling constant associated with $\mathcal O_I(A,B)$
\\   $\dc_{I|I'}$&  coupling constant associated a disconnected operator 
 $\mathcal O_I(A,B)\otimes \mtc O_{I'}(A,B)$ on either matrix 
    \end{tabularx}

 \begin{tabularx}{.85\textwidth}{lX}
   $\Cn$&  free algebra in $n$ generators 
 \\   $\CnN$&  free algebra in $n$ generators in $\MN$
\\   $D$&  Dirac operator \\
   $\Day^j, \Day^{X_j}$&  cyclic derivative with respect to $X_j$  \\
   $\partial^j, \partial^{X_j}$& noncommutative derivative with respect to $X_j$
\\ $e_i$&  signs; $e_i=+1$ if $X_i$ is Hermitian, and $e_i=-1$ if it is anti-Hermitian 
\\   field&  a non trivial word in the free algebra,  or in  $\CnNtens{k}$  
\\   $\bar{\mathsf{g}}_I,\bar{\mathsf{g}}_{\!\!\balita}$& coupling constants (not yet renormalized)
 \\   $\mathsf{g}_I, \mathsf{g}_{\!\!\balita}$& renormalized coupling constants
  \\   $\tilde h_k(N)$& corresponds with $[\dot RP^{k+1}]$ of \cite{EichhornKoslowskiFRG} 
  only before an IR-regulator is specified (mind the shift)
\\ $\Hess_\sigma$ &  noncommutative Hessian with diagonal entries scaled by $\sigma=\diag(e_1,\cdots,e_n)$
\\ $\Hess_\sigma^\tau $ & twisted NC-Hessian 
  \\   $h_k(N)$& corresponds to $\tilde h_k(N) /N^2 =[\dot RP^{k+1}]/N^2$ (cf \cite{EichhornKoslowskiFRG})
\\   $I$ &  generic index corresponding to (allowed) elements of $\Cn$
 \\   $i,j$&  indices corresponding typically to $i,j=1,\ldots,n$
 \\   $J$ &  sources (QFT-context)
  \\   $\Lambda$&  is a large integer that serves as (globally in this paper, absolute) UV-cutoff that verifies $ \Lambda \geq N $ ($\Lambda$ corresponds to $N'$ in \cite{TowardsEK}) 
     \\  $\M_N$ &  the space of matrices parametrizing the space of Dirac operators, shorthand for $\M^{p,q}_N$
 \\   $n$&  the number of (random) matrices; number of generators of the free algebra. 
 Caveat: in general $n$ does not coincide with the dimension $p+q$ of the fuzzy geometry
 that originates the matrix model
\\   $N$&  is the ``energy scale'', here an integer that verifies $ \Lambda \geq N $. 
Often here, $N$ is assumed also large   
\\ operator & in QFT-slang for monomial 
in the effective/bare action. 
Thus in our setting, an operator is a NC-polynomial
\\ $\mathcal O_I(X)$ & operator in the random matrix (or matrices) $X$ 
\\   $p$&  number of $+$ signs in the signature of a fuzzy geometry
\\   $q$&  number of minus signs in the same context
\\  $q \pm p $&  dimension $\diagup$ KO-dimension of a fuzzy geometry \\
      $R_N$&  IR-regulator (cutoff function)
    \\  $\STr$&  supertrace (no reference to supersymmetry)
    \\  $\STrN$&  supertrace in the truncation scheme
    \\$t$&  $t$ is the logarithm of the scale, here $t=\log N$
\\ $\Tr, \TrN$&  traces on $M_\Lambda(\C)$ and $\MN$ respectively
\\ $\TrN^k Q $& bracket-saving notation for $[\TrN(Q)]^k$
\end{tabularx}

 \begin{tabularx}{.85\textwidth}{lX}  $\tau $ & permutation $\tau = (13)\in \Sym(4)$ or ``twist''
   \\   $\Delta S_N$&  mass-like (quadratic in the fields) IR-regulator term 
\\   $\mtc W[J],  \mtc W_N[J]$ &  free energy (logarithm of the partition function)
   \\   $X$&   $n$-tuple of matrices, $X=(X_1,\ldots,X_n)$ 
    \\   $X_i$& random matrix obtained by $X_i=\langleb \varphi_i \rangleb$;  
    the averaged field $\varphi_i$, 
    Sec. \ref{sec:FRGEderivation}
     \\   $Z$ &  wave function renormalization constant
     \\   $\mtc Z,  \mtc Z_N$ &  partition functions; the second one IR-regulated
      \end{tabularx}

      \section{Fixed points with two relevant directions}\label{sec:Criteriondiffers}
      
  \review{In the main text the fixed points corresponding 
to the unique solution developing a single relevant direction were 
reported. 
Next tables correspond to two relevant directions,
where uniqueness of the solution is lost. The selection criteria were
the same mentioned in Sections 7.6 and 7.7. 
\subsection{Geometry $(0,2)$ or $(-,-)$}
With two relevant directions five solutions are found. 
The critical exponents are in that case 
   \[\begin{array}{lccc}
  & \theta_1 & \theta_2 & \\
\fxpt_{1,2} & +1.0318 & +0.274913  &\text{\scriptsize(two-fold multiplicity)}\\
\fxpt_{3,4,5} &+0.301298 & +0.027688   &\text{\scriptsize(three-fold multiplicity)}
\end{array}\]
 and the fixed-point coupling constant values are given in Table \ref{tab:CritExp02}. }{}
\begin{table} [H]
{
\small
\[ 
\begin{array}{llllll}
 \text{\scriptsize \textsc{Coupling}}  & \fxpt_{1} & \fxpt_{2} & \fxpt_3 & \fxpt_{4} & \fxpt_{5}  \\
 \eta  & -0.3625 & -0.3625 & -0.3418 & -0.3418 & -0.3418 \\
 \ac_{4} & -0.07972 & -0.07972 & -0.05812 & -0.05812 & -0.1194 \\
 \ac_{6} & \phantom{+}0 & \phantom{+}0 & -5.897\times 10^{-6} & -5.897\times 10^{-6} & -0.00002453 \\
 \cc_{1111} & \phantom{+}0 & \phantom{+}0 & +0.06126 & -0.06126 & \phantom{+}0 \\
 \cc_{2121} & \phantom{+}0 & \phantom{+}0 & -0.00001863 & -0.00001863 & -3.053\times 10^{-9} \\
 \cc_{22} & -0.03986 & -0.03986 & -0.05969 & -0.05969 & +0.001568 \\
 \cc_{3111} & \phantom{+}0 & \phantom{+}0 & -0.00003726 & +0.00003726 & \phantom{+}0 \\
 \cc_{42} & \phantom{+}0 & \phantom{+}0 & -0.00003632 & -0.00003632 & +9.418\times 10^{-7} \\
 \dc_{2|02} & -0.01337 & 0.08013 & -0.01289 & -0.01289 & +0.001253 \\
 \dc_{2|04} & \phantom{+}0 & \phantom{+}0 & -0.00002598 & -0.00002598 & +3.476\times 10^{-6} \\
 \dc_{2|1111} & \phantom{+}0 & \phantom{+}0 & -0.00005407 & +0.00005407 & \phantom{+}0 \\
 \dc_{2|2} & -0.005156 & -0.03632 & -0.004297 & -0.004297 & -0.009011 \\
 \dc_{2|22} & \phantom{+}0 & \phantom{+}0 & -0.000106 & -0.000106 & +2.107\times 10^{-6} \\
 \dc_{2|4} & \phantom{+}0 & \phantom{+}0 & -0.00002598 & -0.00002598 & -0.0001095 \\
 \dc_{11|11} & -0.3782 & -0.004201 & -0.05657 & -5.008\times 10^{-6} & -5.008\times 10^{-6} \\
 \dc_{11|31} & \phantom{+}0 & \phantom{+}0 & -0.000226 & -1.14\times 10^{-8} & -1.14\times 10^{-8} \\
 \dc_{12|3} & \phantom{+}0 & \phantom{+}0 & -0.00005331 & +5.71\times 10^{-7} & +9.208\times 10^{-7} \\
 \dc_{21|21} & \phantom{+}0 & \phantom{+}0 & -0.00008135 & +2.855\times 10^{-7} & -1.231\times 10^{-8} \\
 \dc_{3|3} & \phantom{+}0 & \phantom{+}0 & -8.735\times 10^{-6} & +2.855\times 10^{-7} & -0.00003585 \\
\end{array}
\]}
\caption{\reviewinequation{ All the coupling constant values
at those fixed points of the RG-flow for the $(0,2)$-geometry characterized 
by two relevant directions; for solutions with single relevant direction, cf. eqs. (7.5) \label{tab:CritExp02}}}
\end{table}

\subsection{Geometry $(2,0)$ or $(+,+)$}
\review{
Solutions with two relevant directions are six:
\[\begin{array}{lccc} 
 & \theta_1 & \theta_2 & \\
\fxpt_1 & + 1.0318 & +0.274913 &\\
\fxpt_2 & +1.0318 & +0.274913 &\\
\fxpt_{3,4,5,6} &+ 0.3013 &+ 0.02779\phantom{0} &  \text{\scriptsize(four-fold multiplicity)}
 \end{array}
\]
corresponding to the couplings shown in Table \ref{tab:CritCoupl20}.}{}
\begin{table}[H]
{
\small
\[
\begin{array}{lllll}
   \text{\scriptsize \textsc{Coupling}}  & \fxpt_{1} & \fxpt_{2} & \fxpt_3 & |\fxpt_{4,5,6}| \\
 \eta  & -0.3625 & -0.3625 & -0.3418 & +0.3418 \\
 \ac_{4} & -0.07972 & -0.07972 & -0.1194 & +0.05812 \\
 \ac_{6} & \phantom{+}0& \phantom{+}0& +0.00002453 & +5.897\times 10^{-6} \\
 \cc_{1111} & \phantom{+}0& \phantom{+}0& \phantom{+}0& +0.06126 \\
 \cc_{2121} & \phantom{+}0& \phantom{+}0& +3.053\times 10^{-9} & +0.00001863 \\
 \cc_{22} & -0.03986 & -0.03986 & +0.001568 & +0.05969 \\
 \cc_{3111} & \phantom{+}0& \phantom{+}0& \phantom{+}0& +0.00003726 \\
 \cc_{42} & \phantom{+}0& \phantom{+}0& -9.418\times 10^{-7} & +0.00003632 \\
 \dc_{2|02} & +0.08013 & -0.01337 & +0.001253 & +0.01289 \\
 \dc_{2|04} & \phantom{+}0& \phantom{+}0& +3.476\times 10^{-6} & +0.00002598 \\
 \dc_{2|1111} & \phantom{+}0& \phantom{+}0& \phantom{+}0& +0.00005407 \\
 \dc_{2|2} & -0.03632 & -0.005156 & -0.009011 & +0.004297 \\
 \dc_{2|22} & \phantom{+}0& \phantom{+}0& -2.107\times 10^{-6} & +0.000106 \\
 \dc_{2|4} & \phantom{+}0& \phantom{+}0& +0.0001095 & +0.00002598 \\
 \dc_{12|3} & \phantom{+}0 & \phantom{+}0 & -9.208\times 10^{-7} & +5.71\times 10^{-7} \\
 \dc_{21|21} & \phantom{+}0 & \phantom{+}0 & +1.231\times 10^{-8} & +2.855\times 10^{-7} \\
 \dc_{3|3} & \phantom{+}0 & \phantom{+}0 & +0.00003585 & +2.855\times 10^{-7} \\
 \dc_{1|12} & -0.00985 & -0.00985 & +0.0003325 & +0.0003721 \\
 \dc_{1|14} & \phantom{+}0 & \phantom{+}0 & -1.925\times 10^{-7} & +3.972\times 10^{-7} \\
 \dc_{1|2111} & \phantom{+}0 & \phantom{+}0 & +1.275\times 10^{-9} & +4.158\times 10^{-7} \\
 \dc_{1|3} & -0.00985 & -0.00985 & -0.02262 & +0.0003721 \\
 \dc_{1|32} & \phantom{+}0 & \phantom{+}0 & -1.269\times 10^{-6} & +1.21\times 10^{-6} \\
 \dc_{1|5} & \phantom{+}0 & \phantom{+}0 & +0.00005281 & +3.972\times 10^{-7} \\
 \dc_{01|01} & -0.2543 & -0.2543 & -0.3901 & +0.009373 \\
 \dc_{11|11} & -0.004201 & -0.3782 & -5.008\times 10^{-6} & +5.008\times 10^{-6} \\
 \dc_{11|31} & \phantom{+}0 & \phantom{+}0 & +1.14\times 10^{-8} & +1.14\times 10^{-8} \\
\end{array}
\]}
\caption{\label{tab:CritCoupl20} All the coupling constant values
at those fixed points of the RG-flow for the $(2,0)$-geometry characterized 
by two relevant directions. The bars around the fixed points $3,4,6$
mean that we report only the couplings in absolute value 
(the three solutions differ from sign change in some coupling constants).}
\end{table}

\section{Noncommutative Hessians and Laplacians}\label{sec:NCLap}
\fontsize{11}{14.3}\selectfont %
This appendix lists first the (twisted)  NC-Laplacians 
\[\nabla^2_\tau \Tr_{\A_2}(\mathcal O_{I|I'}(A,B))=\Tr_2 \big\{ \Hess_\tau^\sigma \Tr_{\A_2} [\mathcal O_{I|I'}(A,B) ] \big\}\] (by Claim 2.1)
of the double-trace operators 
$\mathcal O_{I|I'}(A,B)$ (in the left column). Subsequently, the Hessians also 
used in the proof of Theorem 7.2 are provided. 
The operators that do not appear
in the next list can be obtained from these by the $A\leftrightarrow B, \ea \leftrightarrow \eb $ exchange (and in the 
case of the Hessian, with pertinent changes in the matrix structure).  \par 
\fontsize{10}{13.9}\selectfont  
\[
\begin{array}{c|c}
\mathcal O_{I|I'}(A,B) & \nabla^2_\tau \Tr_{\A_2}(\mathcal O_{I|I'}(A,B)) \\[9pt]
N\inv {1_N\otimes (A\cdot A)} & 2 \ea 1_N\otimes 1_N \\[9pt]
 A\otimes A & 2 \ea 1_N\totimes 1_N \\[9pt]
 N\inv {1_N\otimes (A\cdot A\cdot A\cdot A)} & 4 \ea (1_N\otimes (A\cdot A)+(A\cdot A)\otimes 1_N+A\otimes A) \\[9pt]
 N\inv {1_N\otimes (A\cdot A\cdot B\cdot B)} & \eb 1_N\otimes (A\cdot A)+\eb (A\cdot A)\otimes 1_N
 \\[2pt] &\,\, 
 +\ea
   1_N\otimes (B\cdot B)+\ea (B\cdot B)\otimes 1_N \\[9pt]
 N\inv {1_N\otimes (A\cdot B\cdot A\cdot B)} & 2 (\eb A\otimes A+\ea B\otimes B) \\[9pt]
 (A\cdot B)\otimes (A\cdot B) & 2 (\eb A\totimes A+\ea B\totimes B) \\[9pt]
 (A\cdot A)\otimes (B\cdot B) & 2\cdot 1_N\totimes 1_N (\eb \TrN(A\cdot A)+\ea \TrN(B\cdot B)) \\[9pt]
 A\otimes (A\cdot A\cdot A) & 3 \ea (\TrN(A) (A\totimes 1_N+1_N\totimes A)\\[2pt] &  +1_N\otimes (A\cdot A)+(A\cdot A)\otimes
   1_N) \\[9pt]
 A\otimes (A\cdot B\cdot B) & \eb \TrN(A) (A\totimes 1_N+1_N\totimes A)
 \\[2pt] & +\ea 1_N\otimes (B\cdot B)+\ea
   (B\cdot B)\otimes 1_N \\[9pt]
 (A\cdot A)\otimes (A\cdot A) & 4 \ea (1_N\totimes 1_N \TrN(A\cdot A)+2 A\otimes A) \\[9pt]
  (B\cdot B)\otimes (B\cdot B) & 4 \eb (1_N\totimes 1_N \TrN(B\cdot B)+2 B\otimes B) 
  \\[9pt]
   N\inv {1_N\otimes (A\cdot A\cdot A\cdot A\cdot A\cdot A)} & 6 \ea (1_N\otimes (A\cdot A\cdot A\cdot A)+(A\cdot A\cdot A\cdot
   A)\otimes 1_N \\[2pt] &\,\,+ A\otimes (A\cdot A\cdot A) +(A\cdot A\cdot A)\otimes A  \\[2pt]& +(A\cdot A)\otimes (A\cdot A) )\\[9pt]N\inv {1_N\otimes (A\cdot A\cdot A\cdot A\cdot B\cdot B)} & \ea 1_N\otimes (A\cdot A\cdot B\cdot B)+\ea 1_N\otimes
   (A\cdot B\cdot B\cdot A) \\[2pt] &\,\, +\ea 1_N\otimes (B\cdot B\cdot A\cdot A)+\ea (A\cdot A\cdot B\cdot B)\otimes 1_N\\[2pt] &\,\, +\ea
   (A\cdot B\cdot B\cdot A)\otimes 1_N  +\ea (B\cdot B\cdot A\cdot A)\otimes 1_N\\[2pt] &\,\,+\ea A\otimes (A\cdot B\cdot B)+\ea
   A\otimes (B\cdot B\cdot A) \\[2pt] &\,\, +\ea (A\cdot A)\otimes (B\cdot B)+\ea (B\cdot B)\otimes (A\cdot A)\\[2pt] &\,\,+\ea (B\cdot B\cdot A)\otimes A+\eb 1_N\otimes (A\cdot A\cdot A\cdot A)
   \\[2pt] &\,\, +\ea (A\cdot B\cdot B)\otimes
   A +\eb (A\cdot A\cdot A\cdot A)\otimes 1_N\\[9pt] 
   N\inv {1_N\otimes (A\cdot A\cdot A\cdot B\cdot A\cdot B)} & \ea 1_N\otimes (A\cdot B\cdot A\cdot B)+\ea 1_N\otimes
   (B\cdot A\cdot B\cdot A)\\[2pt] &\,\,+\ea (A\cdot B\cdot A\cdot B)\otimes 1_N+\ea (B\cdot A\cdot B\cdot A)\otimes 1_N\\[2pt] &\,\,+\ea
   A\otimes (B\cdot A\cdot B)+\ea B\otimes (A\cdot A\cdot B)\\[2pt] &\,\,+\ea B\otimes (B\cdot A\cdot A)+\ea (A\cdot B)\otimes (B\cdot
   A)\\[2pt] &\,\,+\ea (A\cdot A\cdot B)\otimes B+\ea (B\cdot A\cdot A)\otimes B\\[2pt] &\,\,+\eb A\otimes (A\cdot A\cdot A)+\eb (A\cdot A\cdot A)\otimes A 
   \\[2pt] &\,\, +\ea (B\cdot A)\otimes (A\cdot B) +\ea (B\cdot A\cdot
   B)\otimes  
%
%
    \end{array}\]

\[\begin{array}{c|c} 
 \mathcal O_{I|I'}(A,B) & \nabla_\tau^2 \Tr_{\A_2}(\mathcal O_{I|I'}(A,B)) \\[20pt]
 \\[20pt]
 N\inv {1_N\otimes (A\cdot A\cdot B\cdot A\cdot A\cdot B)} & 2 [\ea 1_N\otimes (B\cdot A\cdot A\cdot B)+\ea (B\cdot
   A\cdot A\cdot B)\otimes 1_N\\[2pt] &\,\,+\ea B\otimes (A\cdot B\cdot A)+\ea (A\cdot B)\otimes (A\cdot B)\\[2pt] &\,\,+\ea (A\cdot B\cdot A)\otimes B+\eb (A\cdot A)\otimes (A\cdot A))  
   \\[2pt] & \,\,+\ea (B\cdot A)\otimes
   (B\cdot A)]
 \\[12pt]
 A\otimes (A\cdot B\cdot B\cdot B\cdot B) & \eb \TrN(A) \big\{1_N\totimes (A\cdot B\cdot B)+1_N\totimes (B\cdot A\cdot
   B)\\[2pt] &\,\,+1_N\totimes (B\cdot B\cdot A)+(A\cdot B\cdot B)\totimes 1_N\\[2pt] &\,\,+(B\cdot B\cdot A)\totimes
   1_N+A\totimes (B\cdot B)+B\totimes (A\cdot B)\\[2pt] &\,\,+B\totimes (B\cdot A)+(A\cdot B)\totimes B+(B\cdot A)\totimes B 
   \\[2pt] &\,\, +(B\cdot A\cdot B)\totimes 1_N +(B\cdot B)\totimes A
   \big\} 
   \\[2pt] &\,\,+\ea
      1_N\otimes (B\cdot B\cdot B\cdot B)+\ea (B\cdot B\cdot B\cdot B)\otimes 1_N 
      \\[12pt] 
 A\otimes (A\cdot A\cdot A\cdot B\cdot B) & \TrN(A) [\ea 1_N\totimes (A\cdot B\cdot B)+\ea 1_N\totimes (B\cdot B\cdot
   A) \\[2pt] &\,\,+\ea (A\cdot B\cdot B)\totimes 1_N+\ea (B\cdot B\cdot A)\totimes 1_N\\[2pt] &\,\,+\ea A\totimes (B\cdot B)+\ea (B\cdot
   B)\totimes A
   \\[2pt] &\,\,   +\eb 1_N\totimes (A\cdot A\cdot A)+\eb (A\cdot A\cdot A)\totimes 1_N]\\[2pt] &\,\,+\ea [1_N\otimes (A\cdot
   A\cdot B\cdot B)+1_N\otimes (A\cdot B\cdot B\cdot A)\\[2pt] &\,\,+(A\cdot A\cdot B\cdot B)\otimes
   1_N+(A\cdot B\cdot B\cdot A)\otimes 1_N
   \\[2pt] & \,\, +1_N\otimes (B\cdot B\cdot A\cdot A) +(B\cdot B\cdot A\cdot A)\otimes 1_N]
   \\[12pt]
 A\otimes (A\cdot A\cdot B\cdot A\cdot B) & \TrN(A) [\ea 1_N\totimes (B\cdot A\cdot B)+\ea (B\cdot A\cdot B)\totimes
   1_N\\[2pt] &\,\,+\ea B\totimes (A\cdot B)+\ea B\totimes (B\cdot A)+\ea (A\cdot B)\totimes B\\[2pt] &\,\,+\ea (B\cdot A)\totimes B+\eb
   A\totimes (A\cdot A)+\eb (A\cdot A)\totimes A]\\[2pt] &\,\,+\ea [ 1_N\otimes (B\cdot A\cdot A\cdot
   B)+1_N\otimes (B\cdot A\cdot B\cdot A)\\[2pt] &\,\,+(B\cdot A\cdot A\cdot B)\otimes
   1_N+(B\cdot A\cdot B\cdot A)\otimes 1_N
   \\[2pt] &\,\,+ 1_N\otimes (A\cdot B\cdot A\cdot B)+ (A\cdot B\cdot A\cdot B)\otimes 1_N ]
     \\[20pt]A\otimes (A\cdot A\cdot A\cdot A\cdot A) & 5 \ea \big\{\TrN(A) [ 1_N\totimes (A\cdot A\cdot A)+(A\cdot A\cdot A)\totimes
   1_N\\[2pt] &\,\,+A\totimes (A\cdot A)+(A\cdot A)\totimes A] \\[2pt] &\,\, +1_N\otimes (A\cdot A\cdot A\cdot A)+(A\cdot A\cdot A\cdot A)\otimes 1_N
     \big\} \\[10pt] & 
    \end{array}\]

\[\begin{array}{c|c} 
 \mathcal O_{I|I'}(A,B) & \nabla_\tau^2  \Tr_{\A_2}(\mathcal O_{I|I'}(A,B)) \\[9pt]
 (A\cdot B)\otimes (A\cdot A\cdot A\cdot B) & \ea \TrN(A\cdot B) (1_N\totimes (A\cdot B)+1_N\totimes (B\cdot A)\\[2pt] &\,\,+(A\cdot
   B)\totimes 1_N+(B\cdot A)\totimes 1_N+A\totimes B+B\totimes A)\\[2pt] &\,\,+\ea [B\otimes (A\cdot A\cdot B)+B\otimes (A\cdot B\cdot A)\\[2pt] &\,\,+B\otimes
   (B\cdot A\cdot A)]+\ea [(A\cdot A\cdot B)\otimes B\\[2pt] &\,\, +(A\cdot B\cdot A)\otimes B+(B\cdot A\cdot A)\otimes B]
   \\[2pt] &\,\,+\eb A\otimes (A\cdot
   A\cdot A)+\eb (A\cdot A\cdot A)\otimes A    \\[8pt]
 (A\cdot B)\otimes (A\cdot B\cdot B\cdot B) & \eb \TrN(A\cdot B) [1_N\totimes (A\cdot B)+1_N\totimes (B\cdot A)\\[2pt] &\,\,+(A\cdot
   B)\totimes 1_N+(B\cdot A)\totimes 1_N+A\totimes B+B\totimes A]\\[2pt] &\,\,+\eb [A\otimes (A\cdot B\cdot B)+A\otimes (B\cdot A\cdot B)
   \\[2pt] &\,\,
   +A\otimes
   (B\cdot B\cdot A)]+\eb [(A\cdot B\cdot B)\otimes A
   \\[2pt] &\,\,
   +(B\cdot A\cdot B)\otimes A+(B\cdot B\cdot A)\otimes A]\\[2pt] &\,\,+\ea B\otimes (B\cdot
   B\cdot B)+\ea (B\cdot B\cdot B)\otimes B 
   \\[9pt]
 (A\cdot A)\otimes (A\cdot A\cdot B\cdot B) & \TrN(A\cdot A) [\eb 1_N\totimes (A\cdot A)+\eb (A\cdot A)\totimes
   1_N \\[2pt] &\,\, \hspace{1.5cm} +\ea 1_N\totimes (B\cdot B)+\ea (B\cdot B)\totimes 1_N]\\[2pt] &\,\, +2 \ea 1_N\otimes 1_N \TrN(A\cdot
   A\cdot B\cdot B)\\[2pt] &\,\,+2 \ea [A\otimes (A\cdot B\cdot B)+A\otimes (B\cdot B\cdot A)] \\[2pt] &\,\,+ 2 \ea [(A\cdot B\cdot B)\otimes A + (B\cdot B\cdot
   A)\otimes A]  
   \\[8pt] 
 (A\cdot A)\otimes (A\cdot B\cdot A\cdot B) & 2 \TrN(A\cdot A) (\eb A\totimes A+\ea B\totimes B)\\[2pt] &\,\,+2 \ea 1_N\otimes
   1_N \TrN(A\cdot B\cdot A\cdot B)\\[2pt] &\,\, + 4 \ea A\otimes (B\cdot A\cdot B)+4 \ea (B\cdot A\cdot B)\otimes A \\[8pt]
 (A\cdot A)\otimes (A\cdot A\cdot A\cdot A) & 2 \ea [2 \TrN(A\cdot A) (1_N\totimes (A\cdot A)+(A\cdot A)\totimes
   1_N\\[2pt] &\,\,+A\totimes A)  + 1_N\totimes 1_N \TrN(A\cdot A\cdot A\cdot A)\\[2pt] &\,\,+4 A\otimes (A\cdot A\cdot A)+4 (A\cdot A\cdot A)\otimes A] \\[8pt]
   (A\cdot A)\otimes (B\cdot B\cdot B\cdot B) & 4 \eb \TrN(A\cdot A) [1_N\totimes (B\cdot B)+(B\cdot B)t\otimes 1_N\\[2pt] &\,\,+B\totimes
   B] +2 \ea 1_N\otimes 1_N \TrN(B\cdot B\cdot B\cdot B)  \\[8pt]
 (A\cdot A\cdot A)\otimes (A\cdot A\cdot A) & 6 \ea \big\{(A\totimes 1_N+1_N\totimes A) \TrN(A\cdot A\cdot A)\\[2pt] &\,\,+3 (A\cdot A)\otimes
   (A\cdot A)\big\} \\[8pt]
 (A\cdot B\cdot B)\otimes (A\cdot A\cdot A) & (A\totimes 1_N+1_N\totimes A) [3 \ea \TrN(A\cdot B\cdot B)\\[2pt] &\,\,+\eb
   \TrN(A\cdot A\cdot A)]\\[2pt] &\,\, +  3 \ea (A\cdot A)\otimes (B\cdot B)+3 \ea (B\cdot B)\otimes (A\cdot A) \\[8pt]
 (A\cdot A\cdot B)\otimes (A\cdot A\cdot B) & 2 \big\{\ea (B\totimes 1_N+1_N\totimes B) \TrN(A\cdot A\cdot B)\\[2pt] &\,\, + \ea [(A\cdot
   B)\otimes (A\cdot B)+(A\cdot B)\otimes (B\cdot A)\\[2pt] &\,\, +(B\cdot A)\otimes (A\cdot B)+(B\cdot A)\otimes (B\cdot A)]
   \\[2pt] &\,\,+\eb (A\cdot A)\otimes (A\cdot
   A)\big\}  
\end{array}\]

 \fontsize{11}{14.0}\selectfont 
\begin{landscape}
  \afterpage{
  \newgeometry{left=3.3cm, right=3.3cm,top=1.5in, bottom=1cm}
 \begin{align*}
 \begin{array}{cc} \mathcal O_{I|I'} & \Hess_\sigma \Tr_{\A_2}(\mathcal O_{I|I'}) \quad \mbox{ (untwisted)} \\[4ex]
 \frac1N 1\otimes A\cdot A\cdot A\cdot A\cdot A\cdot A  & \left(
\begin{array}{cc}
 6 \ea ({1}\otimes  A^4+ A^4\otimes {1}+A\otimes A^3+A^2\otimes A^2+A^3\otimes A) & 0 \\
 0 & 0 \\
\end{array}
\right) \\ [4ex]
 (A\cdot B)\otimes (A\cdot B) 
 & \left(
\begin{array}{cc}
 2 \ea B\totimes B & 2 ({1}\otimes {1} \Tr (AB)+B\totimes A) \\
 2 ({1}\otimes {1} \Tr (AB)+A\totimes B) & 2 \eb A\totimes A \\
\end{array}
\right) \\ [4ex]
  (A\cdot B\cdot B)\otimes (A\cdot A\cdot A) 
  & \left(
\begin{array}{cc}
 3 \ea [(A\otimes {1}+{1}\otimes A] \Tr(A B B)+A^2\totimes B^2+B^2\totimes A^2) & \Tr A^3 (B\otimes {1}+{1}\otimes B)+3 A^2\totimes AB+3 A^2\totimes BA \\
 \Tr A^3 (B\otimes {1}+{1}\otimes B)+3 AB\totimes A^2+3 BA\totimes A^2 & \eb (A\otimes {1}+{1}\otimes A) \Tr A^3 \\
\end{array}
\right) \\ [4ex]
 A\otimes (A\cdot A\cdot A\cdot A\cdot A) 
  & \left(
\begin{array}{cc}
 5 \ea [\Tr A ({1}\otimes  A^3+ A^3\otimes {1}+A\otimes A^2+A^2\otimes A)+{1}\totimes  A^4+ A^4\totimes {1}] & 0 \\
 0 & 0 \\
\end{array}
\right) \\ [4ex]
 (A\cdot A\cdot A)\otimes (A\cdot A\cdot A) & \left(
\begin{array}{cc}
 6 \ea [(A\otimes {1}+{1}\otimes A) \Tr A^3+3 A^2\totimes A^2] & 0 \\
 0 & 0 \\
\end{array}
\right) \\ [4ex]
 (A\cdot A)\otimes (A\cdot A\cdot A\cdot A)  & \left(
\begin{array}{cc}
 2 \ea (2 \Tr A^2 ({1}\otimes A^2+A^2\otimes {1}+A\otimes A)+{1}\otimes {1} \Tr A^4+4 A\totimes  A^3+4  A^3\totimes A) & 0 \\
 0 & 0 \\
 \end{array}
 \right) \\ [4ex]
(A\cdot A)\otimes (B\cdot B\cdot B\cdot B) & \left(
\begin{array}{cc}
 2 \ea {1}\otimes {1} \Tr B^4 & 8 A\totimes B^3 \\
 8 B^3\totimes A & 4 \eb \Tr A^2 ({1}\otimes B^2+B^2\otimes {1}+B\otimes B) \\
\end{array}
\right) \\ [4ex]
 \end{array}
\end{align*}
  \newpage 
\vspace{-10cm}
\[
\hspace{-11cm}\includegraphics[width=.9\textheight]{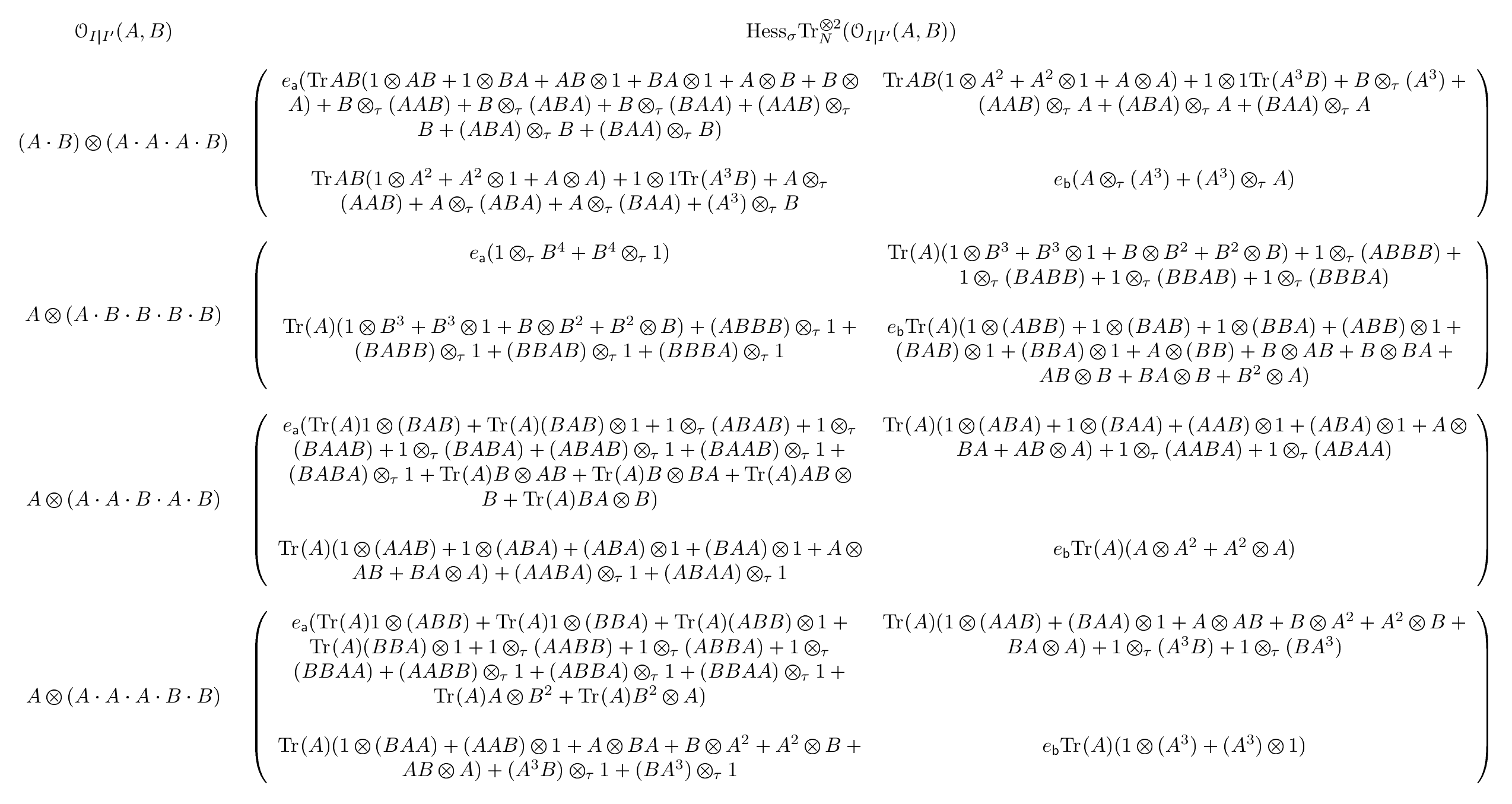}\]%
 \newpage%
 \centering
\[\hspace{-11cm}\raisebox{10cm}{\includegraphics[width=.9\textheight]{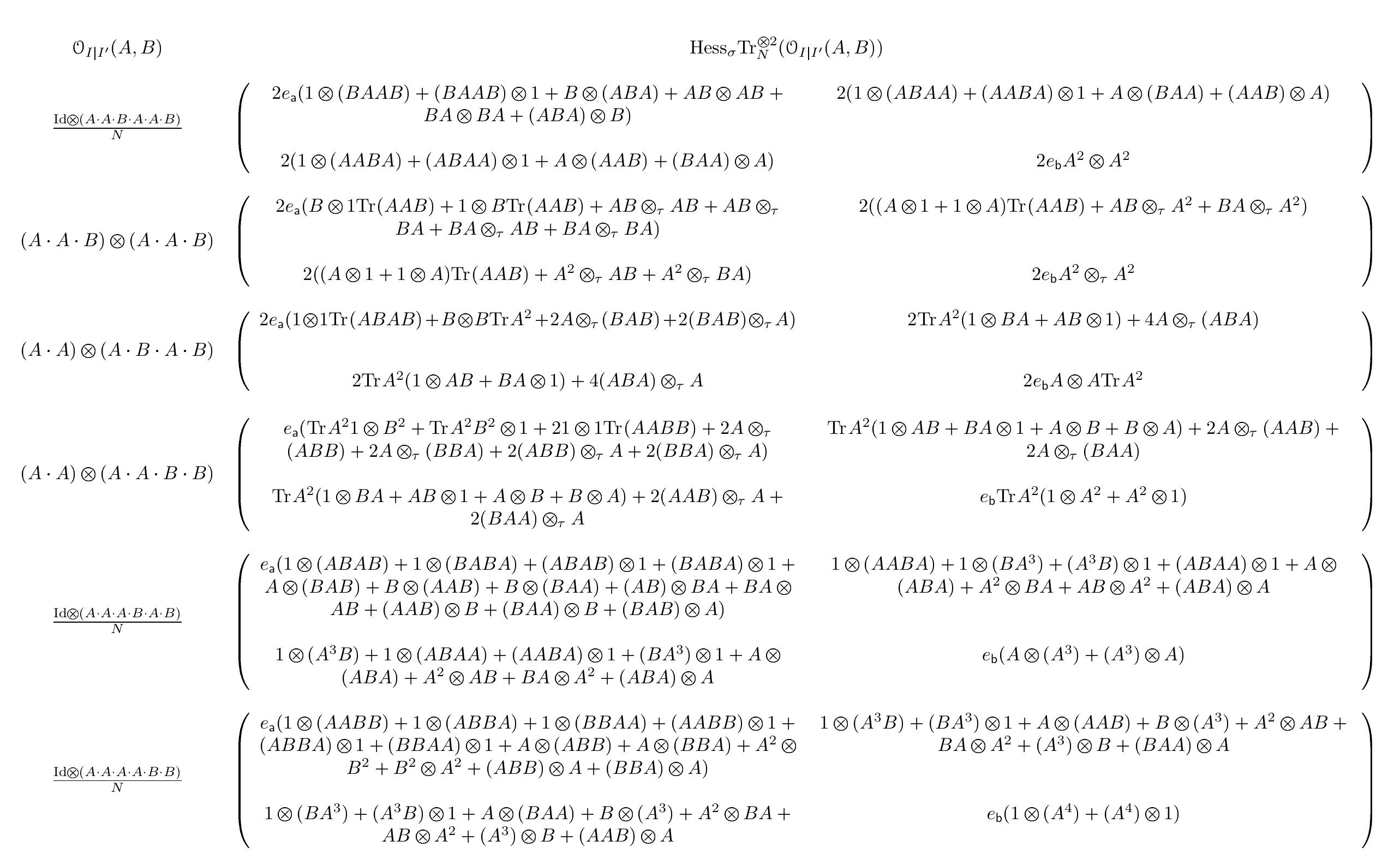}}\]

 \section{Beta-functions for disconnected couplings $\!\!\!\!\!\!\!\!\!\!\!\!\!\!\!\!\!\!\!\!\!\!\!\!\!\!\!\!\!\!\!\!\!\!\!\!\!\!\!\!\!\!\!\!\!\!\!\!\!\!\!\!\!\!\!\!\!\!\!\!\!\!\!\!\!\!\!\!\!\!\!\!\!\!\!\!\!\!\!\!\!\!\!\!\!\!\!\!\!\!\!\!\!\!\!\!\!\!\!\!\!\!\!\!\!\!\!\!$} \label{sec:DiscoFixedPt}
 \centering
  \begin{align*}
  2 h_{2} (\ac_{4} (12 \ac_{6}+5 \dc_{1|5})+18 \ac_{6} \dc_{1|3} \ea+\eb (\cc_{42} (\dc_{1|12}    -\cc_{1111} \ea)+\cc_{22} \ea (\cc_{3111}+\dc_{1|32})))+\dc_{1|5} (3 \eta +3)&=\beta( \dc_{1|5} ) \\[3ex]
   2 h_{2} (\bc_{4} (12 \bc_{6}+5 \dc_{01|05})+\eb (18 \bc_{6} \dc_{01|03}+\cc_{22} \ea (\cc_{1311} +\dc_{01|23}))+\cc_{24} \ea (\dc_{01|21}-\cc_{1111} \eb))+\dc_{01|05} (3 \eta +3)&=\beta( \dc_{01|05}
   ) \\[3ex]
 h_{2} \big(-\ac_{4} \cc_{1111} \ea \eb-\bc_{4} \cc_{1111} \ea \eb+2 \cc_{1111}^2-4 \cc_{1111} \cc_{22}-8 \cc_{1111} \dc_{11|11} \ea \eb+4 \cc_{22}^2+8 \cc_{22} \dc_{11|11} \ea \eb+4
   \dc_{11|11}^2\big)  & \\ +h_{1} (-2 \cc_{1212} \eb-\cc_{1311} \eb-2 \cc_{2121} \ea-\cc_{3111} \ea-2 \dc_{11|13} \eb-2 \dc_{11|31} \ea)+\dc_{11|11} (2 \eta +2)&=\beta( \dc_{11|11} ) \\[3ex]
 h_{2} \big(\ac_{4} \cc_{22}  +4 \ac_{4} \dc_{2|02}  +\bc_{4} \cc_{22}  +4 \bc_{4} \dc_{2|02}  +\cc_{22}^2 \ea \eb+12 \cc_{22} \dc_{02|02} 
   +12 \cc_{22} \dc_{2|2}  +24 \dc_{02|02} \dc_{2|02}  +24 \dc_{2|02} \dc_{2|2}  \big)  & \\ +h_{1} \big(-\frac{1}{2} \cc_{24} \ea -\frac{1}{2} \cc_{42} 
   \eb-\frac{1}{2} \dc_{02|22} \ea -2 \dc_{02|4}  \eb-2 \dc_{2|04} \ea -\frac{1}{2} \dc_{2|22}  \eb\big)+\dc_{2|02} (2 \eta +2)&=\beta( \dc_{2|02} ) \\[3ex]
 h_{2} \big(2 \ac_{4}^2 \ea+6 \ac_{4} \dc_{1|3} -2 \cc_{1111} \cc_{22} \ea +2 \cc_{22} \dc_{1|12}  \big)  & \\ +h_{1} \big(-6 \ac_{6} -\cc_{3111} \ea \eb-\dc_{12|3} \ea
   \eb-\dc_{1|32} \ea \eb-5 \dc_{1|5} -6 \dc_{3|3} \big)+\dc_{1|3} (2 \eta +2)&=\beta( \dc_{1|3} ) \\[3ex]
 h_{2} \big(2 \ac_{4} \cc_{22} \ea -2 \bc_{4} \cc_{1111} \ea +2 \bc_{4} \dc_{1|12} -4 \cc_{1111} \cc_{22} \eb-4 \cc_{1111} \dc_{1|12} \ea \eb+4 \cc_{22}^2 \eb+4 \cc_{22} \dc_{1|12}
   \ea \eb+6 \cc_{22} \dc_{1|3}  \big)  & \\ +h_{1} \big(-2 \cc_{1212} -2 \cc_{1311} -\cc_{3111} \ea \eb-2 \cc_{42} \ea \eb-2 \dc_{12|12} -3 \dc_{12|3} \ea \eb-3
   \dc_{1|14} -\dc_{1|2111} \ea \eb-2 \dc_{1|32} \ea \eb\big)+\dc_{1|12} (2 \eta +2)&=\beta( \dc_{1|12} ) \\[3ex]
    h_{2} \big(-2 \ac_{4} \cc_{1111}  \eb+2 \ac_{4} \dc_{01|21} +2 \bc_{4} \cc_{22}  \eb-4 \cc_{1111} \cc_{22} \ea-4 \cc_{1111} \dc_{01|21} \ea \eb+4 \cc_{22}^2 \ea+6 \cc_{22} \dc_{01|03}
    +4 \cc_{22} \dc_{01|21} \ea \eb\big)  & \\ +h_{1} \big(-\cc_{1311} \ea \eb-2 \cc_{2121} -2 \cc_{24} \ea \eb-2 \cc_{3111} -\dc_{01|1211} \ea \eb-2 \dc_{01|23} \ea
   \eb-3 \dc_{01|41} -3 \dc_{21|03} \ea \eb-2 \dc_{21|21} \big)+\dc_{01|21} (2 \eta +2)&=\beta( \dc_{01|21} )  \\[3ex]
 h_{2} \big(2 \bc_{4}^2 \eb+6 \bc_{4} \dc_{01|03} -2 \cc_{1111} \cc_{22}  \eb+2 \cc_{22} \dc_{01|21}  \big)  & \\ +h_{1} \big(-6 \bc_{6} -\cc_{1311} \ea \eb-5 \dc_{01|05}
   -\dc_{01|23} \ea \eb-6 \dc_{03|03} -\dc_{21|03} \ea \eb\big)+\dc_{01|03} (2 \eta +2)&=\beta( \dc_{01|03} ) \hspace{1cm}
   \end{align*}  \newpage
      
\begin{align*}
 2 h_{2} (\ac_{4} \cc_{24}+\bc_{4} (2 \cc_{1212}+2 \cc_{1311}+3 \dc_{1|14})-6 \bc_{6} \cc_{1111} \ea \eb+6 \bc_{6} \dc_{1|12} \eb-2 \cc_{1111} \cc_{24} \ea \eb-2 \cc_{1111} \dc_{1|14} \ea \eb&   \\+4 \cc_{1212}
   \cc_{22} \ea \eb+4 \cc_{1212} \dc_{1|12} \eb+4 \cc_{1311} \cc_{22} \ea \eb+4 \cc_{1311} \dc_{1|12} \eb+4 \cc_{22} \cc_{24} \ea \eb+\cc_{22} \cc_{3111} \ea \eb+2 \cc_{22} \cc_{42} \ea
   \eb&   \\+2 \cc_{22} \dc_{1|14} \ea \eb+\cc_{22} \dc_{1|2111} \ea \eb+2 \cc_{22} \dc_{1|32} \ea \eb+2 \cc_{24} \dc_{1|12} \eb+3 \cc_{24} \dc_{1|3} \ea)+\dc_{1|14} (3 \eta +3)&=\beta(
   \dc_{1|14} )  \\[3ex]
 2 h_{2} (5 \ac_{4} \cc_{42}+2 \ac_{4} \dc_{1|32}+6 \ac_{6} \cc_{22} \ea \eb+\bc_{4} (\cc_{3111}+\dc_{1|32})-2 \cc_{1111} \cc_{1212} \ea \eb-4 \cc_{1111} \cc_{2121} \ea \eb-2 \cc_{1111} \cc_{24} \ea \eb&   \\-2
   \cc_{1111} \dc_{1|2111} \ea \eb+2 \cc_{1212} \dc_{1|12} \eb+2 \cc_{1311} \cc_{22} \ea \eb+2 \cc_{22} \cc_{3111} \ea \eb+4 \cc_{22} \cc_{42} \ea \eb+2 \cc_{22} \dc_{1|14} \ea \eb+2 \cc_{22}
   \dc_{1|32} \ea \eb&   \\+5 \cc_{22} \dc_{1|5} \ea \eb+2 \cc_{24} \dc_{1|12} \eb+2 \cc_{3111} \dc_{1|12} \eb+2 \cc_{42} \dc_{1|12} \eb+9 \cc_{42} \dc_{1|3} \ea)+\dc_{1|32} (3 \eta +3)&=\beta(
   \dc_{1|32} ) \\[3ex]
 2 h_{2} (\ac_{4} (2 \cc_{2121}+3 \cc_{3111}+\dc_{1|2111})-2 \cc_{1111} \cc_{1311} \ea \eb-\cc_{1111} \cc_{24} \ea \eb-2 \cc_{1111} \cc_{3111} \ea \eb-2 \cc_{1111} \cc_{42} \ea \eb-2 \cc_{1111} \dc_{1|2111}
   \ea \eb&   \\-2 \cc_{1111} \dc_{1|32} \ea \eb+2 \cc_{1212} \cc_{22} \ea \eb+2 \cc_{1311} \dc_{1|12} \eb+\cc_{2121} (8 \cc_{22} \ea \eb+4 \dc_{1|12} \eb+6 \dc_{1|3} \ea)+4 \cc_{22} \cc_{3111}
   \ea \eb&   \\+\cc_{22} \dc_{1|14} \ea \eb+4 \cc_{22} \dc_{1|2111} \ea \eb+\cc_{24} \dc_{1|12} \eb+2 \cc_{3111} \dc_{1|12} \eb+6 \cc_{3111} \dc_{1|3} \ea)+\dc_{1|2111} (3 \eta +3)&=\beta(
   \dc_{1|2111} ) \\[3ex]
 2 h_{2} (\ac_{4} (2 \cc_{2121}+2 \cc_{3111}+3 \dc_{01|41})-6 \ac_{6} \cc_{1111} \ea \eb+6 \ac_{6} \dc_{01|21} \ea+\bc_{4} \cc_{42}-2 \cc_{1111} \cc_{42} \ea \eb-2 \cc_{1111} \dc_{01|41} \ea \eb+\cc_{1311}
   \cc_{22} \ea \eb&   \\+4 \cc_{2121} \cc_{22} \ea \eb+4 \cc_{2121} \dc_{01|21} \ea+2 \cc_{22} \cc_{24} \ea \eb+4 \cc_{22} \cc_{3111} \ea \eb+4 \cc_{22} \cc_{42} \ea \eb+\cc_{22} \dc_{01|1211}
   \ea \eb+2 \cc_{22} \dc_{01|23} \ea \e&   \\b+2 \cc_{22} \dc_{01|41} \ea \eb+4 \cc_{3111} \dc_{01|21} \ea+3 \cc_{42} \dc_{01|03} \eb+2 \cc_{42} \dc_{01|21} \ea)+\dc_{01|41} (3 \eta +3)&=\beta(
   \dc_{01|41} )  
    \\[3ex]
 h_{2} \bigg(\frac{\ac_{4}^2}{2}+8 \ac_{4} \dc_{2|2}+\frac{\cc_{1111}^2 }{6}+\frac{\cc_{22}^2 }{3}+\frac{8}{3} \cc_{22} \dc_{2|02} +\frac{8 \dc_{2|02}^2 }{3}+24 \dc_{2|2}^2\bigg)    +h_{1} \bigg(-\ac_{6}
   \ea-\frac{\cc_{2121} \eb}{3}-\frac{\dc_{2|22} \eb}{3}-\frac{4 \dc_{2|4} \ea}{3}\bigg)+\dc_{2|2} (2 \eta +2)&=\beta( \dc_{2|2} ) \\[1.5ex]
 h_{2} \bigg(\frac{\bc_{4}^2}{2}+8 \bc_{4} \dc_{02|02}+\frac{\cc_{1111}^2 }{6}+\frac{\cc_{22}^2 }{3}+\frac{8}{3} \cc_{22} \dc_{2|02} +24 \dc_{02|02}^2+\frac{8 \dc_{2|02}^2 }{3}\bigg)    +h_{1} \bigg(-\bc_{6}
   \eb-\frac{\cc_{1212} \ea}{3}-\frac{4 \dc_{02|04} \eb}{3}-\frac{\dc_{02|22} \ea}{3}\bigg)+\dc_{02|02} (2 \eta +2)&=\beta( \dc_{02|02} ) \end{align*} 
   \newpage

   \begin{align*} 2 h_{2} (\ac_{4} (\cc_{1311}+\dc_{01|23})+5 \bc_{4} \cc_{24}+2 \bc_{4} \dc_{01|23}+6 \bc_{6} \cc_{22} \ea \eb-4 \cc_{1111} \cc_{1212} \ea \eb-2 \cc_{1111} \cc_{2121} \ea \eb-2 \cc_{1111} \cc_{42} \ea   \eb& \\
   -2 \cc_{1111} \dc_{01|1211} \ea \eb+2 \cc_{1311} \cc_{22} \ea \eb+2 \cc_{1311} \dc_{01|21} \ea+2 \cc_{2121} \dc_{01|21} \ea+4 \cc_{22} \cc_{24} \ea \eb+2 \cc_{22} \cc_{3111} \ea \eb+5   \cc_{22} \dc_{01|05} \ea \eb+2 \cc_{22} \dc_{01|23} \ea \eb& \\
   +2 \cc_{22} \dc_{01|41} \ea \eb+9 \cc_{24} \dc_{01|03} \eb+2 \cc_{24} \dc_{01|21} \ea+2 \cc_{42} \dc_{01|21} \ea)+\dc_{01|23} (3 \eta   +3)&=\beta( \dc_{01|23} ) \\[3ex]
 2 h_{2} (\bc_{4} (2 \cc_{1212}+3 \cc_{1311}+\dc_{01|1211})-2 \cc_{1111} \cc_{1311} \ea \eb-2 \cc_{1111} \cc_{24} \ea \eb-2 \cc_{1111} \cc_{3111} \ea \eb-\cc_{1111} \cc_{42} \ea \eb-2 \cc_{1111}   \dc_{01|1211} \ea \eb& \\
 -2 \cc_{1111} \dc_{01|23} \ea \eb+\cc_{1212} (8 \cc_{22} \ea \eb+6 \dc_{01|03} \eb+4 \dc_{01|21} \ea)+4 \cc_{1311} \cc_{22} \ea \eb+6 \cc_{1311} \dc_{01|03} \eb+2   \cc_{1311} \dc_{01|21} \ea+2 \cc_{2121} \cc_{22} \ea \eb& \\
 +4 \cc_{22} \dc_{01|1211} \ea \eb+\cc_{22} \dc_{01|41} \ea \eb+2 \cc_{3111} \dc_{01|21} \ea+\cc_{42} \dc_{01|21} \ea)+\dc_{01|1211} (3 \eta   +3)&=\beta( \dc_{01|1211} )\hspace{1cm} \\[3ex]
 2 h_{2} (2 \ac_{4} (3 \cc_{2121}+2 \cc_{3111}+\dc_{11|31})-6 \ac_{6} \cc_{1111} \ea \eb+\bc_{4} \cc_{3111}-\cc_{1111} \cc_{1311} \ea \eb-4 \cc_{1111} \cc_{2121} \ea \eb-2 \cc_{1111} \cc_{24} \ea \eb& \\
 -6   \cc_{1111} \cc_{3111} \ea \eb-4 \cc_{1111} \cc_{42} \ea \eb-4 \cc_{1111} \dc_{11|31} \ea \eb+4 \cc_{1212} \cc_{22} \ea \eb+2 \cc_{1311} \cc_{22} \ea \eb+8 \cc_{2121} \cc_{22} \ea \eb+8   \cc_{2121} \dc_{11|11}& \\
 +8 \cc_{22} \cc_{3111} \ea \eb+8 \cc_{22} \cc_{42} \ea \eb+2 \cc_{22} \dc_{11|13} \ea \eb+4 \cc_{22} \dc_{11|31} \ea \eb+8 \cc_{3111} \dc_{11|11}+4 \cc_{42} \dc_{11|11}+2 \dc_{11|11}   \dc_{11|31})+\dc_{11|31} (3 \eta +3)&=\beta( \dc_{11|31} ) \\[3ex]
 2 h_{2} (\ac_{4} \cc_{1311}+2 \bc_{4} (3 \cc_{1212}+2 \cc_{1311}+\dc_{11|13})-6 \bc_{6} \cc_{1111} \ea \eb-4 \cc_{1111} \cc_{1212} \ea \eb-6 \cc_{1111} \cc_{1311} \ea \eb-4 \cc_{1111} \cc_{24} \ea   \eb& \\
 -\cc_{1111} \cc_{3111} \ea \eb-2 \cc_{1111} \cc_{42} \ea \eb-4 \cc_{1111} \dc_{11|13} \ea \eb+8 \cc_{1212} \cc_{22} \ea \eb+8 \cc_{1212} \dc_{11|11}+8 \cc_{1311} \cc_{22} \ea \eb+8   \cc_{1311} \dc_{11|11}& \\
 +4 \cc_{2121} \cc_{22} \ea \eb+8 \cc_{22} \cc_{24} \ea \eb+2 \cc_{22} \cc_{3111} \ea \eb+4 \cc_{22} \dc_{11|13} \ea \eb+2 \cc_{22} \dc_{11|31} \ea \eb+4 \cc_{24}   \dc_{11|11}+2 \dc_{11|11} \dc_{11|13})+\dc_{11|13} (3 \eta +3)&=\beta( \dc_{11|13} ) \\[3ex]
 2 h_{2} (\ac_{4} (2 \cc_{2121}+6 \cc_{42}+3 \dc_{2|22})+2 (3 \ac_{6} \cc_{22} \ea \eb-\cc_{1111} \cc_{1311} \ea \eb-\cc_{1111} \cc_{3111} \ea \eb-2 \cc_{1111} \dc_{2|1111} \ea \eb+\cc_{1212} \cc_{22}   \ea \eb& \\
 +4 \cc_{1212} \dc_{2|02} \ea \eb+2 \cc_{2121} \cc_{22} \ea \eb+12 \cc_{2121} \dc_{2|2}+2 \cc_{22} \cc_{24} \ea \eb+2 \cc_{22} \cc_{42} \ea \eb+\cc_{22} \dc_{02|22} \ea \eb+2   \cc_{22} \dc_{2|04} \ea \eb& \\
 +\cc_{22} \dc_{2|22} \ea \eb+2 \cc_{22} \dc_{2|4} \ea \eb+6 \cc_{24} \dc_{2|02} \ea \eb+18 \cc_{42} \dc_{2|2}+2 \dc_{02|22} \dc_{2|02} \ea \eb+6 \dc_{2|2}   \dc_{2|22})+\bc_{4} (2 \cc_{2121}+\dc_{2|22}))+\dc_{2|22} (3 \eta +3)&=\beta( \dc_{2|22} ) \\[3ex]
 2 h_{2} (3 \ac_{4} \cc_{3111}+2 \ac_{4} \dc_{2|1111}-2 \cc_{1111} \cc_{1212} \ea \eb-2 \cc_{1111} \cc_{42} \ea \eb-2 \cc_{1111} \dc_{2|22} \ea \eb+2 \cc_{1311} \cc_{22} \ea \eb+8 \cc_{1311} \dc_{2|02}   \ea \eb & \\
 +4 \cc_{22} \cc_{3111} \ea \eb+2 \cc_{22} \dc_{02|1111} \ea \eb+4 \cc_{22} \dc_{2|1111} \ea \eb+24 \cc_{3111} \dc_{2|2}+4 \dc_{02|1111} \dc_{2|02} \ea \eb+12 \dc_{2|1111}   \dc_{2|2})+\dc_{2|1111} (3 \eta +3)&=\beta( \dc_{2|1111} ) \end{align*}     
 \newpage  
 \phantom{a}
 \hspace{2cm}
  \begin{align*} 2 h_{2} (6 \ac_{4} (3 \ac_{6}+\dc_{2|4})+72 \ac_{6} \dc_{2|2}-\cc_{1111} \cc_{3111} \ea \eb+2 \cc_{2121} \cc_{22} \ea \eb+\cc_{22} \cc_{42} \ea \eb+2 \cc_{22} \dc_{02|4} \ea \eb+\cc_{22} \dc_{2|22}   \ea \eb& \\
 +4 \cc_{42} \dc_{2|02} \ea \eb+4 \dc_{02|4} \dc_{2|02} \ea \eb+12 \dc_{2|2} \dc_{2|4})+\dc_{2|4} (3 \eta +3)&=\beta( \dc_{2|4} ) \\[3ex]
 2 h_{2} (\ac_{4} (\cc_{24}+2 \dc_{2|04})+\bc_{4} (\cc_{24}+4 \dc_{2|04})+6 \bc_{6} \cc_{22} \ea \eb+24 \bc_{6} \dc_{2|02} \ea \eb+2 \cc_{22} \cc_{24} \ea \eb+\cc_{22} \cc_{42} \ea \eb+2 \cc_{22}   \dc_{02|04} \ea \eb+& \\
 \cc_{22} \dc_{2|22} \ea \eb+12 \cc_{24} \dc_{2|2}+4 \dc_{02|04} \dc_{2|02} \ea \eb+12 \dc_{2|04} \dc_{2|2})+\dc_{2|04} (3 \eta +3)&=\beta( \dc_{2|04} ) \\[3ex]
 2 h_{2} (\ac_{4} (2 \cc_{1212}+\dc_{02|22})+\bc_{4} (2 \cc_{1212}+6 \cc_{24}+3 \dc_{02|22})+2 (3 \bc_{6} \cc_{22} \ea \eb-\cc_{1111} \cc_{1311} \ea \eb-\cc_{1111} \cc_{3111} \ea \eb-2 \cc_{1111} \dc_{02|1111}   \ea \eb& \\
 +2 \cc_{1212} \cc_{22} \ea \eb+12 \cc_{1212} \dc_{02|02}+\cc_{2121} \cc_{22} \ea \eb+4 \cc_{2121} \dc_{2|02} \ea \eb+2 \cc_{22} \cc_{24} \ea \eb+2 \cc_{22} \cc_{42} \ea \eb+2   \cc_{22} \dc_{02|04} \ea \eb+\cc_{22} \dc_{02|22} \ea \eb& \\
 +2 \cc_{22} \dc_{02|4} \ea \eb+\cc_{22} \dc_{2|22} \ea \eb+18 \cc_{24} \dc_{02|02}+6 \cc_{42} \dc_{2|02} \ea \eb+6 \dc_{02|02}   \dc_{02|22}+2 \dc_{2|02} \dc_{2|22} \ea \eb))+\dc_{02|22} (3 \eta +3)&=\beta( \dc_{02|22} ) \\[3ex]
 2 h_{2} (3 \bc_{4} \cc_{1311}+2 \bc_{4} \dc_{02|1111}-2 \cc_{1111} \cc_{2121} \ea \eb-2 \cc_{1111} \cc_{24} \ea \eb-2 \cc_{1111} \dc_{02|22} \ea \eb+4 \cc_{1311} \cc_{22} \ea \eb+24 \cc_{1311}   \dc_{02|02}& \\
 +2 \cc_{22} \cc_{3111} \ea \eb+4 \cc_{22} \dc_{02|1111} \ea \eb+2 \cc_{22} \dc_{2|1111} \ea \eb+8 \cc_{3111} \dc_{2|02} \ea \eb+12 \dc_{02|02} \dc_{02|1111}+4 \dc_{2|02} \dc_{2|1111} \ea   \eb)+\dc_{02|1111} (3 \eta +3)&=\beta( \dc_{02|1111} ) \\[3ex]
 2 h_{2} (6 \bc_{4} (3 \bc_{6}+\dc_{02|04})+72 \bc_{6} \dc_{02|02}-\cc_{1111} \cc_{1311} \ea \eb+2 \cc_{1212} \cc_{22} \ea \eb+\cc_{22} \cc_{24} \ea \eb+\cc_{22} \dc_{02|22} \ea \eb+2 \cc_{22} \dc_{2|04}   \ea \eb& \\
 +4 \cc_{24} \dc_{2|02} \ea \eb+12 \dc_{02|02} \dc_{02|04}+4 \dc_{2|02} \dc_{2|04} \ea \eb)+\dc_{02|04} (3 \eta +3)&=\beta( \dc_{02|04} ) \hspace{1cm} \\[3ex]
 2 h_{2} (\ac_{4} (\cc_{42}+4 \dc_{02|4})+6 \ac_{6} \cc_{22} \ea \eb+24 \ac_{6} \dc_{2|02} \ea \eb+\bc_{4} (\cc_{42}+2 \dc_{02|4})+\cc_{22} \cc_{24} \ea \eb+2 \cc_{22} \cc_{42} \ea \eb+\cc_{22}   \dc_{02|22} \ea \eb& \\
 +2 \cc_{22} \dc_{2|4} \ea \eb+12 \cc_{42} \dc_{02|02}+12 \dc_{02|02} \dc_{02|4}+4 \dc_{2|02} \dc_{2|4} \ea \eb)+\dc_{02|4} (3 \eta +3)&=\beta( \dc_{02|4} )
\end{align*} 

\vspace{1cm} 
\color{gray}

\hspace{4cm}(The blank space has been sacrificed in order to get \\
\hspace{4cm}beta-functions for dual couplings on the same page.)
\color{black}
\newpage 

 \begin{align*} 
 2 h_{2} (\ac_{4} (9 \ac_{6}+6 \dc_{3|3})+\ea \eb (\cc_{22} (\cc_{3111}+\dc_{12|3})-\cc_{1111} \cc_{2121}))+\dc_{3|3} (3 \eta +3)&=\beta( \dc_{3|3} ) \\[3ex]
 2 h_{2} (\ac_{4} (\cc_{3111}+3 (\cc_{42}+\dc_{12|3}))+\ea \eb (6 \ac_{6} \cc_{22}-\cc_{1111} \cc_{24}& \\-2 \cc_{1111} \cc_{42}-2 \cc_{1111} \dc_{12|3}+2 \cc_{1212} \cc_{22}+2 \cc_{1311} \cc_{22}+2 \cc_{22} \cc_{3111}& \\
 +4 \cc_{22}   \cc_{42}+2 \cc_{22} \dc_{12|12}+2 \cc_{22} \dc_{12|3}+6 \cc_{22} \dc_{3|3})+\bc_{4} (\cc_{3111}+\dc_{12|3}))+\dc_{12|3} (3 \eta +3)&=\beta( \dc_{12|3} ) \\[3ex]
 2 h_{2} (\ac_{4} (4 \cc_{2121}+3 \cc_{3111}+2 \dc_{21|21})+\ea \eb (-3 \ac_{6} \cc_{1111}-\cc_{1111} (2 \cc_{1212}+\cc_{1311}& \\+2 (2 \cc_{2121}+\cc_{3111}+\cc_{42}+2 \dc_{21|21}))+\cc_{22} (\cc_{1311}+4 \cc_{2121}+2 \cc_{24}& \\
 +4   \cc_{3111}+4 \cc_{42}+3 \dc_{21|03}+4 \dc_{21|21}))+\bc_{4} \cc_{2121})+\dc_{21|21} (3 \eta +3)&=\beta( \dc_{21|21} ) \\[3ex]
 2 h_{2} (\bc_{4} (9 \bc_{6}+6 \dc_{03|03})+\ea \eb (\cc_{22} (\cc_{1311}+\dc_{21|03})-\cc_{1111} \cc_{1212}))+\dc_{03|03} (3 \eta +3)&=\beta( \dc_{03|03} ) \hspace{3cm} \\[3ex]
 2 h_{2} (\ac_{4} (\cc_{1311}+\dc_{21|03})+\bc_{4} (\cc_{1311}+3 (\cc_{24}+\dc_{21|03}))+\ea \eb (6 \bc_{6} \cc_{22}& \\-2 \cc_{1111} \cc_{24}-\cc_{1111} \cc_{42}-2 \cc_{1111} \dc_{21|03}+2 \cc_{1311} \cc_{22}+2 \cc_{2121} \cc_{22}& \\
 +4   \cc_{22} \cc_{24}+2 \cc_{22} \cc_{3111}+6 \cc_{22} \dc_{03|03}+2 \cc_{22} \dc_{21|03}+2 \cc_{22} \dc_{21|21}))+\dc_{21|03} (3 \eta +3)&=\beta( \dc_{21|03} ) \\[3ex]
 2 h_{2} (\ac_{4} \cc_{1212}+\bc_{4} (4 \cc_{1212}+3 \cc_{1311}+2 \dc_{12|12})+\ea \eb (-3 \bc_{6} \cc_{1111}-\cc_{1111} (4 \cc_{1212}& \\+2 \cc_{1311}+2 \cc_{2121}+2 \cc_{24}+\cc_{3111}+4 \dc_{12|12})+\cc_{22} (4 \cc_{1212}+4   \cc_{1311}& \\
+4 \cc_{24}+\cc_{3111}+2 \cc_{42}+4 \dc_{12|12}+3 \dc_{12|3})))+\dc_{12|12} (3 \eta +3)&=\beta( \dc_{12|12} ) \end{align*}
\vspace{1cm}

\color{gray}

(The blank space has been sacrificed in order to get \\
beta-functions for dual couplings on the same page.)
\color{black}
  \restoregeometry
 } 
 
\end{landscape}
 
\section{Intermediate steps in the proof of Theorem 7.2}\label{sec:IntermediateSteps} 
In order to compute the quantum fluctuations reported in Theorem 7.2,
one computed the $FP\inv$ expansion. Computing to first oder and taking its supertrace one gets
the expression below; from a selected operator, its the large-$N$ limit is taken 
according to the procedures and scalings described in the main text 
(Tables \ref{tab:QQops} and \ref{tab:SexticOps}).
\allowdisplaybreaks[4]
 \[\begin{array}{l} 
 \phantom+ 
  \TrN^2(A) \times (\ea \acb_{4}-\ea \ccb_{1111}+24 \ea \dcb_{2|2}+2 \eb \dcb_{11|11}+2 N \dcb_{1|12}+6 N
    \dcb_{1|3}) \\[8pt] + 
  \TrN^2(B) \times (2 \ea \dcb_{11|11}+\eb \bcb_{4}-\eb \ccb_{1111}+24 \eb \dcb_{02|02}+6 N \dcb_{01|03}+2 N
    \dcb_{01|21}) 
    \\[8pt]
   + \TrN(A\cdot A) \times \big(2 \ea \eb N \dcb_{01|21}+4 N^2 \ea  \dcb_{2|02} + 12 N^2 \ea  \dcb_{2|2} \\[8pt]\hspace{100pt} +  2 \ea N
    \acb_{4}+2 \ea N \ccb_{22}+6 N \dcb_{1|3}\big) 
    \\[8pt] 
  + \TrN(B\cdot B) \times \big(2 \ea \eb N \dcb_{1|12}+12 N^2 \eb  \dcb_{02|02}+4 N^2 \eb  \dcb_{2|02} 
  \\[8pt]\hspace{100pt} +  2 \eb N
    \bcb_{4}+2 \eb N \ccb_{22}+6 N \dcb_{01|03}\big) 
    \\[8pt]  
    + 
  \TrN(A\cdot A\cdot A\cdot A) \times \big(2 N^2 \ea  \dcb_{2|4}+12 \ea N \acb_{6}+10 \ea N \dcb_{1|5}
 \\[8pt] \hspace{100pt} +2 N^2 \eb 
    \dcb_{02|4}+2 \eb N \ccb_{42}+2 \eb N \dcb_{01|41}\big) 
      \\    [8pt] 
   +  \TrN(B\cdot B\cdot B\cdot B)  \times\big(2 N^2 \ea  \dcb_{2|04}+2 \ea N \ccb_{24}+2 \ea N \dcb_{1|14}
 \\[8pt] \hspace{100pt} +2 N^2 \eb 
    \dcb_{02|04}+12 \eb N \bcb_{6}+10 \eb N \dcb_{01|05}\big) \\[8pt] + 
  \TrN(A\cdot A\cdot B\cdot B) \times \big(2 N^2 \ea  \dcb_{2|22}+2 \ea N \ccb_{42}+2 \ea N \dcb_{1|32}
 \\[8pt] \hspace{100pt} +2 N^2 \eb 
    \dcb_{02|22}+2 \eb N \ccb_{24}+2 \eb N \dcb_{01|23}\big) 
    \\[8pt] 
+  \TrN(A\cdot B\cdot A\cdot B)  \times\big(2 N^2 \ea  \dcb_{2|1111}+2 \ea N \ccb_{3111}+2 \ea N \dcb_{1|2111}
 \\[8pt] \hspace{100pt} +2 \eb
    N^2 \dcb_{02|1111}+2 \eb N \ccb_{1311}+2 \eb N \dcb_{01|1211}\big)
    \\[8pt] 
    + 
  \TrN(A\cdot A) \TrN(B\cdot B) \times (8 \ea N \dcb_{02|4}+2 \ea N \dcb_{2|22}+2 \ea \ccb_{42}+6 \ea
    \dcb_{12|3} \\[8pt]\hspace{100pt}+2 \eb N \dcb_{02|22}+8 \eb N \dcb_{2|04}+2 \eb \ccb_{24}+6 \eb \dcb_{21|03}) 
    \\[8pt] 
    + 
  \TrN(A) \TrN(A\cdot A\cdot A) \times (10 \ea N \dcb_{1|5}+12 \ea N \dcb_{3|3}+12 \ea \acb_{6}+16 \ea
    \dcb_{2|4} \\[8pt]\hspace{100pt}+2 \eb N \dcb_{12|3}+2 \eb N \dcb_{1|32}+2 \eb \ccb_{3111}+2 \eb \dcb_{11|31}) \\[8pt] + 
  \TrN(A) \TrN(A\cdot B\cdot B) \times (6 \ea N \dcb_{12|3}+2 \ea N \dcb_{1|32}+2 \ea \ccb_{42}+4 \ea
    \dcb_{2|22} \\[8pt]\hspace{100pt}+4 \eb N \dcb_{12|12}+2 \eb N \dcb_{1|14}+2 \eb \ccb_{1311}+2 \eb \dcb_{11|13}) 
     \\[8pt] 
    + 
  \TrN(B) \TrN(A\cdot A\cdot B) \times (2 \ea N \dcb_{01|41}+4 \ea N \dcb_{21|21}+2 \ea \ccb_{3111}+2 \ea
    \dcb_{11|31} \\[8pt]\hspace{100pt}+2 \eb N \dcb_{01|23}+6 \eb N \dcb_{21|03}+2 \eb \ccb_{24}+4 \eb \dcb_{02|22}) 
    \\[8pt] 
      + 
  \TrN(B) \TrN(B\cdot B\cdot B) \times (2 \ea N \dcb_{01|23}+2 \ea N \dcb_{21|03}+2 \ea \ccb_{1311}+2 \ea
    \dcb_{11|13} \\[8pt]\hspace{100pt}+10 \eb N \dcb_{01|05}+12 \eb N \dcb_{03|03}+12 \eb \bcb_{6}+16 \eb \dcb_{02|04}) 
       \end{array}\] 
   \vspace{.2cm}
   \[\begin{array}{l} 
    + 
  \TrN^2(A\cdot B) \times  (2 \ea N \dcb_{11|31}+2 \ea \ccb_{2121}+2 \ea \dcb_{21|21}\\[8pt]\hspace{100pt}+2 \eb N \dcb_{11|13}+2
    \eb \ccb_{1212}+2 \eb \dcb_{12|12}) \\[8pt] + 
    \TrN^2(A\cdot A) \times (8 \ea N \dcb_{2|4}+6 \ea \acb_{6}+18 \ea \dcb_{3|3}\\[8pt]\hspace{100pt}+2 \eb N \dcb_{2|22}+2 \eb
    \ccb_{2121}+2 \eb \dcb_{21|21}) \\[8pt] + 
  \TrN^2(B\cdot B) \times (2 \ea N \dcb_{02|22}+2 \ea \ccb_{1212}+2 \ea \dcb_{12|12}\\[8pt]\hspace{100pt}+8 \eb N \dcb_{02|04}+6
    \eb \bcb_{6}+18 \eb \dcb_{03|03}) \,.
 \end{array}\]
\allowdisplaybreaks[0]%

\newcommand{\etalchar}[1]{$^{#1}$}

\end{document}